\documentclass[times,sort&compress,3p]{elsarticle}
\journal{Journal of Multivariate Analysis}
\usepackage[labelfont=bf]{caption}

\usepackage{subfigure,amsmath,amsfonts,amssymb,amsthm,booktabs,color,epsfig,graphicx,hyperref,url,float}
\theoremstyle{plain}
\newtheorem{theorem}{Theorem}

\newtheorem{lemma}{Lemma}
\newtheorem{corollary}{Corollary}
\newtheorem{assumption}{Assumption}
\theoremstyle{definition}

\begin{document}

\begin{frontmatter}

\title{A new preferential model with homophily for recommender systems}

\author[1]{Hanyang Tian}
\author[2]{ Bo Zhang\corref{mycorrespondingauthor}}
\author[3]{Ruixue Jiang}

\address{ International Institute of Finance, School of Management,\\ University of Science and Technology of China}

\cortext[mycorrespondingauthor]{Corresponding author. Email address: \url{wbchpmp@ustc.edu.cn}}



\begin{abstract}
``Rich-get-richer'' and ``homophily'' are two important phenomena in evolving social networks. ``Rich-get-richer'' means people with higher followings are more likely to attract new fans,
and ``homophily'' means
people prefer to bond with others of the same social group or who have some other attribute in common.
To formalize the phenomena simultaneously in the context of an evolving social network, we consider a K-groups preferential attachment (KPA) network model, which is helpful for the social network’s recommender system.

The main contribution of this paper is to propose a new evolving social network model with the mechanisms of rich-get-richer and homophily. We show that the KPA model exhibits a power-law degree distribution for each group and prove the central limit theorem (CLT) for the maximum likelihood estimation (MLE) of the parameters in the KPA model. We illustrate our results through simulated data and explore the usage of this model with real data examples.
\end{abstract}


\end{frontmatter}
\section{Introduction}
The ``rich-get-richer" (or preferential attachment) mechanism has been a focus of network researchers in recent decades. New nodes and edges are constantly added to the evolving network in preferential attachment models. The new edges are attached to older nodes chosen
according to a probability distribution which is an affine function of the older nodes' degree. This way, nodes with a high degree are more likely to attract edges' connection and then have a higher degree, which explains why such models are called rich-get-richer models.
A basic introduction to the preferential attachment model can be found in \cite{easley2010networks} and \cite{barabasi1999emergence}. To further understand the statistical properties of the preferential attachment model,  \cite{van2017random}, \cite{chung2006complex}, and \cite{durrett2007random} show the limit theory and asymptotic characteristics.

``Homophily" mechanism has a profound effect on individuals' and groups' behavior in society (\cite{lazarsfeld1954friendship}). For researchers, homophily is a well-documented phenomenon in social networks (\cite{mcpherson2001birds}): people often prefer to connect with others with similar characteristics. For example, people are more likely to build social relations, like marriage, friendship, and colleagues,
with someone of the same age, education, or hobbies. In other words, homophily influence the connection structure in human society. Furthermore, many studies show that homophily can influence not only  society's static structures but also society's dynamic operations. (\cite{jackson2010social}), (\cite{jackson2013diffusion}) show the effect of homophily on the welfare of individuals and diffusion patterns of information in social networks.

On social network platforms (e.g., Twitter, TikTok, and Sina Weibo), rich-get-richer and homophily often co-occur.
This study supposes that it takes two steps for a person to connect with another in a social network: (1) Become aware of someone through a friend referral or a social media feed. (2) Decide whether to follow or connect with that person.
People easily become aware of an Internet celebrity (the first step), which implies that the rich-get-richer phenomenon affects the celebrity’s set of followers.
Considering the second step, a person prefers to follow somebody with the same hobby, which means homophily is also involved.
Moreover, homophily can work contrary to the influence of rich-get-richer.
For example, Lebron James is a basketball superstar with a huge following. His popularity makes it easy for him to get more followers, but someone who has no interest in basketball will not follow him even if he is recommended by friends.
Thus, it is meaningful to study evolving networks while considering interactions of rich-get-richer and homophily.
Unfortunately, previous studies typically consider rich-get-richer and homophily separately.
Some recent papers (\cite{lee2019homophily},\cite{avin2020mixed}) have tried to combine the ideas, but they do not focus on statistical problems such as estimators and central limit theorems.

%

In this paper, we propose the KPA (K-groups preferential attachment) model based on the  Barab$\acute{a}$si--Albert model (\cite{barabasi1999emergence} and \cite{albert2002statistical}). The unit time point of the dynamic process is  divided into two parts:
(1) [rich-get-richer] The evolving network tries to connect a chosen old node to a new node. The higher the degree of the old node, the higher the probability that it will be chosen to connect to the new node.
(2) [homophily] The new node will accept the network’s recommended connection with a probability dependent on the similarity of the two nodes.

We divide all nodes into $K$ groups according to a specific feature. Homophily states that nodes in the same group are more easily connected. A parameter $\theta$ is introduced to the classic Barab$\acute{a}$si--Albert model to exhibit the influence of homophily on the generation of evolving networks.
Using the KPA model, we obtain some theoretical results about degrees. Then we propose the estimators of the homophily and some other parameters in an evolving network featuring both rich-get-richer and homophily. We also give the joint asymptotic distribution of these estimators.  It is commonly acknowledged that recommender systems play a vital role in the big data era (see \cite{jannach2010recommender} and \cite{ricci2015recommender}).
Accurate estimation of the effects of homophily is helpful to improve the recommender system of any social network platform.
If the homophily is strong, recommending a connection with a node in a very different group (i.e., one with dissimilar nodes) is inefficient.
In contrast, when the homophily is not strong, recommending a node from different groups is meaningful.

The paper is organized as follows. In Section 2, we introduce the specific construction process of the KPA model and the significance of various random variables. The main asymptotic results are in Section 3. The estimation of parameters is given in Section 4. Section 5 focuses on the change point. Section 6 contains proofs of the important results. Simulations to illustrate the theoretical results and applications to real-life data are in the supplementary material.

\section{Model}

According to the classic Barab$\acute{a}$si--Albert model, an initial graph $G_0$ has an isolated node of degree one. Graph $G_t$ represent the state of the evolving network at time $t, \, t\in\{0,1,\cdots,n,\cdots\}$.
There are two operations on the evolving network:

\begin{itemize}
\item Vertex-step: at time $t$, a new node $w$ is added to the network and connects node $u$ with the edge $(w,u)$.
\item Edge-step: at time $t$, no new node arrives, but a new edge $(w,u)$ is added to the network. Nodes $w,u$  are pre-existing in the network before time $t$.
\end{itemize}

Begin with the initial graph $G_0$. For $t>0$, $G_t$ is the graph at time $t$ by modifying $G_{t-1}$ by either taking a vertex-step or
an edge-step.
Let $V_t$ be the number of nodes in graph $G_t$. Furthermore, let $v_t:=V_t-V_{t-1}$. We can find that $G_t$ is formed from $G_{t-1}$ by taking a vertex-step when $v_t= 1$ or an edge-step when $v_t=0$.

Consider nodes with similar characteristics to belong to the same group.
There are $K$ groups of nodes in the network, where $K$ is a fixed known constant.  For node $i$ from group $k$, set $g_i=k,\,k\in \{1,\cdots,K\}$. $g_i$ is the group label of node $i$.

We assume the following assumptions:

\begin{assumption}\label{a1}
$g_i$ is known (or observed) for each node $i$.
\end{assumption}

\begin{assumption}\label{a2}
The new node comes from group $k$ with unknown probability $p_k$ at vertex-step, $ k= 1,\cdots,K$. Where $\sum_{k=1}^K p_k=1, p_k\in [0,1]$.
\end{assumption}

\begin{assumption}\label{a3}
In the initial graph $G_0$, the number of nodes from group $k$ is $p_kn_0$. $n_0$ is a constant large enough for $p_kn_0$ to be an integer, $ k= 1,\cdots,K$.
\end{assumption}

\begin{assumption}\label{a4}
$\{v_t\}_{t=1}$ are the i.i.d. random variables of Bernoulli distribution $B(1,q)$. $q$ is unknown.
\end{assumption}

\begin{figure}[H]
\centering
\label{F0}
\subfigure[$v_t=1$]{
\begin{minipage}[t]{0.45\linewidth}
\centering
\includegraphics[width=2.0in]{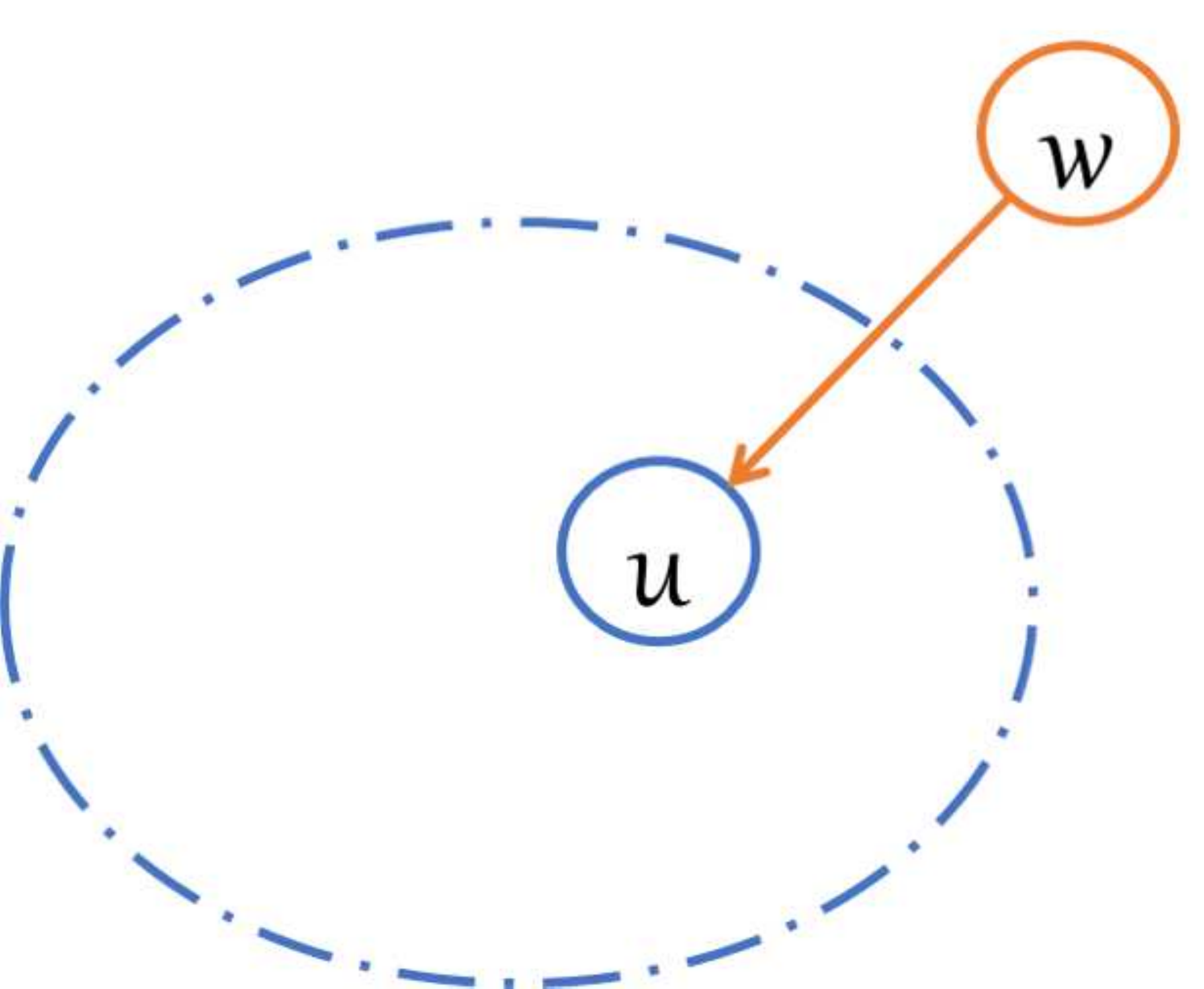}
\end{minipage}%
}%
\subfigure[$v_t=0$]{
\begin{minipage}[t]{0.45\linewidth}
\centering
\includegraphics[width=1.7in]{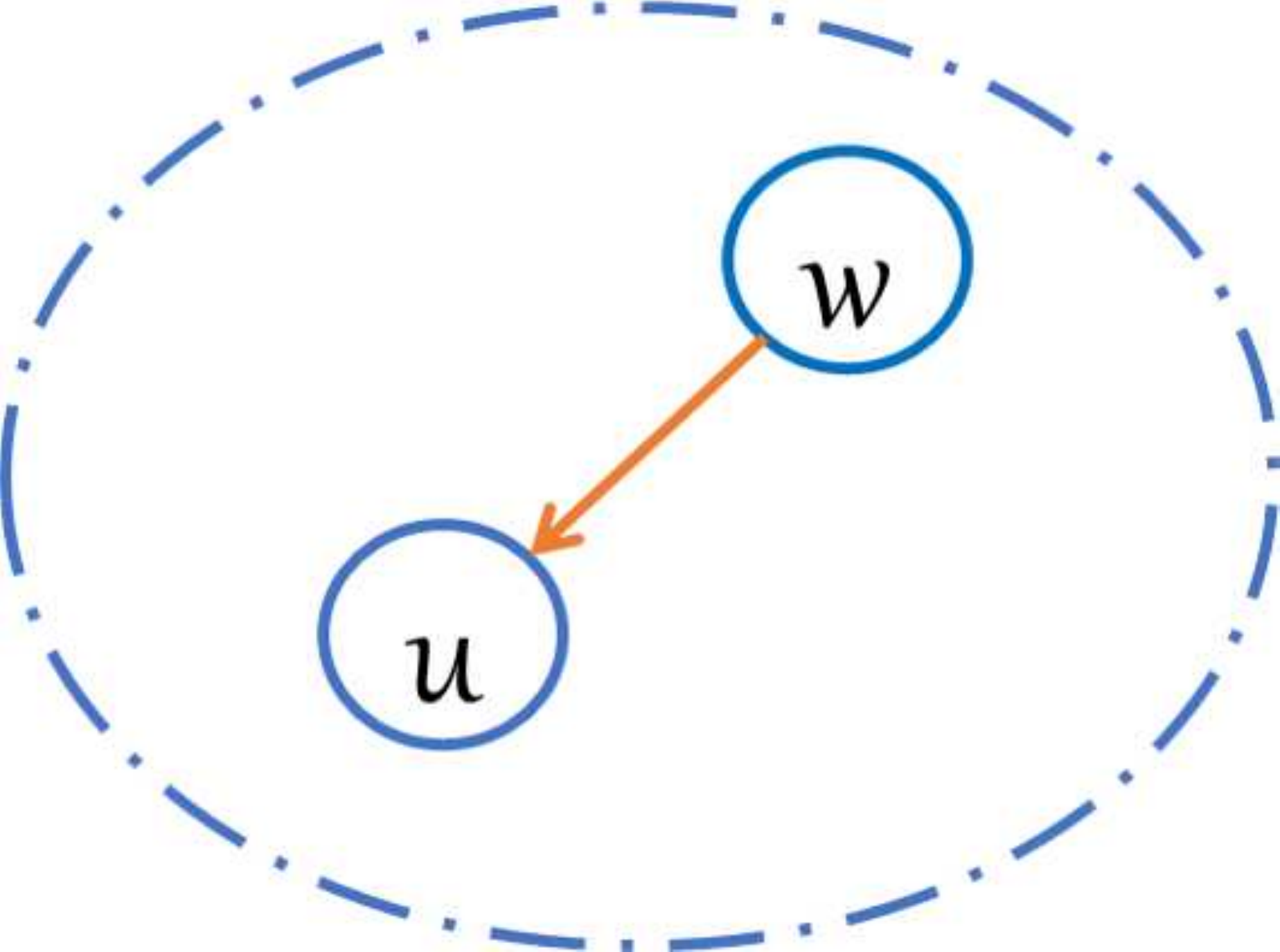}
\end{minipage}%
}%
\centering
\caption{The blue shows the graph $G_{t-1}$ and the orange shows the new node or edge added at time $t$.}
\end{figure}

Next, we describe the process of the KPA model construction in detail.

Beginning with the initial graph $G_0$ with $n_0$ $(n_0\geq K)$ nodes of degree one from $K$ groups, $v_t$ is generated randomly at time $t$. $G_t$ is formed from $G_{t-1}$ by taking a vertex-step when $v_t= 1$ or an edge-step when $v_t=0$. $d_i(t-1)$ is the degree of node $i$ in graph $G_{t-1}$.

\begin{itemize}

\item Vertex-step:
 we break down a vertex-step process into five small steps.
\begin{itemize}
\item[step 1.]
A new node $w$ is added to the network at time $t$.

\item[step 2.]
An old node $u_1$ is chosen to connect $w$ with probability $\frac{d_{u_1}(t-1)}{\sum_{i\in G_{t-1}}d_i(t-1)}$.
If $u_1,w$ are in the same group, $w$ will accept the connection without hesitation.
Otherwise, $w$ will reject the connection with probability $1-\gamma$, $\gamma\in (0,1]$.

\item[step 3.]
If $w$ successfully connects to $u_1$ in step 2, the vertex-step process will break. Else, the network will choose another node to connect with $w$ in step 4 after $w$ rejects $u_1$.

\item[step 4.]
An old node $u_2$ is chosen with probability $\alpha \times \frac{d_{u_2}(t-1)}{\sum_{g_i=g_w}d_i(t-1)},\,\alpha\in (0,1]$ if $u_2,w$ are in the same group.
Otherwise, $u_2$ is chosen with probability $(1-\alpha)\times \frac{d_{u_2}(t-1)}{\sum_{g_i\neq g_w}d_i(t-1)}$.
Furthermore, if $u_2,w$ are in the same group, $w$ will definitely connect with node $u_2$. Else, $w$ will reject $u_2$ with probability $1-\gamma$.

\item[step 5.]
If $w$ successfully connects to $u_2$ in step 4, the vertex-step process will break. Else, the network will choose another node to connect with $w$ and go back to step 4.
\end{itemize}

\item Edge-step: same to the vertex-step except step 1.

\begin{itemize}
\item[step 1.]
No new node arrives at time $t$. Randomly selected an old node $w$ with probability $\frac{d_{w}(t-1)}{\sum_{i\in G_{t-1}}d_i(t-1)}$.
\end{itemize}
\end{itemize}

We suppose that people do not refuse to make friends with people who share their interests, so nodes must accept connections from the same group. Furthermore, we assume that if a person refuses to be friends with a celebrity because their interests do not match, he will pay more attention to people who have the same interests as himself, not just famous people in the future. So if the new node rejects the connection from other groups for the first time, the old nodes in the same group will be chosen with probability $\alpha$ next time, where $\alpha$ is large.
The parameter $\gamma$ indicates the probability that a node accepts connections from other groups, and if $\gamma=1$, there is no difference between the acceptability of connections from the same group and different groups. The parameter $\alpha$ exhibits the tendency of choosing nodes from the same group. Both parameters $\gamma$ and $\alpha$ reflect homophily on the network.

With the above details to the KPA model's generation, we can calculate the conditional probability of the connection at time $t$. For example,
if $v_t=1$, $g_w=g_u$ at time $t$, we get the probability of connection between the new node $w$ and  old node $u$:
\begin{eqnarray}
\nonumber &&P((w,u)|{G}_{t-1}) \\
\nonumber =&&\frac{d_u(t-1)}{\sum_{i}d_i(t-1)}
+\frac{(1-\gamma)\sum_{g_i\neq g_w} d_i(t-1)}{\sum_{i}d_i(t-1)}\frac{\alpha d_u(t-1)}{\sum_{g_i=g_w}d_i(t-1)} \sum_{j=0}^{\infty}\left[(1-\alpha)(1-\gamma)\right]^j\\
=&& \frac{d_u(t-1)}{\sum_{i}d_i(t-1)}
+\frac{\alpha(1-\gamma)}{1-(1-\alpha)(1-\gamma)}\frac{\sum_{g_i\neq g_w} d_i(t-1)}{\sum_{i}d_i(t-1)}\frac{ d_u(t-1)}{\sum_{g_i=g_w}d_i(t-1)} .
\end{eqnarray}

Let $\theta:=\frac{\gamma}{\gamma+\alpha(1-\gamma)}\,\in (0,1]$. $\frac{\alpha(1-\gamma)}{1-(1-\alpha)(1-\gamma)}=1-\theta$. We find that $\theta$ is the parameter that ultimately determines the influence of homophily on the network structure. So how to obtain the information of parameter $\theta$ is the focus of our work.
If $\theta$ is very small, the probability of connections for nodes to different groups is small. Recommender systems need to focus more on recommending users with the same interests as the target users.    Else if $\theta$ is very large, recommender systems only need to recommend popular users. Figure \ref{F1} shows the influence of $\theta$ on the network structure.

\begin{figure}[H]
\centering
\subfigure[$\theta=0.1$]{
\begin{minipage}[t]{0.32\linewidth}
\centering
\includegraphics[width=1.8in]{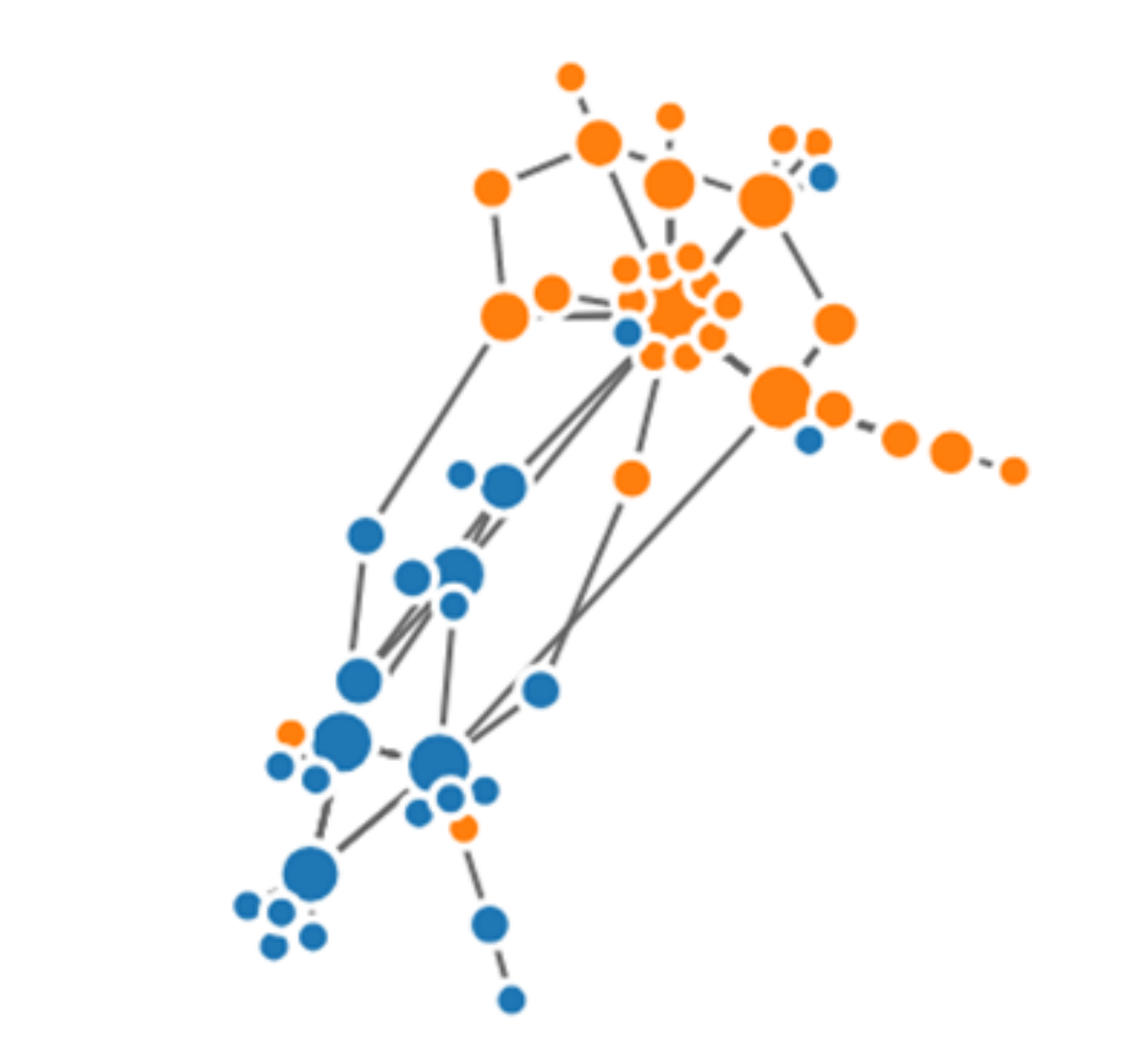}
\end{minipage}%
}%
\subfigure[$\theta=0.5$]{
\begin{minipage}[t]{0.32\linewidth}
\centering
\includegraphics[width=1.8in]{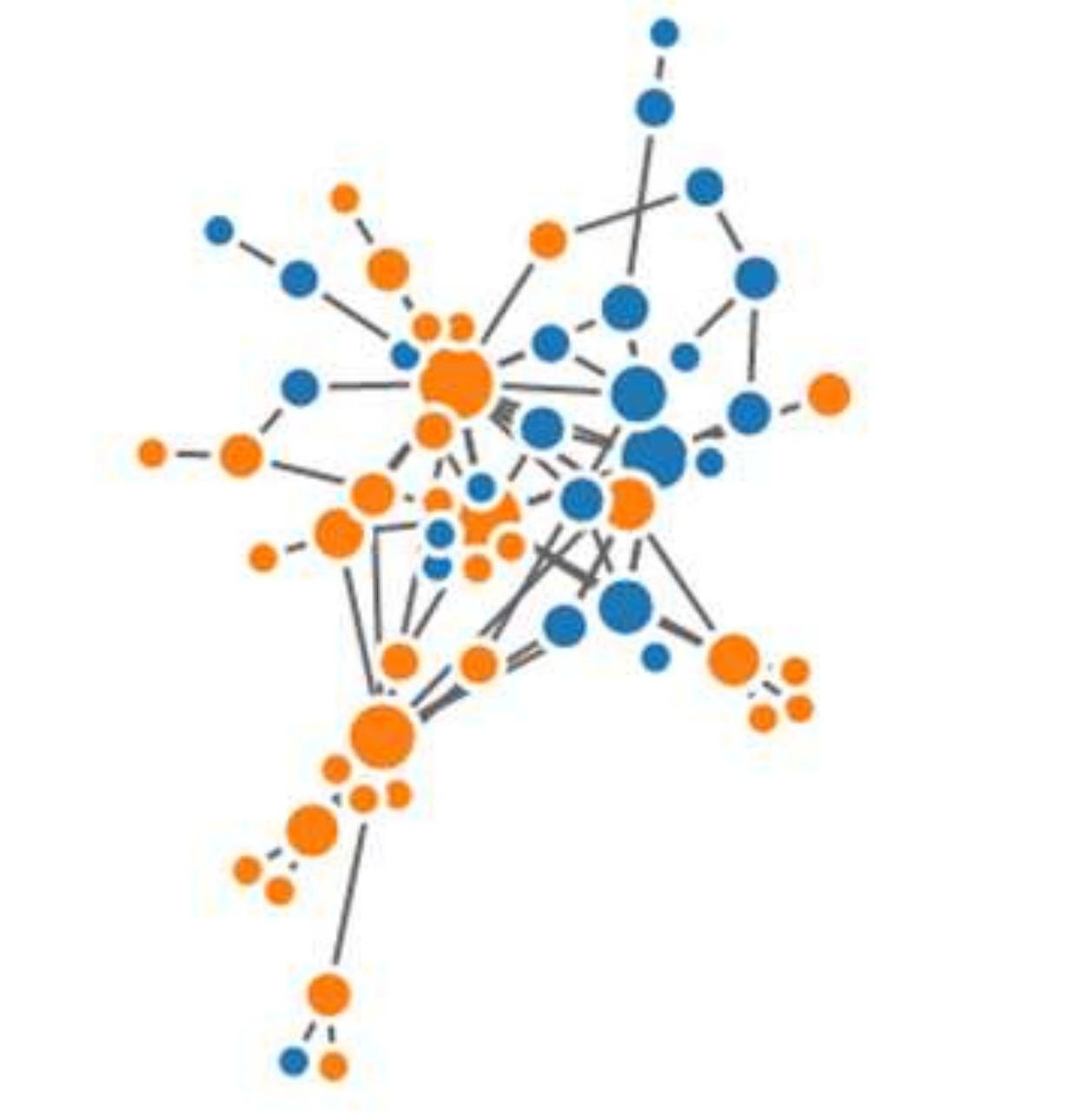}
\end{minipage}%
}%
\subfigure[$\theta=0.9$]{
\begin{minipage}[t]{0.32\linewidth}
\centering
\includegraphics[width=1.8in]{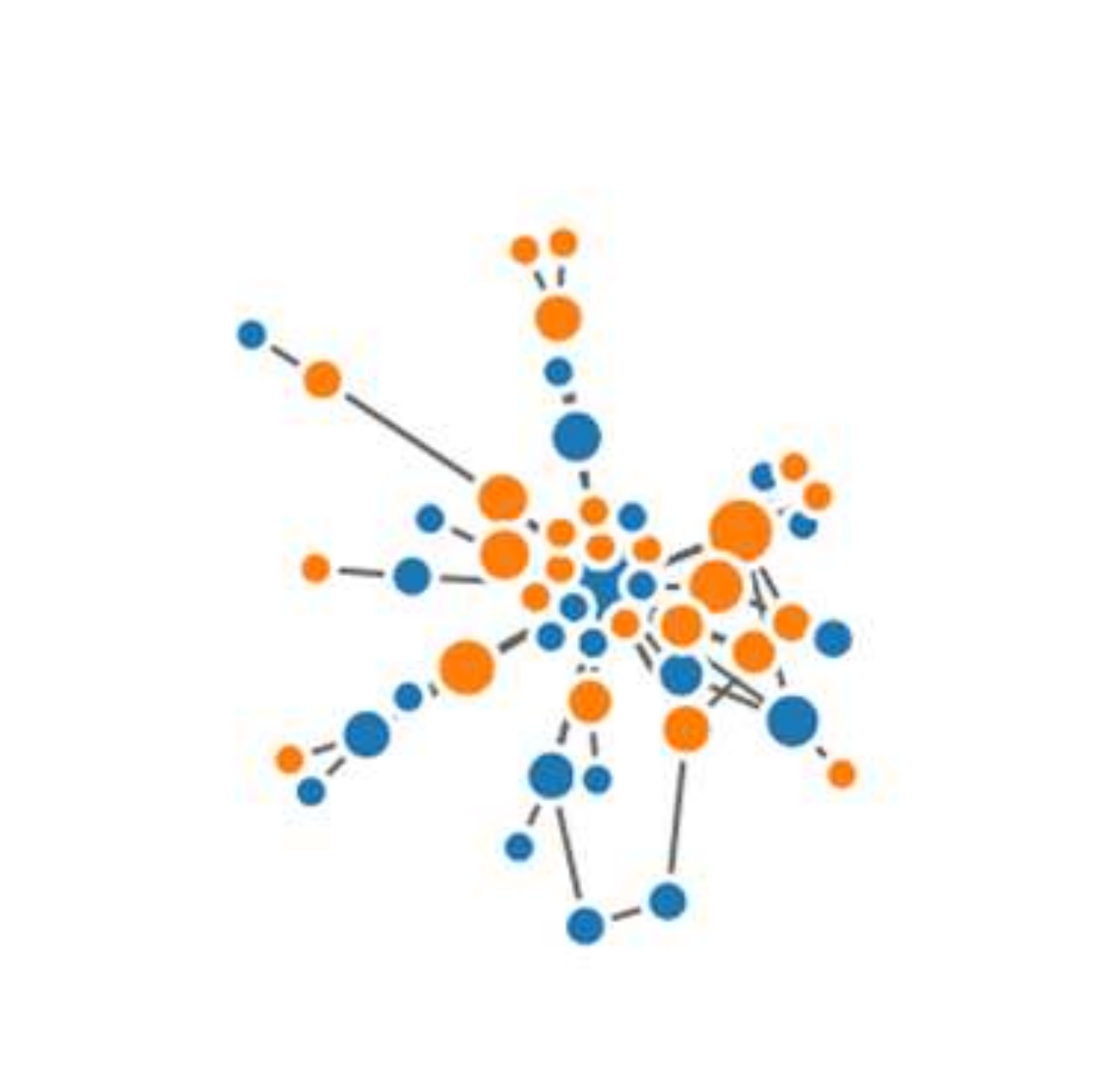}
\end{minipage}%
}%

\centering
\caption{The network generated by the KPA model with the time range $t\in [0,100]$. $K=2,q=0.5,p_1=p_2=0.5$. The different colors represent different groups.}
\label{F1}
\end{figure}

In the following sections, we will substitute $\theta$ for $\gamma$ and $\alpha$. We can get the following lemmas of connection probability with the homophily parameter $\theta$:
\begin{lemma}\label{le1}
For the KPA model, if $v_t=1$ at time $t$, a new node $w$ is added to the network and connects with node $u$ by edge $e_t=(w,u)$ with probability:
\begin{eqnarray}
\nonumber &&P((w,u)|g_w,g_u,{G}_{t-1},v_t=1)\\
=&&\begin{aligned}
\left\{
\begin{array}{lr}
\frac{d_u(t-1)}{\sum_{i}d_i(t-1)}+(1-\theta) \frac{\sum_{g_i\neq g_w} d_i(t-1)}{\sum_{i}d_i(t-1)} \frac{d_u(t-1)}{\sum_{g_i=g_w}d_i(t-1)},      & \text{if } g_w=g_u  \\
\theta\frac{d_u(t-1)}{\sum_{i}d_i(t-1)},      & \text{if } g_w\neq g_u .
\end{array} \right.
\end{aligned}
\end{eqnarray}
\end{lemma}

\begin{lemma}\label{le2}
For the KPA model, if $v_t=0$ at time $t$, a new edge $e_t=(w,u)$ between two existing nodes is added to the network with probability:
\begin{eqnarray}
\nonumber &&P((w,u)|g_w,g_u,{G}_{t-1},v_t=0)\\
=&&\left\{
\begin{array}{lr}
\frac{d_w(t-1)}{\sum_{i}d_i(t-1)}\left[\frac{d_u(t-1)}{\sum_{i}d_i(t-1)}+(1-\theta) \frac{\sum_{g_i\neq g_w} d_i(t-1)}{\sum_{i}d_i(t-1)} \frac{d_u(t-1)}{\sum_{g_i=g_w}d_i(t-1)}\right],   & \text{if }g_w=g_u   \\
\theta\frac{d_w(t-1)}{\sum_{i}d_i(t-1)}\frac{d_u(t-1)}{\sum_{i}d_i(t-1)},      & \text{if }g_w\neq g_u .
\end{array} \right.
\end{eqnarray}
\end{lemma}

\section{Asymptotic results}

\begin{theorem}\label{th1}
Under Assumptions \ref{a1}--\ref{a4}. $d_i(t)$ is the degree of node $i$ in graph $G_t$. Let $D^k_t=\sum_{i}d_i(t)\mathbf{1}_{\{g_i=k\}}$ be the total degrees from group $k$ in $G_t$, $k\in \{1,\cdots,K\}$.
\begin{equation}
 \frac{D^k_t}{2t}\stackrel{a.s.}{\longrightarrow} p_k.
\end{equation}
$\frac{D^k_t}{2t}$ is the ratio of degrees from group $k$,
that is, the sum of the nodes' degrees from group $k$ divided by the total degrees at time $t$.
\end{theorem}

Let random variable $e_t=(w,u)$ be the edge added at time $t$.
Let $X_t=\mathbf{1}_{\{e_t=(w,u),g_w=g_u\}}$ ($\mathbf{1}_{\{\cdot\}}$ is the indicator function) and  $S_t=\sum_{i=1}^t X_i$. $S_t$ is the number of edges where two nodes are from the same group in graph $G_t$.

\begin{corollary}\label{co1}
Under the conditions of Theorem \ref{th1}, we have
\begin{equation}
\frac{S_t}{t}\stackrel{a.s.}{\longrightarrow} 1+\theta\left(\sum_{k=1}^Kp_k^2-1\right).
\end{equation}
\end{corollary}

Theorem \ref{th1} implies a limit of the ratio of degrees from group $k$. However, the limit might be quiet different from what you get at a particular time $t$. We give a probabilistic estimate of the difference by the following theorem.

\begin{theorem}\label{th2}
Under the conditions of Theorem \ref{th1}, for some time point $t$ we have:
\begin{equation}
\begin{cases}
P\left(|D_t^k-p_k(2t+n_0)|\geq{2c(t)\sqrt{t}}\right)\leq C e^{-c(t)^2},  & \text{if }\, \frac{1}{2-\theta}<q\leq 1 \\
P\left(|D_t^k-p_k(2t+n_0)|\geq{2c(t)\sqrt{\log(t)}}\right)\leq C e^{-\frac{c(t)^2}{t}}, & \text{if }\,q=\frac{1}{2-\theta}\\
 P\left(|D_t^k-p_k(2t+n_0)|\geq 2c(t)\right)\leq C e^{-\frac{c(t)^2}{t^{2-q(2-\theta)}}}, & \text{if }\,0<q<\frac{1}{2-\theta},
\end{cases}
\end{equation}
where $c(t)$ is a strictly monotonically increasing function of $t$, and $C$ is a constant greater than 0.
\end{theorem}

The degree distribution obeying the power-law is an attractive property of the classic preferential attachment model. For the KPA model, the nodes are from $K$ groups. We will exhibit the power-law degree distribution for each group.

%

\begin{theorem}\label{th3}
Under the conditions of Theorem \ref{th1}, let $m_d^k(t)$ denote the number of nodes with degree $d$ from group $k$ in graph $G_t$, $k=1,\cdots,K$. Note that $m_1^k(0)=p_kn_0,\,m_0^k(t)=0$.
Letting $M_d^k=\lim_{t\to\infty} \frac{E(m_d^k(t))}{t}$, we have:
\begin{equation}
\begin{aligned}
M_d^k=\frac{2qp_k}{4-q}\prod_{j=2}^d\frac{(j-1)(2-q)}{2+j(2-q)}\propto  d^{-\beta_k},
\end{aligned}
\end{equation}
where $\beta_k=1+\frac{2}{2-q}$.
\end{theorem}

Theorem \ref{th3} implies that different groups' degree have the same power-law distribution.
Figure \ref{figlaw} shows the power-law degree distribution of a simulated network data.
\begin{figure}[H]
\centering

\subfigure[$\theta=0.1,p_1=0.5,q=0.5$]{
\begin{minipage}[t]{0.3\linewidth}
\centering
\includegraphics[width=2in]{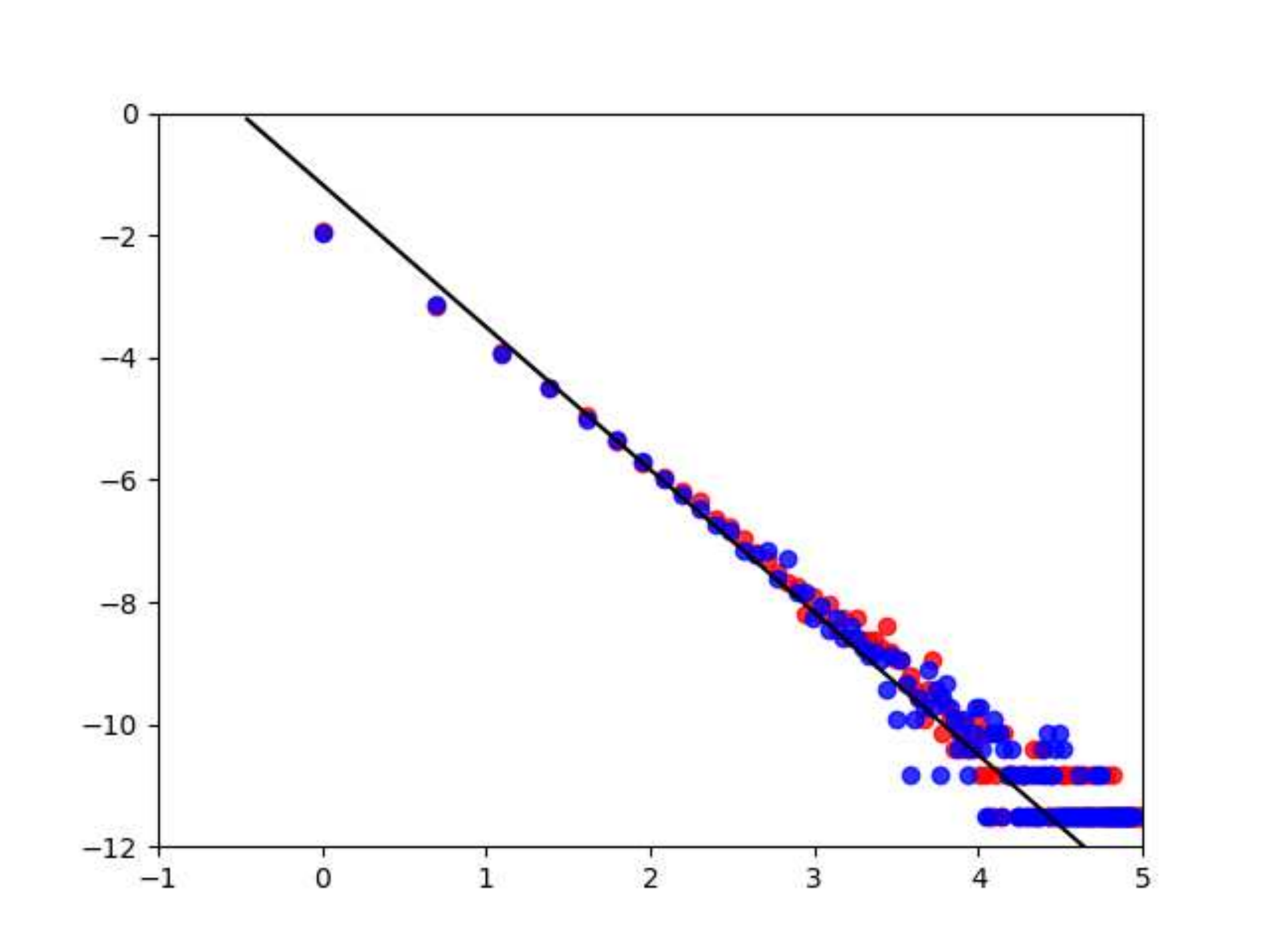}
\end{minipage}%
}%
\subfigure[$\theta=0.5,p_1=0.5,q=0.5$]{
\begin{minipage}[t]{0.3\linewidth}
\centering
\includegraphics[width=2in]{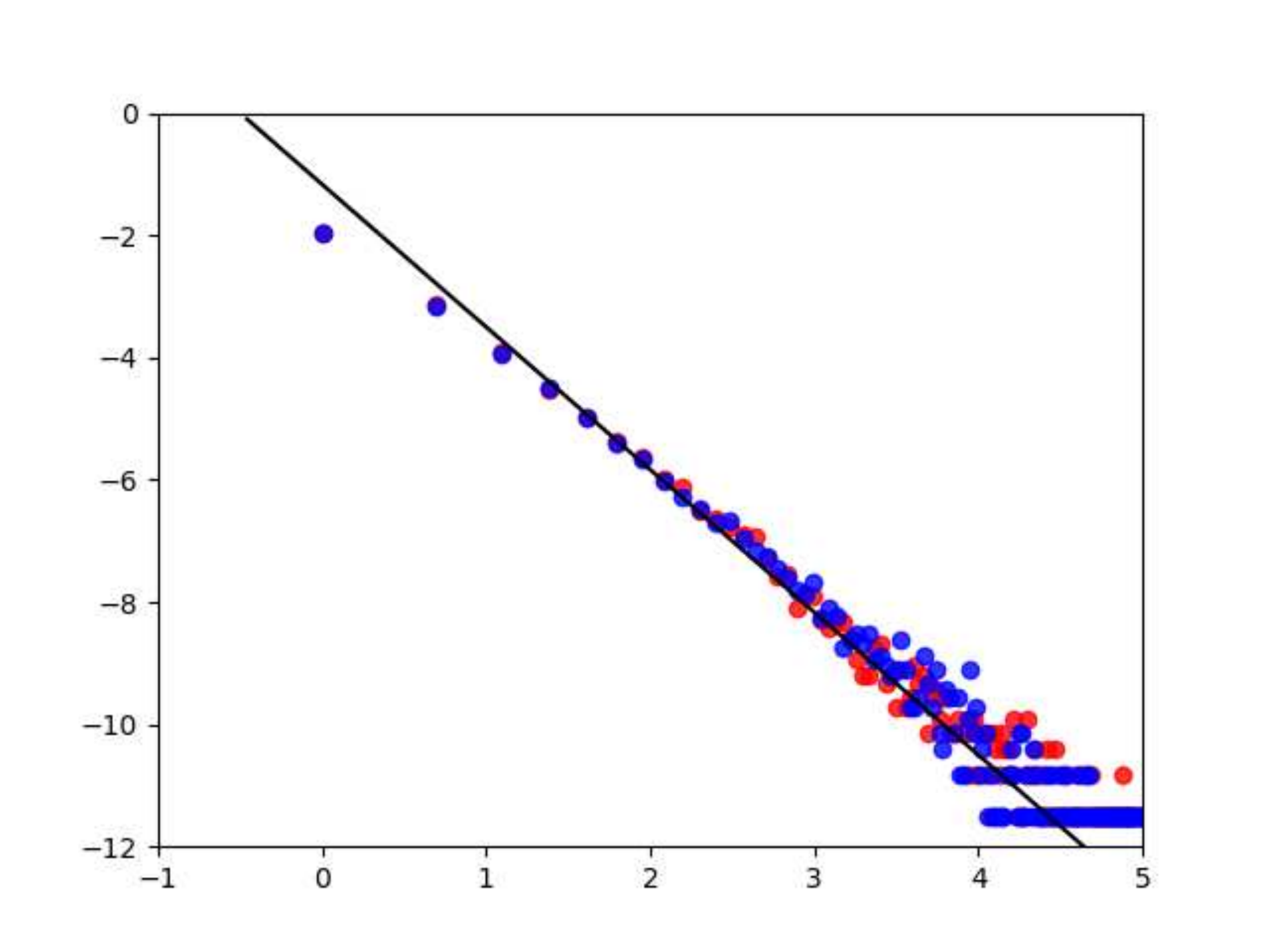}
\end{minipage}%
}%
\subfigure[$\theta=0.9,p_1=0.5,q=0.5$]{
\begin{minipage}[t]{0.3\linewidth}
\centering
\includegraphics[width=2in]{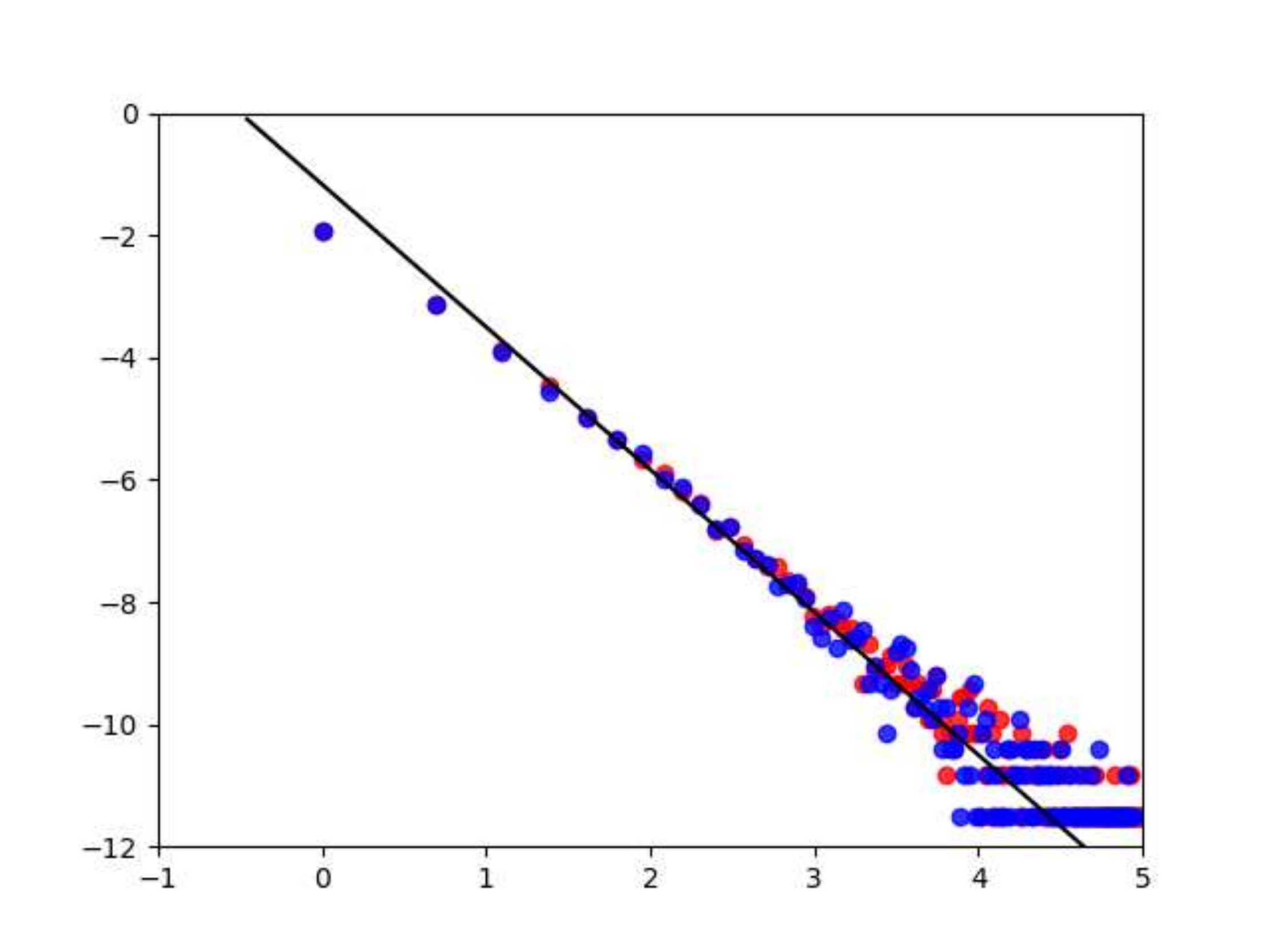}
\end{minipage}%
}%

\subfigure[$\theta=0.5,p_1=0.1,q=0.5$]{
\begin{minipage}[t]{0.3\linewidth}
\centering
\includegraphics[width=2in]{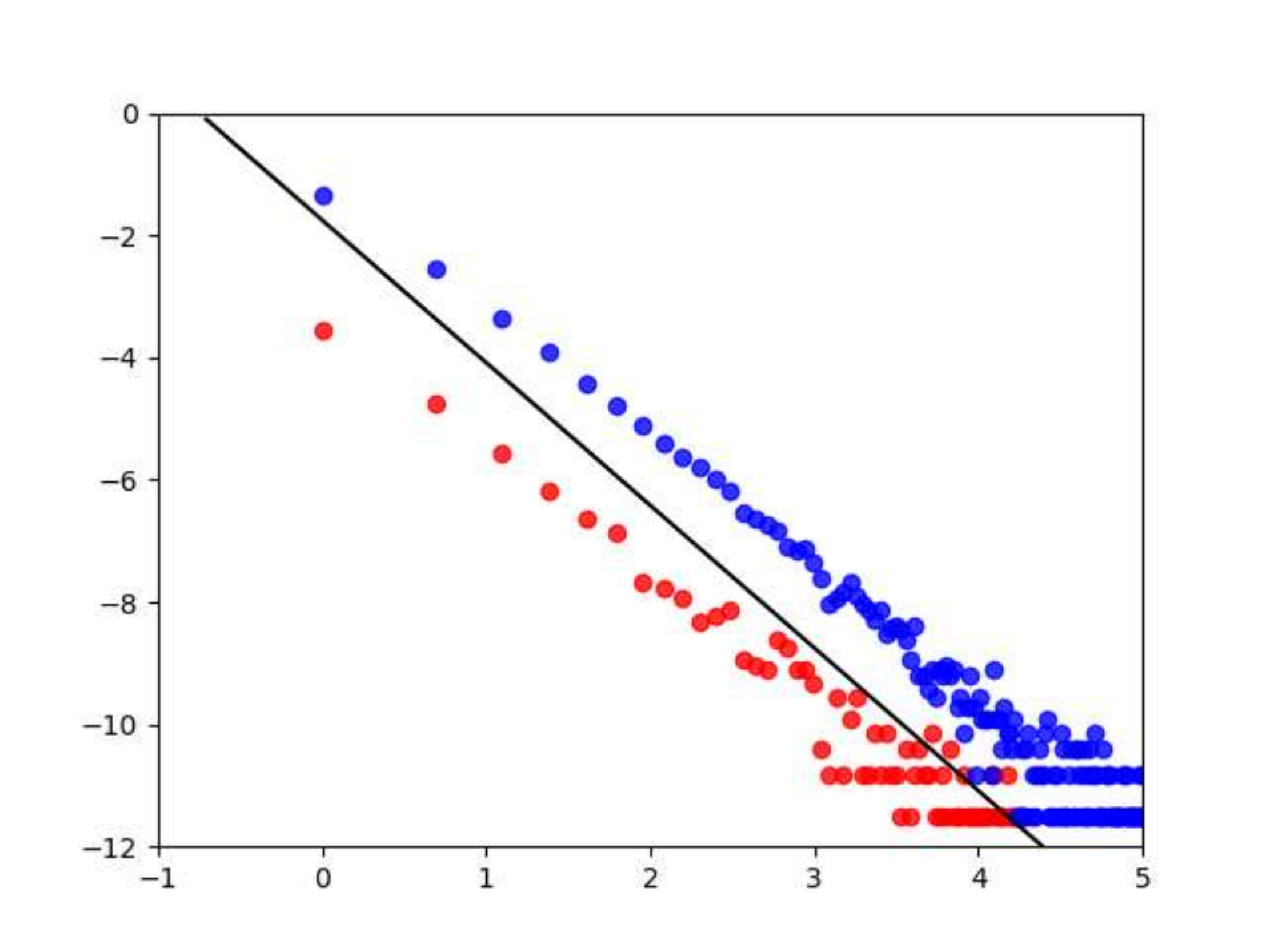}
\end{minipage}%
}%
\subfigure[$\theta=0.5,p_1=0.5,q=0.5$]{
\begin{minipage}[t]{0.3\linewidth}
\centering
\includegraphics[width=2in]{theta05-eps-converted-to.pdf}
\end{minipage}%
}%
\subfigure[$\theta=0.5,p_1=0.9,q=0.5$]{
\begin{minipage}[t]{0.3\linewidth}
\centering
\includegraphics[width=2in]{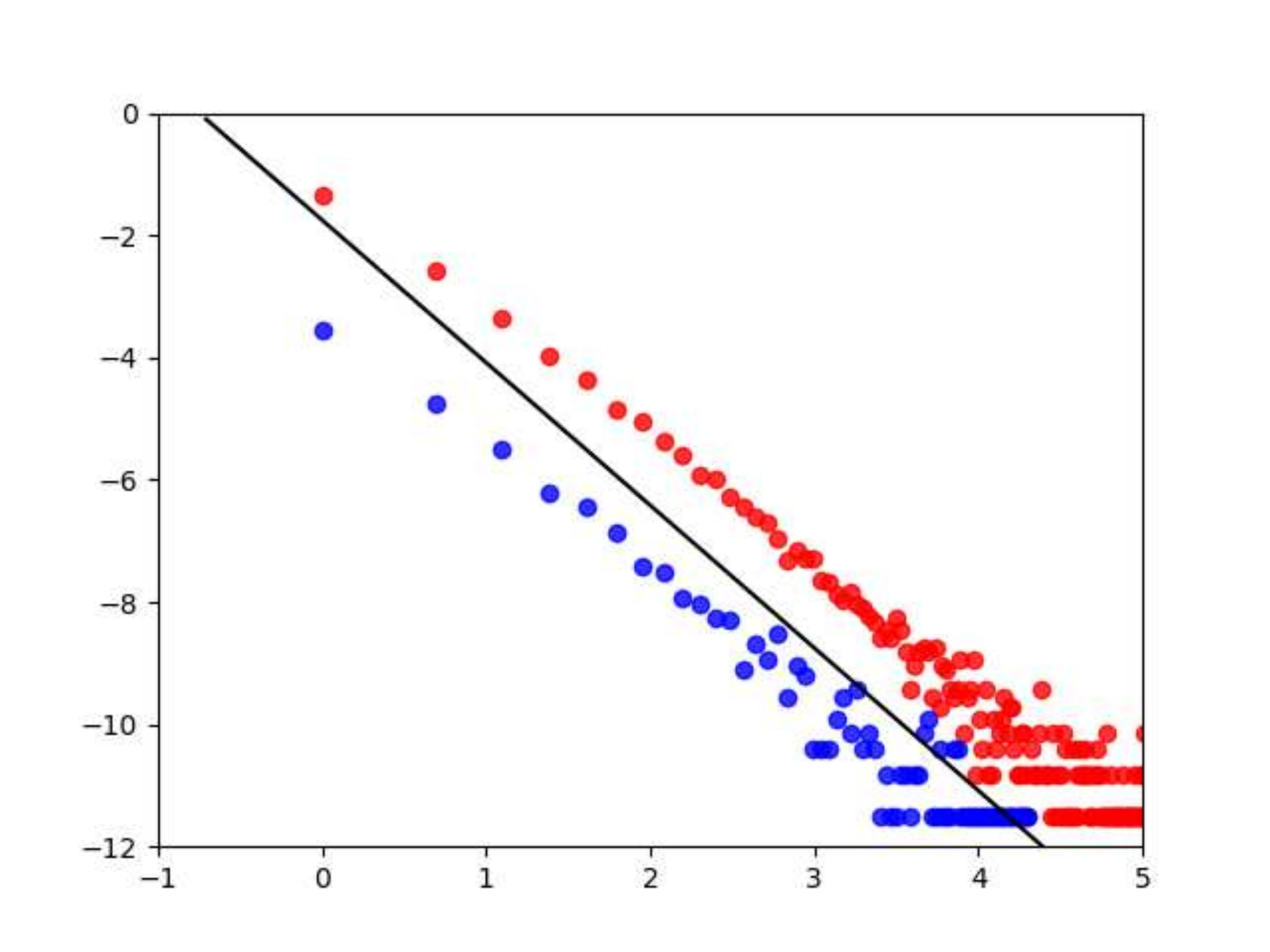}
\end{minipage}%
}%

\subfigure[$\theta=0.5,p_1=0.5,q=0.1$]{
\begin{minipage}[t]{0.3\linewidth}
\centering
\includegraphics[width=2in]{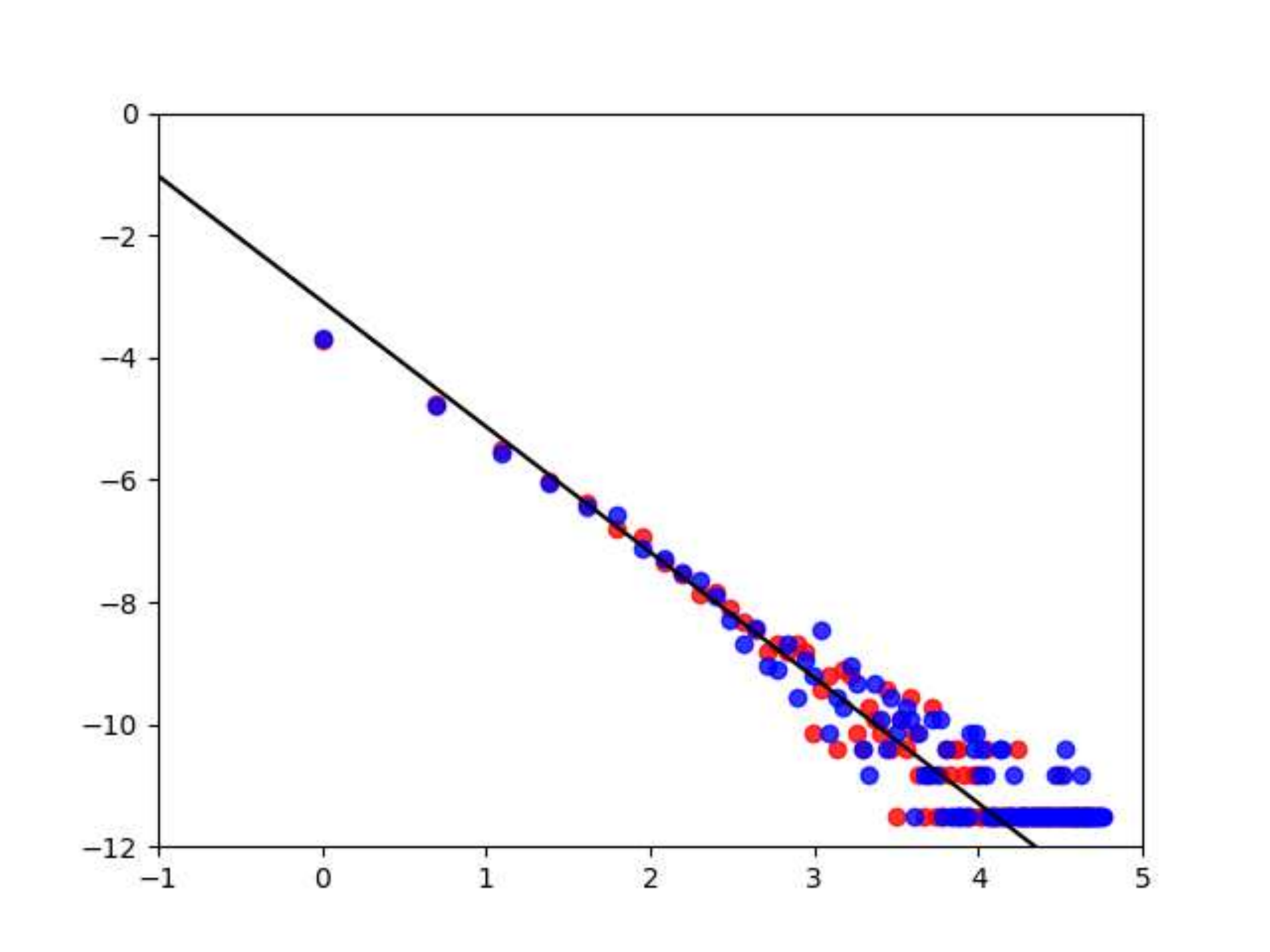}
\end{minipage}%
}%
\subfigure[$\theta=0.5,p_1=0.5,q=0.5$]{
\begin{minipage}[t]{0.3\linewidth}
\centering
\includegraphics[width=2in]{theta05-eps-converted-to.pdf}
\end{minipage}%
}%
\subfigure[$\theta=0.5,p_1=0.5,q=0.9$]{
\begin{minipage}[t]{0.33\linewidth}
\centering
\includegraphics[width=2in]{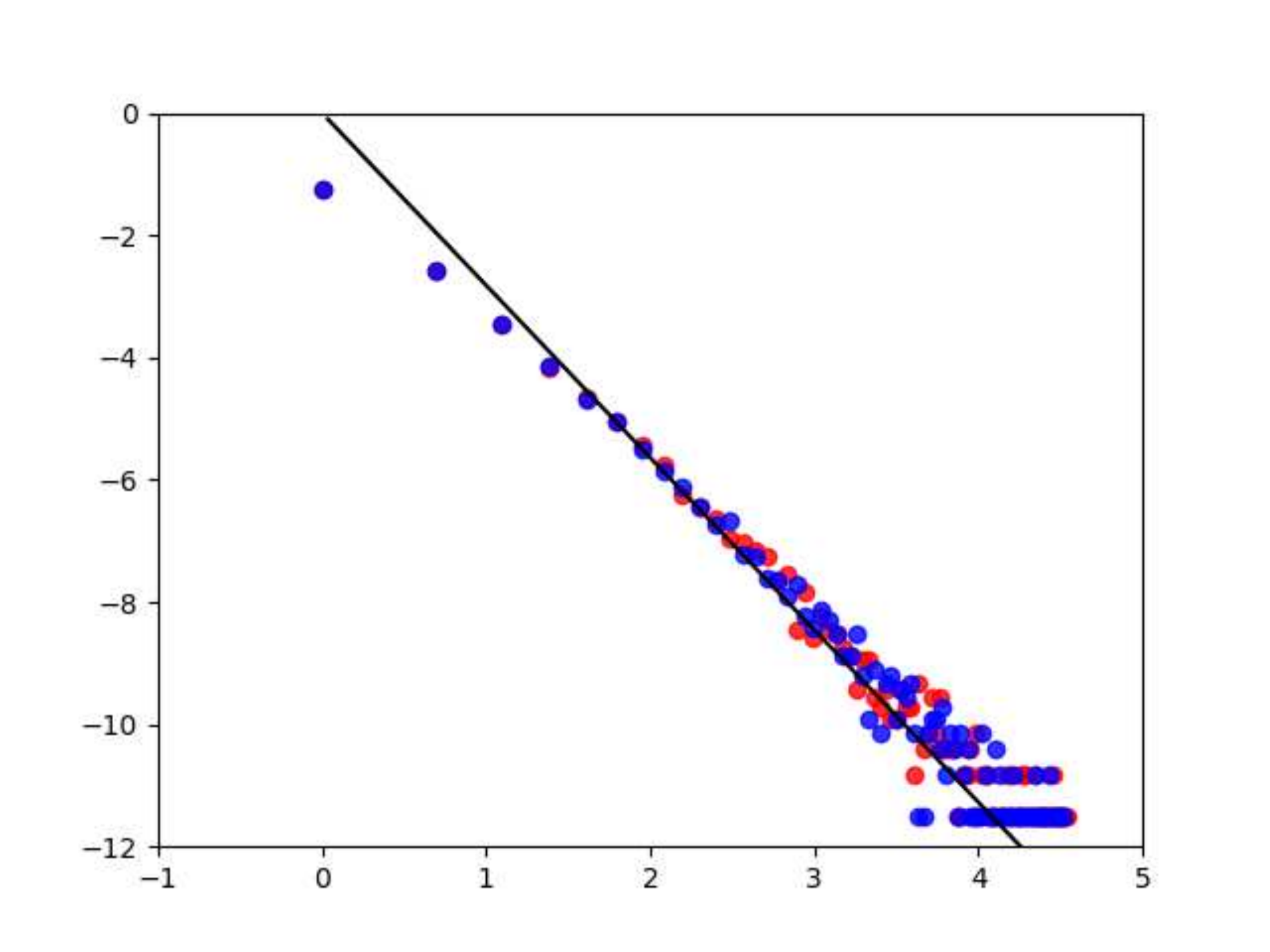}
\end{minipage}%
}%

\centering
\caption{The x-axis is $\log(d)$ and y-axis is $\log(\frac{m^k_d(T)}{T})$, where $T=100000,\,k=1,2$. The red and blue points come from different groups and the black line is for $y=\beta_k x+C$.}
\label{figlaw}
\end{figure}

\section{Parameter estimation}

\subsection{With historical information}
The parameters involved in the KPA model are $\psi=(\theta,\{p_k\}_{k=1}^K,q)$, and $\sum_{k=1}^K p_k=1$.
Let $\{G\}_{t=0}^T$ be the evolving network process generated by the KPA model in the time range $[0,T]$.
If we have all the information of $\{G\}_{t=0}^T$, we can get the MLE of the parameters:
\begin{eqnarray}
\nonumber &&\hat{p}_k=\frac{\sum_{t=1}^T \,\mathbf{1}_{\{v_t=1,g_w=k\}}}{\sum_{t=1}^T \mathbf{1}_{\{v_t=1\}}},\, k=1,\cdots,K\\
\nonumber &&\hat{q}=\frac{\sum_{t=1}^T {v_t}}{T}\\
&& \hat{\theta}=\arg \max_{\theta\in (0,1]}\log L_2(\theta|\{G_t\}_{t=0}^T).
\end{eqnarray}

$\{v_t=1,g_w=k\}$ implies that a vertex-step happens at time $t$ and the new node $w$ is from group $k$.
Letting $P^k_t=\frac{D_{t-1}^k}{2(t-1)+n_0}$, we have:
\begin{eqnarray}\label{mle1}
\nonumber &&\log L_2(\theta|\{G_t\}_{t=0}^T)\\
\nonumber =&&\sum_{t=1}^T \sum_{k=1}^K\mathbf{1}_{\{v_t=1,g_w=g_u=k\}}\log\left (P^k_t+(1-\theta)\left(1-P^k_t\right)\right)\\
\nonumber &&+\sum_{t=1}^T\sum_{k=1}^K{ \mathbf{1}_{\{v_t=1,g_w= k,g_u\neq k\}}}\log \left(\left(1-P^k_t\right)\theta\right)\\
\nonumber &&+ \sum_{t=1}^T\sum_{k=1}^K{\mathbf{1}_{\{v_t=0,g_w=g_u= k\}}}\log \left(P^k_t\left (P^k_t+(1-\theta)\left(1-P^k_t\right)\right) \right)\\
 &&+\sum_{t=1}^T\sum_{k=1}^K{\mathbf{1}_{\{v_t=0,g_w= k,g_u\neq k\}}}\log\left(P^k_t\left(1-P^k_t\right)\theta\right).
\end{eqnarray}

\begin{theorem}\label{th4}
Network data $\{G_t\}_{t=0}^T$ is generated by the KPA model under Assumptions \ref{a1}--\ref{a4}. Moreover, all the information of $\{G_t\}_{t=0}^T$ is known.
Let $\psi^*=(\theta^*,\{p^*_k\}_{k=1}^K,q^*)$ be the real parameters for the KPA model. $\hat{\psi}=(\hat{\theta},\{\hat{p}_k\}_{k=1}^K,\hat{q})$ are the MLE after given $\{G_t\}_{t=0}^T$. $\hat{\psi}$ satisfy:
\begin{eqnarray}
\hat{\psi}\stackrel{a.s.}{\longrightarrow} \psi^*.
\end{eqnarray}
\end{theorem}

Theorem \ref{th4} guarantees the MLE's convergence. Further, we can get the asymptotic normality.
\begin{theorem}\label{th5}
Under the conditions of Theorem \ref{th4}, the CLT for MLE is
\begin{equation}
\begin{aligned}
 \sqrt{T}\left(\mathbf{\hat{\psi}}-\psi^*\right)\xrightarrow{d} N(0,\mathbf{\Sigma}^{-1}),
\end{aligned}
\end{equation}
where
\begin{equation}
\begin{aligned}
\mathbf{\Sigma}=\begin{bmatrix}
\sum_{k=1}^K \left[\frac{p^*_{k}(1-p^*_{k})}{\theta^*}+\frac{p^*_{k}(1-p^*_{k})^2}{p^*_k+(1-p^*_k)(1-\theta^*)}\right]& 0& 0 \\
0 &\mathbf{\Sigma_{22}} &0\\
0&0&\frac{1}{q^*(1-q^*)}
\end{bmatrix}
\end{aligned}
\end{equation}
and
$\mathbf{\Sigma_{22}}$ is a $(K-1)\times (K-1)$ symmetric matrix satisfying
\begin{equation}
\Sigma_{22}(ij)=\begin{cases}
\frac{q^*}{1-\sum_{k=1}^{K-1}p^*_k},& i\neq j\\
\frac{q^*(1-\sum_{k\neq l}^{K-1}p^*_k)}{p^*_l(1-\sum_{k=1}^{K-1}p^*_k)},& i=j=l.
\end{cases}
\end{equation}
\end{theorem}

\begin{corollary}\label{co2}
Under the conditions of Theorem \ref{th4}, the CLT for $\mathbf{\hat{\Sigma}^{1/2}}\mathbf{\hat{\psi}}$ is
\begin{equation}
\begin{aligned}
 \sqrt{T}\mathbf{\hat{\Sigma}^{1/2}}(\mathbf{\hat{\psi}}-\psi^*)\xrightarrow{d} N(0,I_{K+1}),
\end{aligned}
\end{equation}

where $I_{K+1}$ is a ${(K+1)}\times {(K+1)}$ identity matrix. $\mathbf{\hat{\Sigma}}$ is the estimator of $\mathbf{{\Sigma}}$ that
\begin{equation}
\begin{aligned}
\mathbf{\hat{\Sigma}}=\begin{bmatrix}
\sum_{k=1}^K \left[\frac{\hat{p}_{k}(1-\hat{p}_{k})}{\hat{\theta}}+\frac{\hat{p}_{k}(1-\hat{p}_{k})^2}{\hat{p}_k+(1-\hat{p}_k)(1-\hat{\theta})}\right]& 0& 0 \\
0 &\mathbf{\hat{\Sigma}_{22}} &0\\
0&0&\frac{1}{\hat{q}(1-\hat{q})}
\end{bmatrix}.
\end{aligned}
\end{equation}

\begin{equation}
\hat{\Sigma}_{22}(ij)=\begin{cases}
\frac{\hat{q}}{1-\sum_{k=1}^{K-1}\hat{p}_k},& i\neq j\\
\frac{\hat{q}(1-\sum_{k\neq l}^{K-1}\hat{p}_k)}{\hat{p}_l(1-\sum_{k=1}^{K-1}\hat{p}_k)},& i=j=l.
\end{cases}
\end{equation}
\end{corollary}

Theorems \ref{th4}--\ref{th5} exhibit the excellent asymptotic properties of the MLE. Moreover, Corollary \ref{co2} allows us to do hypothesis testing on unknown parameters.

\subsection{Snapshot}

A snapshot of the evolving network means the present state of the network without historical information. In particular, at time $t$, we can only obtain the graph $G_t$ without information prior to time $t$.

We propose a parameter estimation procedure based on the single snapshot $G_T$. Let
\begin{eqnarray}
\nonumber &&\tilde{p}_k=\frac{V^k_{T}}{V_T},\, k=1,\cdots,K\\
\nonumber &&\tilde{q}=\frac{V_{T}}{E_T}\\
&&\tilde{\theta}=\arg \max_{\theta\in (0,1]} L_T(\theta|G_T).
\end{eqnarray}
$E_T$ and $V_T$ are the numbers of edges and nodes in graph $G_T$. $V^k_{T}$ is the number of nodes from group $k$ in $G_T$. $E_T^{k,1}$ is the number of edges where both nodes are from group $k$, and $E_T^{k,0}$ is the number of edges that connecting a node from group $k$ to other groups.
\begin{eqnarray}\label{smle1}
\nonumber &&L_T(\theta|G_T)\\
\nonumber =&&\sum_{k=1}^K\left[E^{k,1}_{T}\log (\frac{D_{T}^k}{2E_T}(\frac{D_{T}^k}{2E_T}+(1-\theta)(1-\frac{D_{T}^k}{2E_T})))\right]\\
 &&+\sum_{k=1}^K\left[E^{k,0}_{T}\log (\frac{D_{T}^k}{2E_T}\theta(1-\frac{D_{T}^k}{2E_T}))\right].
\end{eqnarray}

\begin{theorem}\label{th6}
The evolving network data $\{G_t\}_{t=0}^T$ is generated by the KPA model under Assumptions \ref{a1}--\ref{a4}. Only the information of $G_T$ is known. Let $\psi^*=(\theta^*,\{p^*_k\}_{k=1}^K,q^*)$ be the real parameters for the KPA model. The parameters' estimator $\tilde{\psi}=(\tilde{\theta},\{\tilde{p}_k\}_{k=1}^K,\tilde{q})$ based on the single snapshot $G_T$ satisfies:
\begin{eqnarray}
\tilde{\psi}\stackrel{a.s.}{\longrightarrow} \psi^*.
\end{eqnarray}
\end{theorem}

\section{The change point of $\theta$}
We assume that there is a change point $\tau^*$ when the parameter $\theta^*$ changes, which means that the influence of homophily on the network structure changes. If the value of $\theta^*$ becomes larger, nodes are more receptive to different groups, and vice versa.

\begin{assumption}\label{a5}
$\tau^*\in[t_0,T-t_0]$, where  $\frac{t_0}{T}\equiv c_0$, $c_0$ is a constant and $c_0\in (0,1)$.
\end{assumption}

\begin{assumption}\label{a6}
There is no difference for parameters $(\{p^*_k\}_{k=1}^K,q^*)$  before and after the change point $\tau^*$.
\end{assumption}

Assumption \ref{a5} guarantees we have enough information before and
after the change to detect the point $\tau^*$. We assume that within the period $[0,\tau^*]$, the network follows a KPA model with parameter $\theta^*_1$. Within the period $[\tau^*+1,T]$, the network follows a KPA model with parameter $\theta^*_2$ such that $\theta^*_2\neq \theta^*_1$. Assumption \ref{a6} excludes the influence of other factors. Nevertheless, our results can be easily
extended if other parameters also change.

We estimate the change point $\tau^*$ by the maximum likelihood method:
\begin{eqnarray}\label{changepoint}
\nonumber &&\hat{\tau}=\arg\max_{t_0\leq \tau\leq T-t_0} [\max_{\theta_1\in(0,1]}{\log{L_2}({\theta}_{1}|\{G(t)\}_{t=0}^{\tau})}+\max_{\theta_2\in(0,1]} \log{L_2}({\theta}_2|\{G(t)\}_{t=\tau}^T]\\
\nonumber&&\hat{\theta}_1=\arg\max_{{\theta}_1\in(0,1] }{\log{L_2}({\theta}_{1}|\{G(t)\}_{t=0}^{\hat{\tau}})}\\
&&\hat{\theta}_2={\arg\max_{{\theta}_2\in(0,1] }\log{L_2}({\theta}_2|\{G(t)\}_{t=\hat{\tau}}^T)}.
\end{eqnarray}

$\log{L_2}({\theta}|\{G(t)\})$ is defined by Equation (\ref{mle1}).

\begin{theorem}\label{changep}
Under Assumptions \ref{a1}-- \ref{a6}, we have
\begin{equation}
\begin{aligned}
\frac{|\hat{\tau}-\tau^*|}{T}\xrightarrow{i.p.} 0.
\end{aligned}
\end{equation}
\end{theorem}

\section{Proofs}

\begin{proof}[Proof of Theorem \ref{th1}]
Let $Z^k_t=D^k_t-D^k_{t-1},t>0$ be the increased degree of group $k$ at time $t$. We define the $\sigma-$algebra $\mathcal{F}_t=\sigma(G_0,\cdots,G_t)$. Let $P_t^k=\frac{D^k_{t-1}}{2(t-1)+n_0}$,
by Lemmas \ref{le1}--\ref{le2}, we have
\begin{eqnarray}\label{eq1}
\nonumber E[Z^k_{t}|\mathcal{F}_{t-1}]=&&2qp_k\left[P_t^k+(1-\theta)\left(1-P_t^k\right)\right]+\theta q(1-p_k) P_t^k+\theta qp_k (1- P_t^k)\\
\nonumber&&+2(1-q)P_t^k\left[P_t^k+(1-\theta)\left(1-P_t^k\right)\right]+2\theta(1-q)(1-P_t^k) P_t^k\\
\nonumber  =&&2qp_k-qp_k\theta\left(1-P_t^k\right)+[\theta q(1-p_k)+2(1-q)]P_t^k\\
=&&(2-\theta)qp_k+[2-q(2-\theta)]P_t^k.
\end{eqnarray}

Let $Z^0_k(i):=Z_{i}^k-E(Z_{i}^k|\mathcal{F}_{i-1})$, $S_{k}(t)=\sum_{i=1}^{t} Z^0_k{(i)}$.
By Chebyshev's inequality, for any $\epsilon>0$,
\begin{equation}\label{eq2}
 P\left(\left|\frac{ S_k({t^2})}{t^2}\right|>\epsilon\right)\leq \frac{E(S_k({t^2}))^2}{t^4\epsilon^2}=\frac{\sum_{i=1}^{t^2} E (Z^0_k{(i)})^2 + \sum_{0<i\neq j\leq t } E(Z^0_k{(i)} Z^0_k{(j)})}{t^4\epsilon^2}.
\end{equation}

Further, $E [Z^0_k{(t)}]^2\leq4$ and for any $i>j$,
\begin{eqnarray}
 E[Z^0_k{(i)}Z^0_k{(j)}]=E\left[E[Z^0_k{(i)}Z^0_k{(j)}|\mathcal{F}_{i-1}]\right]=E\left[Z^0_k{(j)}E[Z^0_k{(i)}|\mathcal{F}_{i-1}]\right]=0.
\end{eqnarray}

By Equations (\ref{eq2})--(\ref{eq3}), we can conclude that
\begin{equation}\label{eq3}
P\left(\left|\frac{S_k{(t^2)}}{t^2}\right|>\epsilon\right)\leq \frac{4 t^2}{t^4\epsilon}\leq \frac{c}{t^2}.
\end{equation}

Equation (\ref{eq3}) implies
$\quad \sum_{t=1}^{\infty}P\left(\left|\frac{ S_k({t^2})}{t^2}\right|>\epsilon\right) <\infty$. Further, we can get
\begin{equation}\label{eq4}
\frac{S_k({t^2})}{t^2}\stackrel{a.s.}{\longrightarrow}0.
\end{equation}

From Equation (\ref{eq4}),
$\forall t\geq 1$, there exists $m$ such that
\begin{eqnarray}\label{eq5}
\nonumber &&\max_{m^2\leq t<(m+1)^2}\left|\frac{S_k(t)}{t}\right|\leq \frac{1}{m^2}\max_{m^2\leq t<(m+1)^2}|S_k(t)|\\
\nonumber \leq&& \frac{1}{m^2}\left(|S_k({m^2})|+\max_{m^2\leq t<(m+1)^2}|S_k(t)-S_k({m^2})|\right)\\
\nonumber \leq&& \frac{1}{m^2} \left(|S_k({m^2})|+2((m+1)^2-m^2)\right)\\
\leq &&\frac{|S_k({m^2})|}{m^2}+\frac{4m+2}{m^2}.
\end{eqnarray}

Equation (\ref{eq5}) implies $\frac{S_k({t})}{t}\stackrel{a.s.}{\longrightarrow} 0$.
Next, by the convergence theorem,
\begin{eqnarray}\label{eq6}
\nonumber\lim_{t\rightarrow \infty}\frac{D^k_{t}}{2t}=&&\lim_{t\rightarrow \infty} \left(\frac{n_0p_k+\sum_{i=1}^t Z_i^k}{2t}\right)\\
\nonumber=&& \lim_{t\rightarrow \infty}\left[\frac{n_0p_k+S_k(t)+\sum_{i=1}^t E(Z_{i}^k|\mathcal{F}_{i-1})}{2t}\right]\\
=&&\lim_{t\rightarrow \infty}\frac{\sum_{i=1}^t E(Z_{i}^k|\mathcal{F}_{i-1})}{2t}.
\end{eqnarray}

When $t$ is large enough, we can conclude from Equations (\ref{eq1}) and (\ref{eq6}) that
\begin{eqnarray}\label{f1}
\nonumber \frac{D_t^k}{2t+n_0}=&&\frac{\sum_{i=1}^t E(Z_{i}^k|{G}_{i-1})}{2t}+o_p(1)\\
=&& \frac{(2-\theta)qp_k}{2}+\frac{[2-q(2-\theta)]}{2t}\sum_{i=1}^t\frac{D_{i-1}^k}{2(i-1)+n_0}+o_p(1).
\end{eqnarray}

According to the Equation (\ref{f1}), let:
\begin{equation}\label{f2}
f_k(x)=\frac{(2-\theta)qp_k}{2}+\frac{[2-q(2-\theta)]}{2}x.
\end{equation}

By the Banach fixed point theorem,

$f_k(x):\,(R,\,| \cdot |)\to (R,\,| \cdot |)$ is a contraction mapping with only one fixed point $x=p_k$.

By Equation (\ref{f1}), {when $t$ is large enough}, we have:
\begin{equation}\label{f3}
\left|\frac{D_t^k}{2t+n_0}-p_k\right|=\frac{[2-q(2-\theta)]}{2}\left|\frac{\sum_{i=1}^t[\frac{D_{i-1}^k}{2(i-1)+n_0}-p_k]}{t}\right|,
\end{equation}
where $0<\frac{|2-q(2-\theta)|}{2}<1$.

Equation (\ref{f3}) implies that $\frac{D_t^k}{2t+n_0}$ approaches $p_k$ as $t\to \infty$. By Equations (\ref{f1})--(\ref{f3}), we can deduce $\frac{D_t^k}{2t}\stackrel{a.s.}{\longrightarrow}p_k$.
\end{proof}

\begin{proof}[Proof of Corollary \ref{co1}]
By the same operations as in Equations (\ref{eq1})--(\ref{eq6}), we have:
\begin{equation}\label{eq7}
\nonumber\frac{\sum_{i=1}^t [X_i-E(X_i|\mathcal{F}_{i-1})] }{t}\stackrel{a.s.}{\longrightarrow} 0.
\end{equation}

Then:
\begin{equation}\label{eq8}
\nonumber\frac{S_t}{t}=\frac{\sum_{i=1}^t [X_i-E(X_i|\mathcal{F}_{i-1})]+\sum_{i=1}^tE(X_i|\mathcal{F}_{i-1}) }{t}\stackrel{a.s.}{\longrightarrow}\lim_{t\to\infty} \frac{\sum_{i=1}^tE(X_i|\mathcal{F}_{i-1})}{t}.
\end{equation}

Let $P_t^k=\frac{D^k_{t-1}}{2(t-1)+n_0}$, we can get:
\begin{eqnarray}\label{eq9}
\nonumber &&E(X_{t}|\mathcal{F}_{t-1})\\
\nonumber=&&q\sum_{k=1}^K p_k \left[P_t^k+(1-\theta)\left(1-P_t^k\right)\right]+(1-q)\sum_{k=1}^K P_t^k \left[P_t^k+(1-\theta)\left(1-P_t^k\right)\right].
\end{eqnarray}

By Theorem \ref{th1}, we have $\frac{D_t^k}{2t}\stackrel{a.s.}{\longrightarrow} p_k$, so
\begin{eqnarray}
\nonumber \frac{S_t}{t}\stackrel{a.s.}{\longrightarrow}\lim_{t\to\infty} \nonumber\frac{\sum_{i=1}^tE(X_i|\mathcal{F}_{i-1})}{t}=&&\lim_{t\to\infty} E(X_t|\mathcal{F}_{t-1})\\
\nonumber=&&\sum_{k=1}^K p_k\left [1-\theta(1-p_k)\right]\\
\nonumber= &&1+\theta\left(\sum_{k=1}^Kp_k^2-1\right).
\end{eqnarray}
\end{proof}

\begin{proof}[Proof of the Theorem \ref{th2}]
\begin{eqnarray}\label{eq10}
\nonumber E\left[D_t^k-p_k(2t+n_0)\right|\mathcal{F}_{t-1}]=&& \left[1+\frac{2-q(2-\theta)}{2(t-1)+n_0}\right]D_{t-1}^k+(2-\theta)qp_k-p_k(2t+n_0)\\
\nonumber =&&\left[1+\frac{2-q(2-\theta)}{2(t-1)+n_0}\right]\left[D_{t-1}^k-p_k(2(t-1)+n_0)\right].
\end{eqnarray}

Let $X^0_k({t})=\frac{D_{t}^k-p_k(2t+n_0)}{\prod_{i=1}^t [1+\frac{2-q(2-\theta)}{2(i-1)+n_0}]},\forall t\geq 1$. $X^0_k(0)={D_{0}^k-p_kn_0}=0$. Obviously, $\{X^0_k({t})\}_t$ is a martingale sequence.
\begin{eqnarray}\label{eq11}
\nonumber X^0_k({t})-X^0_k({t-1})= &&\frac{D_{t}^k-p_k(2t+n_0)}{\prod_{i=1}^t [1+\frac{2-q(2-\theta)}{2(i-1)+n_0}] }-\frac{D_{t-1}^k-p_k(2t-2+n_0)}{\prod_{i=1}^{t-1} [1+\frac{2-q(2-\theta)}{2(i-1)+n_0}] }\\
\nonumber =&&\frac{D_{t}^k-p_k(2t+n_0)-[1+\frac{2-q(2-\theta)}{2(t-1)+n_0}][D_{t-1}^k-p_k(2t-2+n_0)]}{\prod_{i=1}^t [1+\frac{2-q(2-\theta)}{2(i-1)+n_0}]}\\
\nonumber=&&\frac{(D_{t}^k-D_{t-1}^k)-2p_k+\left[\frac{2-q(2-\theta)}{2(t-1)+n_0}\right][D_{t-1}^k-p_k(2t-2+n_0)]}{\prod_{i=1}^t \left[1+\frac{2-q(2-\theta)}{2(i-1)+n_0}\right] }.
\end{eqnarray}

$D^k_t$ is the total of the degrees from group $k$ in the graph $G_t$, so $0\leq D_{t}^k-D_{t-1}^k\leq 2,\, 1\leq D_{t-1}^k\leq 2t-2+n_0$.
Then we have
\begin{equation}\label{eq12}
\begin{aligned}
\nonumber|X^0_k({t})-X^0_k({t-1})|\leq\frac{4}{\prod_{i=1}^t \left[1+\frac{2-q(2-\theta)}{2(i-1)+n_0}\right]}.
\end{aligned}
\end{equation}
What is more, we get the following upper bound for $Var(X_k^0(t)|\mathcal{F}_{t-1})$:

\begin{eqnarray}\label{eq13}
\nonumber Var(X^0_k(t)|\mathcal{F}_{t-1})=&&\frac{Var(D_{t}^k-p_k(2t+n_0)|\mathcal{F}_{t-1})}{\prod_{i=1}^t \left[1+\frac{2-q(2-\theta)}{2(i-1)+n_0}\right]^2}=\frac{Var(D_{t}^k|\mathcal{F}_{t-1})}{\prod_{i=1}^t \left[1+\frac{2-q(2-\theta)}{2(i-1)+n_0}\right]^2}\\
\nonumber =&&\frac{E[(D_{t}^k-E(D_t^k|\mathcal{F}_{t-1}))^2|\mathcal{F}_{t-1}]}{\prod_{i=1}^t \left[1+\frac{2-q(2-\theta)}{2(i-1)+n_0}\right]^2}\\
\nonumber \leq&& \frac{E[(D_{t}^k-D_{t-1}^k)^2|\mathcal{F}_{t-1}]}{\prod_{i=1}^t \left[1+\frac{2-q(2-\theta)}{2(i-1)+n_0}\right]^2}\\
\nonumber\leq&& \frac{4}{\prod_{i=1}^t \left[1+\frac{2-q(2-\theta)}{2(i-1)+n_0}\right]^2}.
\end{eqnarray}

By Theorems 2.22 and 2.26 in \cite{chung2006complex}, let $M_j=\frac{4}{\prod_{i=1}^j \left[1+\frac{2-q(2-\theta)}{2(i-1)+n_0}\right] }$ and $\sigma_j^2=\frac{4}{\prod_{i=1}^j \left[1+\frac{2-q(2-\theta)}{2(i-1)+n_0}\right]^2}$.
Then we get
\begin{equation}\label{eq14}
\begin{aligned}
\nonumber Pr(|X_k^0(t)-EX_k^0(t)|\geq\lambda)\leq e^{-\frac{\lambda^2}{2\sum_{j=1}^t(\sigma_j^2 +M_j^2)}}.
\end{aligned}
\end{equation}

Note that $EX_k^0(t)=EX_k^0(1)=X_k^0(0)=0$, so we have:
\begin{equation}\label{eq15}
\begin{aligned}
Pr(|X_k^0(t)|\geq\lambda)\leq e^{-\frac{\lambda^2}{2\sum_{j=1}^t(\sigma_j^2 +M_j^2)}}.
\end{aligned}
\end{equation}

By Stirling's approximation,
\begin{eqnarray}\label{eq16}
\nonumber \prod_{i=1}^t \left[1+\frac{2-q(2-\theta)}{2(i-1)+n_0}\right]
=&& \prod_{i=1}^t \left[\frac{2(i-1)+n_0+2-q(2-\theta)}{2(i-1)+n_0}\right]\\
\nonumber =&& \frac{\prod_{i=1}^t[i+(n_0-q(2-\theta))/2]}{\prod_{i=1}^t[i+(n_0-2)/2]}\\
=&& O(t^{1-q(2-\theta)/2})
\end{eqnarray}
\begin{eqnarray}\label{eq16.5}
 \prod_{i=1}^t \left[1+\frac{2-q(2-\theta)}{2(i-1)+n_0}\right]^2
= O(t^{2-q(2-\theta)})
\end{eqnarray}

For $t\geq 1$, we have:
\begin{equation}\label{eq17}
\begin{aligned}
\sum_{j=1}^t(\sigma_j^2 +M_j^2)=O\left(\sum_{j=1}^t j^{q(2-\theta)-2}\right).
\end{aligned}
\end{equation}

For $q(2-\theta)-2<0$, the function $f(x)=x^{q(2-\theta)-2}$ is strictly monotonically decreasing with respect to $x$ on $(0,+\infty)$.

Then we have the inequality:
\begin{equation}\label{eq18}
\begin{aligned}
\int_1^t j^{q(2-\theta)-2}{d}j< \sum_{j=1}^t j^{q(2-\theta)-2}<1+\int_1^{t-1} j^{q(2-\theta)-2}{d}j.
\end{aligned}
\end{equation}

Further,
\begin{equation}\label{eq19}
\begin{aligned}
\int_1^t j^{q(2-\theta)-2}{d}j
=\begin{cases}
t^{q(2-\theta)-1}-1,  & \text{if }\, \frac{1}{2-\theta}<q\leq 1 \\
\log t, & \text{if }\,q =\frac{1}{2-\theta}\\
1-t^{q(2-\theta)-1},  & \text{if }\,0<q < \frac{1}{2-\theta}.
\end{cases}
\end{aligned}
\end{equation}

Equations (\ref{eq18})--(\ref{eq19}) imply that when $t$ is large enough, we have
\begin{equation}\label{eq20}
\begin{aligned}
\sum_{j=1}^t j^{q(2-\theta)-2}
=\begin{cases}
O(t^{q(2-\theta)-1}),  & \text{if }\, \frac{1}{2-\theta}<q\leq 1 \\
O(\log t), & \text{if }\,q =\frac{1}{2-\theta}\\
O(1),  & \text{if }\,0<q < \frac{1}{2-\theta} .
\end{cases}
\end{aligned}
\end{equation}

Combine Equations (\ref{eq15})--(\ref{eq20}), we can get:

(1) If $\frac{1}{2-\theta}<q\leq 1 $, let $\lambda=\frac{2c(t)\sqrt{t}}{\prod_{i=1}^t \left[1+\frac{2-q(2-\theta)}{2(i-1)+n_0}\right]}$ and $c(t)$ is a strictly monotonically increasing function of $t$, then
\begin{equation}\label{eq21}
\begin{aligned}
\nonumber P\left(|X^0_k(t)|\geq\frac{2c(t)\sqrt{t}}{\prod_{i=1}^t \left[1+\frac{2-q(2-\theta)}{2(i-1)+n_0}\right]}\right)\leq e^{-\frac{\lambda^2}{2\sum_{i=1}^n(\sigma_i^2 +M_i^2)}}\propto e^{-c(t)^2}.
\end{aligned}
\end{equation}
\begin{equation}\label{eq22}
\begin{aligned}
\nonumber P\left(|D_t^k-p_k(2t+n_0)|\geq{2c(t)\sqrt{t}}\right)\leq C e^{-c(t)^2}.
\end{aligned}
\end{equation}

(2) Otherwise, if $q=\frac{1}{2-\theta}$, let $\lambda=\frac{2c(t)\sqrt{\log(t)}}{\prod_{i=1}^t \left[1+\frac{2-q(2-\theta)}{2(i-1)+n_0}\right]}$, then
\begin{equation}\label{eq24}
\begin{aligned}
\nonumber P\left(|D_t^k-p_k(2t+n_0)|\geq{2c(t)\sqrt{\log(t)}}\right)\leq C e^{-\frac{c(t)^2}{t}}.
\end{aligned}
\end{equation}

(3) Otherwise, if $0<q < \frac{1}{2-\theta}$, let $\lambda=\frac{2c(t)}{\prod_{i=1}^t \left[1+\frac{2-q(2-\theta)}{2(i-1)+n_0}\right]}$,
then
\begin{equation}\label{eq26}
\begin{aligned}
\nonumber P\left(|D_t^k-p_k(2t+n_0)|\geq 2c(t)\right)\leq C e^{-\frac{c(t)^2}{t^{2-q(2-\theta)}}}.
\end{aligned}
\end{equation}
\end{proof}

\begin{lemma}\label{le3}
Suppose that there are stochastic processes $\{a_t\},\,\{b_t\}$, and constant sequence $\{c_t\}$ satisfying the recurrence relation
\begin{equation}\label{eq27}
E(a_{t+1})=E\left[\left(1-\frac{b_t}{t+t_1}\right)a_t\right]+c_t ,\,\forall t\geq t_0.
\end{equation}
where $\lim_{t\to\infty} c_t=c,\, \lim_{t\to\infty} b_t\stackrel{a.s.}{\longrightarrow} b>0 \text{ and } 0<\frac{a_t}{t}\leq C$, $C$ is a positive constant. Then $\lim_{t\to\infty}\frac{E(a_t)}{t}$ exists and
\begin{equation}\label{eq28}
\lim_{t\to\infty}\frac{E(a_t)}{t}=\frac{c}{1+b}.
\end{equation}
\end{lemma}

\begin{proof}[Proof of Lemma \ref{le3}]
Without loss of generality, we can assume $t_1 = 0$ after shifting $t$ by $t_1$.
By rearranging the recurrence relation, we have
\begin{eqnarray}\label{eq29}
\nonumber \frac{E(a_{t+1})}{t+1}-\frac{c}{1+b}=&&\frac{E\left[\left(1-\frac{b_t}{t}\right)a_t\right]+c_t}{t+1}-\frac{c}{1+b}\\ \nonumber=&&\left[\frac{E(a_t)}{t}-\frac{c}{1+b}\right]\left(1-\frac{1+b}{t+1}\right)+\frac{c_t-c}{1+t}+\frac{E(({b-b_t}){a_t})}{t(t+1)}.
\end{eqnarray}

Letting $s_t=\left|\frac{E(a_t)}{t}-\frac{c}{1+b}\right|$, the triangle inequality now gives:
\begin{equation}\label{eq30}
\nonumber s_{t+1}\leq s_t\left|1-\frac{1+b}{t+1}\right|+\left|\frac{c_t-c}{t+1}\right|+\left|\frac{E((b-b_t){a_t})}{t(t+1)}\right|.
\end{equation}

Using the fact that ${\lim_{t\to\infty} E(b_t)}=b,\,\lim_{t\to\infty}c_t= c,\left|\frac{a_t}{t}\right|\leq C$, $\forall \epsilon>0$, when $t$ is large enough, we have:
\begin{equation}\label{eq31}
\nonumber\left|\frac{c_t-c}{t+1}\right|+\left|\frac{E((b-b_t){a_t})}{t(t+1)}\right|<\frac{(1+b)\epsilon}{t+1}.
\end{equation}

For fixed $\epsilon>0$, when $t$ is large enough, we have
\begin{equation}\label{eq32}
\nonumber s_{t+1}-\epsilon<(s_{t}-\epsilon)\left(1-\frac{1+b}{t+1}\right),\, \left(1-\frac{1+b}{t+1}\right)>0.
\end{equation}

Since $0<1+b<\infty$, it is not difficult to see that $\prod_{t=1}^{\infty}\left(1-\frac{1+b}{t+1}\right)\to 0$.
Then for fixed $\epsilon>0$:
\begin{equation}\label{eq33}
\nonumber0\leq\lim_{t\to\infty} s_t<\epsilon.
\end{equation}

$\epsilon$ can be chosen arbitrarily, so we can get $\lim_{t\to\infty} s_t\to 0$ as $\epsilon\to 0$ as desired.
Therefore, we have proved that
\begin{equation}\label{eq34}
\nonumber\lim_{t\to\infty}\frac{E(a_t)}{t}=\frac{c}{1+b}.
\end{equation}
\end{proof}

\begin{proof}[Proof of Theorem \ref{th3}]

$\forall d\geq 2$, for each time $t$, let $x=\frac{m^k_d(t-1)\times d}{2(t-1)+n_0},\,y=\frac{m^k_{d-1}(t-1)\times (d-1)}{2(t-1)+n_0},\, z=\frac{D_{t-1}^k}{2(t-1)+n_0}$.
Consider analyzing  $E[m^k_d(t)|{G}_{t-1}]$'s situation in detail:

For the vertex-step at time $t$,
\begin{itemize}
    \item if the new node is connected to the old node with degree $d-1$ from group $k$, then $m^k_d(t)$ will be $m^k_d(t-1)+1$. The probability times $m^k_d(t)$ is:
\begin{equation}\label{st1}
q\left[p_k(1-z)(1-\theta)\frac{y}{z}+p_ky+(1-p_k)y\theta\right]\times(m^k_d(t-1)+1).
\end{equation}
    \item if the new node is connected to the old node with degree $d$ from group $k$, then $m^k_d(t)$ will be $m^k_d(t-1)-1$. The probability times $m^k_d(t)$ is:
\begin{equation}\label{st2}
q\left[p_k(1-z)(1-\theta)\frac{x}{z}+p_kx+(1-p_k)x\theta\right]\times(m^k_d(t-1)-1).
\end{equation}
\item $m^k_d(t)$ is the same as $m^k_d(t-1)$. The probability times $m^k_d(t)$ is:
\begin{equation}\label{st3}
q\left[1-p_k(1-z)(1-\theta)\frac{x+y}{z}-p_k(x+y)-(1-p_k)(x+y)\theta\right]\times m^k_d(t-1).
\end{equation}
\end{itemize}

For the edge-step at time $t$,
\begin{itemize}
    \item if the new edge's two nodes with degree $d-1$ are from group $k$, then $m^k_d(t)$ will be $m^k_d(t-1)+2$. The probability times $m^k_d(t)$ is:
\begin{equation}\label{st4}
(1-q)y[y+(1-z)(1-\theta)\frac{y}{z}]\times(m_d^k(t-1)+2).
\end{equation}
    \item if the new edge's two nodes with degree $d$ are from group $k$, then $m^k_d(t)$ will be $m^k_d(t-1)-2$. The probability times $m^k_d(t)$ is:
\begin{equation}\label{st5}
(1-q)x[x+(1-z)(1-\theta)\frac{x}{z}]\times (m_d^k(t-1)-2).
\end{equation}
  \item if one vertex of the new edge is from group $k$ with degree $d-1$, and the other is neither from group $k$ with degree $d$ nor degree $d-1$, then $m^k_d(t)$ will be $m^k_d(t-1)+1$. The probability times $m^k_d(t)$ is:
\begin{equation}\label{st6}
(1-q)[2y(1-z)\theta+2(z-y-x)(y+(1-z)(1-\theta)\frac{y}{z})]\times (m^k_d(t-1)+1).
\end{equation}
  \item if one vertex of the new edge is from group $k$ with degree $d$, and the other is neither from group $k$ with degree $d$ nor degree $d-1$, then $m^k_d(t)$ will be $m^k_d(t-1)-1$. The probability times $m^k_d(t)$ is:
\begin{equation}\label{st7}
(1-q)[2x(1-z)\theta+2(z-y-x)(x+(1-z)(1-\theta)\frac{x}{z})]\times (m^k_d(t-1)-1).
\end{equation}
\item if the $m^k_d(t)$ is the same as $m^k_d(t-1)$. The probability times $m^k_d(t)$ is:
\begin{eqnarray}\label{st8}
\nonumber &&(1-q)[1-(x^2+y^2)(1+\frac{(1-z)(1-\theta)}{z})
 -2(y+x)(1-z)\theta\\
 &&-2(y+x)(z-y-x)(1+\frac{(1-z)(1-\theta)}{z})]\times m_d^k(t-1).
\end{eqnarray}
\end{itemize}

By summing Equations (\ref{st1})--(\ref{st8}):
\begin{eqnarray}\label{eq35}
\nonumber  E(m^k_d(t)|{G}_{t-1})=&&m_d^k(t-1)+q[p_k(y-x)(1+(1-z)(1-\theta)\frac{1}{z})+(1-p_k)\theta(y-x)]\\
\nonumber &&+2(1-q)[(y^2-x^2)(1+(1-z)(1-\theta)\frac{1}{z})]\\
\nonumber &&+(1-q)[2(y-x)(1-z)\theta+2(y-x)(z-y-x)(1+(1-z)(1-\theta)\frac{1}{z})]\\
\nonumber =&& m_d^k(t-1)+q\left[p_k(y-x)\frac{(1-\theta)}{z}+\theta(y-x)\right]+2(1-q)(y-x)\\
\nonumber =&& m_d^k(t-1)\left[1-\frac{qdp_k(1-\theta)}{D_{t-1}^k}-\frac{d(2-(2-\theta)q)}{2(t-1)+n_0}\right]
\\
\nonumber &&+m_{d-1}^k(t-1)\left[\frac{q(d-1)p_k(1-\theta)}{D_{t-1}^k}+\frac{(d-1)(2-(2-\theta)q)}{2(t-1)+n_0}\right].
\end{eqnarray}

If we take the expectation of both sides, we get the following recurrence formula.
\begin{eqnarray}\label{eq36}
\nonumber E(m_d^k(t))= &&E\left[m_d^k(t-1)[1-\frac{qdp_k(1-\theta)}{D_{t-1}^k}-\frac{d(2-(2-\theta)q)}{2(t-1)+n_0}]\right]
\\
\nonumber &&+E\left[m_{d-1}^k(t-1)[\frac{q(d-1)p_k(1-\theta)}{D_{t-1}^k}+\frac{(d-1)(2-(2-\theta)q)}{2(t-1)+n_0}]\right].
\end{eqnarray}

However, when $d=1$, the formula is different. Let $x=\frac{m^k_1(t-1)}{2(t-1)+n_0},\, z=\frac{D_{t-1}^k}{2(t-1)+n_0}$.

For the vertex-step at time $t$,
\begin{itemize}
    \item if the new node from group $k$ is connected to  an old node with degree $>1$ or from another group, then $m^k_1(t)$ will be $m^k_1(t-1)+1$. The probability times $m^k_1(t)$ is:
\begin{equation}\label{st01}
q\left[p_k(1-z)\theta+p_k(1-z)(1-\theta)(z-x)\frac{1}{z}+p_k(z-x)\right]\times(m^k_1(t-1)+1).
\end{equation}
    \item if the new node from other groups contacts an old node with degree 1 from group $k$, then $m^k_d(t)$ will be $m^k_d(t-1)-1$. The probability times $m^k_d(t)$ is:
\begin{equation}\label{st02}
q[(1-p_k)x\theta]\times (m^k_1(t-1)-1).
\end{equation}
\item if the $m^k_d(t)$ is the same as $m^k_d(t-1)$. The probability times $m^k_d(t)$ is:
\begin{equation}\label{st03}
q\left[1-(p_k(1-z)\theta+p_k(z-x)(1-z)(1-\theta)\frac{1}{z}+p_k(z-x))-(1-p_k)x\theta\right]\times m^k_1(t-1).
\end{equation}
\end{itemize}

For the edge-step at time $t$,
\begin{itemize}
\item If the new edge connects both two nodes with degree $1$ from group $k$, then $m^k_d(t)$ will be $m^k_d(t-1)-2$. The probability times $m^k_d(t)$ is:
\begin{equation}\label{st05}
(1-q)\left[x(x+(1-z)(1-\theta)\frac{x}{z})\right]\times(m_1^k(t-1)-2).
\end{equation}
\item If one node of the new edge is from group $k$ with degree $1$, the other is not, then $m^k_1(t)$ will be $m^k_1(t-1)-1$. The probability times $m^k_1(t)$ is:
\begin{equation}\label{st07}
(1-q)\left[2x(1-z)\theta+2x((z-x)+(1-z)(1-\theta)\frac{(z-x)}{z})\right]\times (m^k_1(t-1)-1).
\end{equation}
\item If the $m^k_1(t)$ is the same with $m^k_1(t-1)$, the probability times $m^k_1(t)$ is:
\begin{eqnarray}\label{st08}
\nonumber &&(1-q)[(1-x^2(1+(1-(1-z)\theta)\frac{1}{z})\\
&&-2x(1-z)\theta-2x(z-x)(1+(1-z)(1-\theta)\frac{1}{z}))]\times m_1^k(t-1).
\end{eqnarray}
\end{itemize}

By summing Equations (\ref{st01})--(\ref{st08}):
\begin{eqnarray}\label{eq37}
\nonumber  E(m^k_1(t)|{G}_{t-1})=&&m_1^k(t-1)+[qp_k+2(q-1)x][(1-z)\theta+(z-x)(1-z)(1-\theta)\frac{1}{z}+(z-x)]\\
\nonumber &&-q(1-p_k)x\theta-2(1-q)[(x^2(1+(1-z)(1-\theta)\frac{1}{z})]\\
\nonumber =&& m_1^k(t-1)+[qp_k+2(q-1)x][1-\frac{x}{z}(1-\theta)-x\theta]\\
\nonumber &&-qx\theta+qp_kx\theta+2(q-1)x[x(\frac{1-\theta}{z}+\theta)]\\
\nonumber =&& m_1^k(t-1)+[qp_k+2(q-1)x]-qx\theta-qp_k\frac{x}{z}(1-\theta)\\
\nonumber =&& m_1^k(t-1)\left[1-\frac{[2-(2-\theta)q]}{2(t-1)+n_0}-\frac{qp_k(1-\theta)}{D_{t-1}^k}\right]+qp_k.
\end{eqnarray}

Thus, taking the expectation of both sides
\begin{equation}\label{eq38}
\begin{aligned}
\nonumber E(m^k_1(t))=E\left[m_1^k(t-1)\left(1-\frac{2-(2-\theta)q}{2(t-1)+n_0}-\frac{qp_k(1-\theta)}{D_{t-1}^k}\right)\right]+qp_k.
\end{aligned}
\end{equation}

By Theorem \ref{th1}, we get $\left[\frac{2-(2-\theta)q}{2(t-1)+n_0}+\frac{qp_k(1-\theta)}{D_{t-1}^k}\right]\times \left[(t-1)+\frac{n_0}{2}\right]\stackrel{a.s.}{\longrightarrow} 1-\frac{q}{2}$. By Lemma \ref{le3} we have:
\begin{equation}\label{eq39}
\begin{aligned}
\nonumber M^k_1=\lim_{t\to\infty}\frac{E(m^k_1(t))}{t}=\frac{2qp_k}{4-q}.
\end{aligned}
\end{equation}

Next, consider $d\geq 2$, by Equation (\ref{eq36}) and Lemma \ref{le3}:
\begin{equation}\label{eq40}
\begin{aligned}
\nonumber M_d^k=\lim_{t\to\infty}\frac{E(m^k_d(t))}{t}=\frac{(d-1)(2-q)}{2+d(2-q)}M_{d-1}^k.
\end{aligned}
\end{equation}
Thus, $\forall d\geq 2$, we can write,
\begin{equation}
\begin{aligned}
\nonumber M_d^k=\frac{2qp_k}{4-{q}}\prod_{j=2}^d\frac{(j-1)(2-q)}{2+j(2-q)}.
\end{aligned}
\end{equation}

Then we prove the model's power law,
\begin{equation}
\begin{aligned}
\nonumber\forall d\geq 2,\,\frac{M_d^k}{M_{d-1}^k}=\frac{(d-1)(2-q)}{2+d(2-q)}=1-\frac{2+(2-q)}{2+d(2-q)}=1-\frac{{2}/{(2-q)}+1}{d}+O(\frac{1}{d^2}).
\end{aligned}
\end{equation}

Consider if $M^k_d\propto d^{-\beta_k}$.
Then
\begin{equation}
\begin{aligned}
\nonumber\frac{M^k_d}{M^k_{d-1}}=\frac{d^{-\beta_k}}{(d-1)^{-\beta_k}}=(1-\frac{1}{d})^\beta_k=1-\frac{\beta_k}{d}+O(\frac{1}{d^2}).
\end{aligned}
\end{equation}

So, the $\beta_k$ for $M_d^k$ is $1+\frac{2}{2-q}$.
\end{proof}

\begin{proof}[Proof of Theorem \ref{th4}]

Consider the likelihood function of $\{G_t\}_{t=0}^T$, let $P_t^k=\frac{D_{t-1}^k}{2(t-1)+n_0}$:
\begin{eqnarray}\label{mle2}
\nonumber&& L(\psi |\{G_t\}_{t=0}^T)\\
\nonumber=&&\prod_{k=1}^K\prod_{t=1}^T \left[qp_k\left (P_t^k+(1-\theta)\left(1-P_t^k\right)\right)\right]^{\mathbf{1}_{\{v_t=1,g_w=g_u=k\}}}\times\prod_{k=1}^K \prod_{t=1}^T\left[qp_k\left(1-P_t^k\right)\theta\right]^{\mathbf{1}_{\{v_t=1,g_w= k,g_u\neq k\}}}\\
\nonumber &&\times \prod_{k=1}^K\prod_{t=1}^T \left[(1-q)P_t^k\left (P_t^k+(1-\theta)\left(1-P_t^k\right)\right) \right]^{\mathbf{1}_{\{v_t=0,g_w=g_u= k\}}}\\
\nonumber&&\times\prod_{k=1}^K \prod_{t=1}^T\left[(1-q)P_t^k\left(1-P_t^k\right)\theta\right]^{\mathbf{1}_{\{v_t=0,g_w= k,g_u\neq k\}}}.
\end{eqnarray}

The log likelihood function is
\begin{eqnarray}\label{mle3}
\nonumber &&\log L(\psi|\{G(t)\}_{t=0}^T)\\
\nonumber=&&\sum_{t=1}^T\left[\mathbf{1}_{\{v_t=1\}}\log(q)+\mathbf{1}_{\{v_t=0\}}\log(1-q)\right]+\sum_{k=1}^K\sum_{t=1}^T\mathbf{1}_{\{v_t=1,g_w=k\}}\log(p_k)\\
\nonumber &&+\sum_{k=1}^K\sum_{t=1}^T \left[\mathbf{1}_{\{v_t=1,g_w=g_u=k\}} \log\left (P_t^k+(1-\theta)\left(1-P_t^k\right)\right)+{\mathbf{1}_{\{v_t=1,g_w= k,g_u\neq k\}}}\log \left(\left(1-P_t^k\right)\theta\right)\right]\\
\nonumber &&+\sum_{k=1}^K\sum_{t=1}^T{\mathbf{1}_{\{v_t=0,g_w=g_u= k\}}} \log \left(P_t^k\left (P_t^k+(1-\theta)\left(1-P_t^k\right)\right) \right)\\
&&+\sum_{k=1}^K\sum_{t=1}^T{\mathbf{1}_{\{v_t=0,g_w= k,g_u\neq k\}}} \log\left(P_t^k\left(1-P_t^k\right)\theta\right).
\end{eqnarray}

Since  $\sum_{k=1}^K p_k=1$, we can rewrite $p_K=1-\sum_{k=1}^{K-1}p_k$.
The score functions of Equation (\ref{mle3}) for $p_k,\,k=1,\cdots,K-1$, $q$ are as follows:
\begin{equation}\label{score1}
\begin{aligned}
\frac{\partial}{\partial p_k}\log L(\psi|\{G_t\}_{t=0}^T)=\frac{\sum_{t=1}^T\mathbf{1}_{\{v_t=1,g_w=k\}}}{p_k}-\frac{\sum_{t=1}^T\mathbf{1}_{\{v_t=1,g_w=K\}}}{1-\sum_{k=1}^{K-1}p_k}.
\end{aligned}
\end{equation}

\begin{equation}\label{score2}
\begin{aligned}
\frac{\partial}{\partial q}\log L(\psi|\{G_t\}_{t=0}^T)=\frac{\sum_{t=1}^T\mathbf{1}_{\{v_t=1\}}}{q}-\frac{\sum_{t=1}^T\mathbf{1}_{\{v_t=0\}}}{1-q}.
\end{aligned}
\end{equation}

Equations (\ref{score1}) and (\ref{score2}) are monotonically decreasing continuous functions of $p_k\in[0,1]$ and $q\in[0,1]$.

What is more, when $p_k=0,q=0$, Equations (\ref{score1}) and (\ref{score2}) both go to $+\infty$;
and when $p_k=1,q=1$,  Equations (\ref{score1}) and (\ref{score2}) both go to $-\infty$. Therefore, we know that Equations (\ref{score1}) and (\ref{score2}) both have a unique root:
\begin{equation}\label{mle4}
\begin{aligned}
\nonumber\hat{p}_k=\frac{\sum_{t=1}^T \mathbf{1}_{\{v_t=1,g_w=k\}}}{\sum_{t=1}^T \mathbf{1}_{\{v_t=1\}}}.
\end{aligned}
\end{equation}
\begin{equation}\label{mle5}
\begin{aligned}
\nonumber\hat{q}=\frac{\sum_{t=1}^T \mathbf{1}_{\{v_t=1\}}}{T}.
\end{aligned}
\end{equation}

By the law of large numbers, get
\begin{eqnarray}
\nonumber\hat{p}_k =&&\frac{\sum_{t=1}^T \mathbf{1}_{\{v_t=1,g_w=k\}}}{\sum_{t=1}^T \mathbf{1}_{\{v_t=1\}}}\\
\nonumber=&&\frac{\sum_{t=1}^T \mathbf{1}_{\{v_t=1,g_w=k\}}}{q^*T}\times \frac{q^*T}{\sum_{t=1}^T \mathbf{1}_{\{v_t=1\}}}\\
\nonumber \stackrel{a.s.}{\longrightarrow}&& p^*_k.
\end{eqnarray}
\begin{eqnarray}
\nonumber\hat{q}=\frac{\sum_{t=1}^T \mathbf{1}_{\{v_t=1\}}}{T}\stackrel{a.s.}{\longrightarrow} q^*.
\end{eqnarray}

Next, considering the score function for $\theta$, we have
\begin{eqnarray}\label{score3}
\nonumber&& \frac{\partial}{\partial \theta}\log L(\psi|\{G_t\}_{t=0}^T)\\
\nonumber=&&\frac{\partial}{\partial \theta}\log L_2(\theta|\{G_t\}_{t=0}^T)\\
\nonumber =&&-\sum_{k=1}^K\sum_{t=1}^T\frac{\mathbf{1}_{\{v_t=1,g_w=g_u=k\}}(1-P_t^k)}{ P^k_t+(1-\theta)\left(1-P^k_t\right)}-\sum_{k=1}^K\sum_{t=1}^T\frac{\mathbf{1}_{\{v_t=0,g_w=g_u=k\}}(1-P_t^k)}{ P^k_t+(1-\theta)\left(1-P^k_t\right)}\\
 &&+\sum_{k=1}^K\sum_{t=1}^T\frac{\mathbf{1}_{\{v_t=1,g_w= k,g_u\neq k\}}}{\theta}+\sum_{k=1}^K\sum_{t=1}^T\frac{\mathbf{1}_{\{v_t=0,g_w= k,g_u\neq k\}}}{\theta}.
\end{eqnarray}

Equation (\ref{score3}) is a strictly decreasing continuous function of $\theta\in (0,1]$, and Equation (\ref{score3}) $\to \infty$ when $\theta\to 0$.
So, there is a unique maximum point for $\log L_2(\theta|\{G_t\}_{t=0}^T)$ of $\theta\in (0,1]$.

Let $\hat{\theta}=\arg\max_{\theta\in (0,1]} \log L_2(\theta|\{G_t\}_{t=0}^T)$. $\hat{\theta}$ is the unique maximum point for $\log L_2(\theta|\{G_t\}_{t=0}^T)$. So $\forall \theta\neq\hat{\theta}\in (0,1]$, $\log L_2(\hat{\theta}|\{G_t\}_{t=0}^T)> \log L_2(\theta|\{G_t\}_{t=0}^T)$. And it implies
\begin{equation}\label{mle0}
\frac{1}{T}\log L_2(\hat{\theta}|\{G_t\}_{t=0}^T)>\frac{1}{T} \log L_2(\theta|\{G_t\}_{t=0}^T).
\end{equation}

We have
\begin{eqnarray}\label{mle6}
\nonumber &&\frac{\log L_2(\theta|\{G_t\}_{t=0}^T)}{T}\\
\nonumber=&&\sum_{k=1}^K\sum_{t=1}^T \frac{\mathbf{1}_{\{v_t=1,g_w=g_u=k\}}\log\left (P^k_t+(1-\theta)\left(1-P^k_t\right)\right)}{T}\\
\nonumber &&+\sum_{k=1}^K\sum_{t=1}^T \frac{{\mathbf{1}_{\{v_t=1,g_w= k,g_u\neq k\}}}\log \left(\left(1-P^k_t\right)\theta\right)}{T}\\
\nonumber &&+ \sum_{k=1}^K\sum_{t=1}^T\frac{{\mathbf{1}_{\{v_t=0,g_w=k,g_u= k\}}}\log \left(P^k_t\left (P^k_t+(1-\theta)\left(1-P^k_t\right)\right) \right)}{T}\\
\nonumber &&+\sum_{k=1}^K\sum_{t=1}^T \frac{{\mathbf{1}_{\{v_t=0,g_w= k,g_u\neq k\}}}\log\left(P^k_t\left(1-P^k_t\right)\theta\right)}{T}\\
\nonumber =&&\sum_{k=1}^K \frac{\sum_{t=1}^T\mathbf{1}_{\{v_t=1,g_w=g_u=k\}}\log\left (P^k_t+(1-\theta)\left(1-P^k_t\right)\right)}{\sum_{t=1}^T \mathbf{1}_{\{v_t=1,g_w=g_u=k\}}}\times \frac{\sum_{t=1}^T \mathbf{1}_{\{v_t=1,g_w=g_u=k\}}}{T}\\
\nonumber &&+\sum_{k=1}^K \frac{\sum_{t=1}^T{\mathbf{1}_{\{v_t=1,g_w=k,g_u\neq k\}}}\log \left(\left(1-P^k_t\right)\theta\right)}{\sum_{t=1}^T \mathbf{1}_{\{v_t=1,g_w=k,g_u\neq k\}}}\times \frac{\sum_{t=1}^T \mathbf{1}_{\{v_t=1,g_w=k,g_u\neq k\}}}{T}\\
\nonumber &&+ \sum_{k=1}^K\frac{\sum_{t=1}^T{\mathbf{1}_{\{v_t=0,g_w=g_u= k\}}}\log \left(P^k_t\left (P^k_t+(1-\theta)\left(1-P^k_t\right)\right) \right)}{\sum_{t=1}^T \mathbf{1}_{\{v_t=0,g_w=g_u=k\}}}\times \frac{\sum_{t=1}^T \mathbf{1}_{\{v_t=0,g_w=g_u=k\}}}{T}\\
 &&+\sum_{k=1}^K \frac{\sum_{t=1}^T{\mathbf{1}_{\{v_t=0,g_w= k,g_u\neq k\}}}\log\left(P^k_t\left(1-P^k_t\right)\theta\right)}{\sum_{t=1}^T \mathbf{1}_{\{v_t=0,g_w=k,g_u\neq k\}}}\times \frac{\sum_{t=1}^T \mathbf{1}_{\{v_t=0,g_w=k,g_u\neq k\}}}{T}.
\end{eqnarray}

By LLN and Theorem \ref{th1}, get
\begin{eqnarray}\label{mle8}
\nonumber \lim_{T\to\infty}\frac{\log L_2(\theta|\{G_t\}_{t=0}^T)}{T}=&&\sum_{k=1}^K\log(p^*_k+(1-\theta)(1-p^*_k))\times (p^*_k)^2 \\
&&+\sum_{k=1}^K\log \left(\left(1-p^*_k\right)\theta\right)\times p^*_k(1-p^*_k).
\end{eqnarray}

Consider a random variable $Y=(k_1,k_2)$ such that $k_1,k_2\in \{1,\cdots,K\}$. The probability distribution of $Y$ is
\begin{equation}\label{pdf1}
\begin{aligned}
P(y;\theta)=\begin{cases}
p^*_{k_1}p^*_{k_2}\theta , & \text{if } k_1\neq k_2\\
(p^*_k)^2+p^*_k(1-p^*_k)(1-\theta), &\text{if } k_1= k_2=k.
\end{cases}
\end{aligned}
\end{equation}

For any other $\theta\neq \theta^*\in (0,1]$,
\begin{eqnarray}\label{mle9}
\nonumber E_{\theta^*}\log [P(Y;\theta^*)/P(Y;\theta)]=&&-E_{\theta^*}\log [P(Y;\theta)/P(Y;\theta^*)]\\
\nonumber \geq &&-\log[E_{\theta^*}(P(Y;\theta)/f(Y;\theta^*))]\\
 =&&-\log[1]=0.
\end{eqnarray}

Equation (\ref{mle9}) implies that
\begin{equation}\label{mle10}
\forall \theta\neq \theta^*\in (0,1],\, E_{\theta^*}\log [P(Y;\theta^*)]> E_{\theta^*}\log [P(Y;\theta)].
\end{equation}

Moreover, we can prove $E_{\theta^*}\log [P(Y;\theta^*)]=\lim_{T\to\infty}\frac{\log L_2(\theta^*|\{G_t\}_{t=0}^T)}{T}$.
Thus, by Equation (\ref{mle10}), when $T$ is large enough, we know  $\frac{\log L_2(\theta^*|\{G_t\}_{t=0}^T)}{T}> \frac{\log L_2(\theta|\{G_t\}_{t=0}^T)}{T},\,\forall \theta\neq \theta^*\in [0,1]$. Combining this result with Equation (\ref{mle0}), get $P(\lim_{T\to\infty} \hat{\theta}=\theta^*)=1$,
which means
\begin{equation}
\nonumber\lim_{T\to\infty} \hat{\theta}\stackrel{a.s.}{\longrightarrow} \theta^*.
\end{equation}
\end{proof}

\begin{proof}[Proof of Theorem \ref{th5}]
Let $l(\theta|\{G_t\}_{t=0}^T)=\frac{\partial \log L_2(\theta|\{G_t\}_{t=0}^T)}{\partial \theta}$ and $P_t^k=\frac{D_{t-1}^k}{2(t-1)+n_0}$.
By Equation (\ref{score3}), we get
\begin{eqnarray}\label{I2}
\nonumber l(\theta|\{G_t\}_{t=0}^T)=&&-\sum_{k=1}^K\sum_{t=1}^T\frac{\mathbf{1}_{\{v_t=1,g_w=g_u=k\}}(1-P_t^k)}{ P^k_t+(1-\theta)\left(1-P^k_t\right)}-\sum_{k=1}^K\sum_{t=1}^T\frac{\mathbf{1}_{\{v_t=0,g_w=g_u=k\}}(1-P_t^k)}{ P^k_t+(1-\theta)\left(1-P^k_t\right)}\\
\nonumber&&+\sum_{k=1}^K\sum_{t=1}^T\frac{\mathbf{1}_{\{v_t=1,g_w= k,g_u\neq k\}}}{\theta}+\sum_{k=1}^K\sum_{t=1}^T\frac{\mathbf{1}_{\{v_t=0,g_w= k,g_u\neq k\}}}{\theta}.
\end{eqnarray}

Further, we can get:
\begin{eqnarray}\label{I3}
\nonumber l'(\theta|\{G_t\}_{t=0}^T)=&&\frac{\partial l'(\theta|\{G_t\}_{t=0}^T)}{\partial\theta}\\
\nonumber =&&-\sum_{k=1}^K\sum_{t=1}^T\frac{\mathbf{1}_{\{v_t=1,g_w=g_u=k\}}\left(1-P^k_t\right)^2}{ [P^k_t+(1-\theta)\left(1-P^k_t\right)]^2}-\sum_{k=1}^K\sum_{t=1}^T\frac{\mathbf{1}_{\{v_t=0,g_w=g_u=k\}} \left(1-P^k_t\right)^2}{[P^k_t+(1-\theta)\left(1-P^k_t\right)]^2}
\\
 \nonumber&&-\sum_{k=1}^K\sum_{t=1}^T\frac{\mathbf{1}_{\{v_t=1,g_w= k,g_u\neq k\}}}{\theta^2}-\sum_{k=1}^K\sum_{t=1}^T\frac{\mathbf{1}_{\{v_t=0,g_w= k,g_u\neq k\}}}{\theta^2}.
\end{eqnarray}

$\theta^*$ is the true parameter for the network $\{G_t\}_{t=0}^T$ by KPA model, and taking the Taylor Expansion of $l(\theta|\{G_t\}_{t=0}^T)$ around $\theta^*$, we can get:
\begin{equation}\label{te}
\begin{aligned}
l(\theta|\{G_t\}_{t=0}^T)=l(\theta^*|\{G_t\}_{t=0}^T)+({\theta}-\theta^*) l'(\theta^*|\{G_t\}_{t=0}^T)+({\theta}-\theta^*)^2 l'(\theta^{**}|\{G_t\}_{t=0}^T),
\end{aligned}
\end{equation}
where $\theta^{**}=\theta^*+\mathbf{\xi}\circ(\theta-\theta^*)$, $\xi \in [0,1]$. $l(\hat{\theta}|\{G_t\}_{t=0}^T)=0$ when $T$ is large enough.
Equation (\ref{te}) implies
\begin{equation}\label{te2}
\begin{aligned}
\sqrt{T}(\hat{\theta}-\theta^*)= -\frac{\frac{\sqrt{T}l(\theta^*|\{G(t)\}_{t=0}^T)}{T}}{\frac{l'(\theta^*|\{G(t)\}_{t=0}^T)}{T}+\frac{l'(\theta^{**}|\{G(t)\}_{t=0}^T)(\hat{\theta}-\theta^*)}{T}}.
\end{aligned}
\end{equation}

Consider the random variable $Y$ and its probability distribution $P(\cdot;\theta^*)$ defined by Equation (\ref{pdf1}). 
\begin{eqnarray}\label{I}
\nonumber I(\theta^*)= &&E_{\theta^*}\left[-\frac{\partial^2 \log P(Y;\theta^*)}{\partial(\theta^*)^2}\right]=E_{\theta^*}\left[\frac{P'(Y;\theta^*)^2}{P(Y;\theta^*)^2}\right]\\
=&&\sum_{y}\frac{P'(y;\theta^*)^2}{P(y;\theta^*)}=\sum_{k=1}^K \left[\frac{p^*_{k}(1-p^*_{k})}{\theta^*}+\frac{p^*_{k}(1-p^*_{k})^2}{p^*_k+(1-p^*_k)(1-\theta^*)}\right].
\end{eqnarray}

By the martingale convergence theorem, martingale difference central limit theorem and Equation (\ref{I}), we have that when $T$ tends to infinity,
\begin{equation}\label{mletheta1}
\hat{\theta}\stackrel{a.s.}{\longrightarrow}\theta^*,\ \ \ \sqrt{T}\frac{l(\theta^*|\{G(t)\}_{t=0}^T)}{T}\xrightarrow{d} N(0,I(\theta^*)),
\end{equation}
\begin{equation}\label{mletheta2}
-\frac{l'(\theta^*|\{G(t)\}_{t=0}^T)}{T} \stackrel{a.s.}{\longrightarrow} I(\theta^*),
\end{equation}
and
\begin{equation}\label{mletheta3}
-\frac{l'(\theta^{**}|\{G(t)\}_{t=0}^T)}{T} \stackrel{a.s.}{\longrightarrow} \sum_{k=1}^K \left[\frac{\theta^* p^*_{k}(1-p^*_{k})}{(\theta^{**})^2}+\frac{(p^*_k+(1-p^*_k)(1-\theta^*))p^*_{k}(1-p^*_{k})^2}{(p^*_k+(1-p^*_k)(1-\theta^{**}))^2}\right].
\end{equation}
By Equations (\ref{mletheta1})--(\ref{mletheta3}), it follows that when $T$ tends to infinity,
\begin{equation}
\begin{aligned}
\nonumber\sqrt{T}(\hat{\theta}-\theta^*)\xrightarrow{d} N(0,\frac{1}{I(\theta^*)}).
\end{aligned}
\end{equation}

Let $\mathbf{p}=(p_1,\cdots,p_{K-1})'$ and $p_K=1-\sum_{k=1}^{K-1}p_k$,
we define
\begin{eqnarray}
\nonumber\log L_1(\mathbf{p}|\{G_t\}_{t=0}^T)=\sum_{k=1}^{K} \sum_{t=1}^T \mathbf{1}_{\{v_t=1,g_w=k\}}\log(p_k).
\end{eqnarray}

We also define
\begin{equation}
\nonumber l(\mathbf{p}|\{G_t\}_{t=0}^T)=\frac{\partial \log L_1(\mathbf{p}|\{G_t\}_{t=0}^T)}{\partial\mathbf{p}}=\begin{pmatrix}
\frac{\sum_{t=1}^T\mathbf{1}_{\{v_t=1,g_w=1\}}}{p_1}-\frac{\sum_{t=1}^T\mathbf{1}_{\{v_t=1,g_w=K\}}}{1-\sum_{k=1}^{K-1}p_k} \\
\vdots\\
\frac{\sum_{t=1}^T\mathbf{1}_{\{v_t=1,g_w=K-1\}}}{p_{K-1}}-\frac{\sum_{t=1}^T\mathbf{1}_{\{v_t=1,g_w=K\}}}{1-\sum_{k=1}^{K-1}p_k}
\end{pmatrix}.
\end{equation}

Further,
\begin{eqnarray}
\nonumber&&l'(\mathbf{p}|\{G_t\}_{t=0}^T)= \frac{\partial \log l(\mathbf{p}|\{G_t\}_{t=0}^T)}{\partial\mathbf{p}}\\
\nonumber =&&\begin{bmatrix}
-\frac{\sum_{t=1}^T\mathbf{1}_{\{v_t=1,g_w=1\}}}{p_1^2}-\frac{\sum_{t=1}^T\mathbf{1}_{\{v_t=1,g_w=K\}}}{(1-\sum_{k=1}^{K-1}p_k)^2}& \cdots &-\frac{\sum_{t=1}^T\mathbf{1}_{\{v_t=1,g_w=K\}}}{(1-\sum_{k=1}^{K-1}p_k)^2}    \\
\vdots&  &\vdots\\
\nonumber-\frac{\sum_{t=1}^T\mathbf{1}_{\{v_t=1,g_w=K\}}}{(1-\sum_{k=1}^{K-1}p_k)^2}&\cdots& -\frac{\sum_{t=1}^T\mathbf{1}_{\{v_t=1,g_w=K-1\}}}{p_{K-1}^2}-\frac{\sum_{t=1}^T\mathbf{1}_{\{v_t=1,g_w=K\}}}{(1-\sum_{k=1}^{K-1}p_k)^2}
\end{bmatrix}.
\end{eqnarray}

Let $\mathbf{p}^*$ be the true parameters for KPA model $\{G_t\}_{t=0}^T$. Then, using a multivariate Taylor expansion of $l(\mathbf{p}|\{G_t\}_{t=0}^T)$ round $\mathbf{p}^*$,
we can get:
\begin{eqnarray}\label{mlep0}
l(\mathbf{p}|\{G_t\}_{t=0}^T)= l(\mathbf{p}^*|\{G_t\}_{t=0}^T)+l'(\mathbf{p}^{**}|\{G_t\}_{t=0}^T)(\mathbf{p}-\mathbf{p}^*),
\end{eqnarray}
where $\mathbf{p}^{**}=\mathbf{p}^*+\boldsymbol{\xi}\circ(\mathbf{p}-\mathbf{p}^*)$, $\boldsymbol{\xi} \in [0,1]^{K-1}$.

Set $\mathbf{p}=\hat{\mathbf{p}}$ by the MLE, $\mathbf{p}^{**}=\mathbf{p}^*+\boldsymbol{\xi}\circ(\hat{\mathbf{p}}-\mathbf{p}^*)$. According to Theorem \ref{th4} and by martingale difference central limit theorem, when $T$ goes to infinity,
\begin{eqnarray}\label{mlep1}
\mathbf{p}^{**}\stackrel{a.s.}{\longrightarrow} \mathbf{p}^*,\ \ \ \sqrt{T}\frac{l(\mathbf{p}^{*}|\{G(t)\}_{t=0}^T)}{T}\xrightarrow{d}N(0,I(\mathbf{p}^*)),
\end{eqnarray}
and
\begin{eqnarray}\label{mlep2}
-\frac{l'(\mathbf{p}^{**}|\{G(t)\}_{t=0}^T)}{T}\stackrel{a.s.}{\longrightarrow}I(\mathbf{p}^*),
\end{eqnarray}
where
\begin{equation}\label{mlep3}
I(\mathbf{p}^*)=q^*\times\begin{bmatrix}
\frac{1-\sum_{k=2}^{K-1}p^*_k}{p^*_1(1-\sum_{k=1}^{K-1}p^*_k)}& \frac{1}{1-\sum_{k=1}^{K-1}p^*_k}& \cdots &\frac{1}{1-\sum_{k=1}^{K-1}p^*_k}\\ \frac{1}{1-\sum_{k=1}^{K-1}p^*_k}&
\frac{1-\sum_{k\neq 2}^{K-1}p^*_k}{p^*_2(1-\sum_{k= 1}^{K-1}p^*_k)}& \cdots &\frac{1}{1-\sum_{k=1}^{K-1}p^*_k}\\
\vdots & &\cdots & \vdots \\
 \frac{1}{1-\sum_{k=1}^{K-1}p^*_k}& \cdots &\frac{1}{1-\sum_{k=1}^{K-1}p^*_k}&\frac{1-\sum_{k=1}^{K-2}p^*_k}{p^*_{K-1}(1-\sum_{k=1}^{K-1}p^*_k)}\\
\end{bmatrix}.
\end{equation}

By Equations (\ref{mlep0})--(\ref{mlep3}), then when $T$ goes to infinity,
\begin{equation}
\begin{aligned}
\sqrt{T}(\hat{\mathbf{p}}-\mathbf{p}^*)'\xrightarrow{d} N(0,{I(\mathbf{p}^*)}^{-1}).
\end{aligned}
\end{equation}

For $q$, we have
\begin{eqnarray}
\nonumber \log L_3(q|\{G_t\}_{t=0}^T)=\sum_{t=1}^T [\mathbf{1}_{\{v_t=1\}}\log(q)+\mathbf{1}_{\{v_t=0\}}\log(1-q)].
\end{eqnarray}
\begin{eqnarray}
\nonumber l(q|\{G_t\}_{t=0}^T)=\frac{\partial}{\partial q}\log L(\psi|\{G_t\}_{t=0}^T)=\frac{\sum_{t=1}^T\mathbf{1}_{\{v_t=1\}}}{q}-\frac{\sum_{t=1}^T\mathbf{1}_{\{v_t=0\}}}{1-q}.
\end{eqnarray}
\begin{eqnarray}
\nonumber l'(q|\{G_t\}_{t=0}^T)=\frac{\partial}{\partial q}l'(q|\{G_t\}_{t=0}^T)=-\frac{\sum_{t=1}^T\mathbf{1}_{\{v_t=1\}}}{q^2}-\frac{\sum_{t=1}^T\mathbf{1}_{\{v_t=0\}}}{(1-q)^2}.
\end{eqnarray}

Let $q^*$ be the true parameter, and using a Taylor expansion of $l(q|\{G_t\}_{t=0}^T)$ around $q^*$, we can get:
\begin{equation}\label{mleq0}
l(q|\{G_t\}_{t=0}^T)=l(q^*|\{G_t\}_{t=0}^T)+(q-q^*)l'(q^*|\{G_t\}_{t=0}^T)+({q}-q^*)^2 l'(q^{**}|\{G_t\}_{t=0}^T).
\end{equation}

By Theorem \ref{th4} and martingale difference central limit theorem, when $T$ goes to infinity,
\begin{eqnarray}\label{mleq1}
\sqrt{T}\frac{l(q^{*}|\{G_t\}_{t=0}^T)}{T}\xrightarrow{d} N(0,\frac{1}{q^*(1-q^*)}),\ \ \ \hat{q}
\stackrel{a.s.}{\longrightarrow} q^*,
\end{eqnarray}
\begin{eqnarray}\label{mleq2}
\nonumber\frac{l'(q^{**}|\{G_t\}_{t=0}^T)}{T}= &&-\frac{\sum_{t=1}^T\mathbf{1}_{\{v_t=1\}}}{(q^{**})^2T}-\frac{\sum_{t=1}^T\mathbf{1}_{\{v_t=0\}}}{(1-q^{**})^2T}\\
\stackrel{a.s.}{\longrightarrow}&& -\frac{q^*}{(q^{**})^2}-\frac{1-q^*}{(1-q^{**})^2},
\end{eqnarray}
and
\begin{eqnarray}\label{mleq3}
-\frac{l'(q^{*}|\{G_t\}_{t=0}^T)}{T}
\stackrel{a.s.}{\longrightarrow}\frac{1}{q^*(1-q^*)}.
\end{eqnarray}

Let $q=\hat{q}$, by Equations (\ref{mleq0})-- (\ref{mleq3}),  we have that when $T$ goes to infinity,
\begin{equation}
\begin{aligned}
\nonumber\sqrt{T}(\hat{q}-q^*)\xrightarrow{d} N(0,{q^*(1-q^*)}).
\end{aligned}
\end{equation}

\end{proof}

\begin{proof}[Proof of Theorem \ref{th6}]
By Slutsky's theorem, we have $\forall k \in \{1,\cdots,K\}$,
\begin{eqnarray}
\nonumber
\tilde{p}_k=&&\frac{V_T^k}{V_T}=\frac{n^k_0+\sum_{t=1}^T \mathbf{1}_{\{v_t=1,g_w=k\}}}{n_0+\sum_{t=1}^T \mathbf{1}_{\{v_t=1\}}}\\
\nonumber=&&\frac{n^k_0}{n_0+\sum_{t=1}^T \mathbf{1}_{\{v_t=1\}}}+\hat{p}_k \\
\nonumber \stackrel{a.s.}{\longrightarrow} &&p^*_k.
\end{eqnarray}

\begin{eqnarray}
\nonumber\tilde{q}=&&\frac{V_T}{E_T}=\frac{n_0+\sum_{t=1}^T \mathbf{1}_{\{v_t=1\}}}{e_0+T}\\
\nonumber=&&\frac{T}{e_0+T}\hat{q}+\frac{n_0}{e_0+T}\\
\nonumber \stackrel{a.s.}{\longrightarrow}&&q^*.
\end{eqnarray}
where $e_0,n_0$ means the number of edges and vertices in the initial $G_0$, and $n_0^k$ means the number of vertices from group $k$ in the $G_0$.

Considering $\tilde{\theta}$, we have
\begin{eqnarray}
\nonumber \frac{L_T(\theta|G_T)}{T}= &&\sum_{k=1}^K\left[\frac{E^{k,1}_{T}}{T}\log (\frac{D_{T}^k}{2E_T}(\frac{D_{T}^k}{2E_T}+(1-\theta)(1-\frac{D_{T}^k}{2E_T})))\right]\\
\nonumber &&+\sum_{k=1}^K\left[\frac{E^{k,0}_{T}}{T}\log(\frac{D_{T}^k}{2E_T}\theta(1-\frac{D_{T}^k}{2E_T}))\right].
\end{eqnarray}

By LLN and Theorem \ref{th4}, we have
\begin{eqnarray}\label{snap1}
\nonumber &&\lim_{T\to\infty}\frac{L_T(\theta|G(T))}{T}\\
\nonumber\stackrel{a.s.}{\longrightarrow}&&\sum_{k=1}^K[((p^*_k)^2+p^*_k(1-p^*_k)(1-\theta^*))\log((p_k^*)^2+p^*_k(1-p^*_k)(1-\theta))\\
\nonumber&&+(p^*_k(1-p^*_k)\theta^*)\log(p^*_k(1-p^*_k)\theta)]\\
=&&E_{\theta^*}[\log P(Y;\theta)].
\end{eqnarray}

Equation (\ref{snap1}) implies $\frac{1}{T}L_T(\theta^*|G_T)>\frac{1}{T}L_T(\theta|G_T),\,\forall \theta\neq\theta^*\in (0,1] $ when $T$ is large enough. So we can get $\tilde{\theta}=\arg\max_{\theta\in (0,1]} L_T(\theta|G(T))\stackrel{a.s.}{\longrightarrow} \theta^*$.
\end{proof}

\begin{proof}[Proof of Theorem \ref{changep}]

\begin{equation}
\begin{aligned}
\nonumber\forall \epsilon>0,\, P\left(\frac{|\hat{\tau}-\tau^*|}{T}>\epsilon \right)=P\left({|\hat{\tau}-\tau^*|}>\epsilon T \right).
\end{aligned}
\end{equation}

Let

$\hat{\theta}_{1,\hat{\tau}}=\arg\max_{\theta\in(0,1]}\log L_2(\theta|\{G_t\}_{t=0}^{\hat{\tau}})$, $\hat{\theta}_{2,\hat{\tau}}=\arg\max_{\theta\in(0,1]}\log L_2(\theta|\{G_t\}_{t=\hat{\tau}}^T)$.

 $\hat{\theta}_{1,\tau^*}=\arg\max_{\theta\in(0,1]}\log L_2(\theta|\{G_t\}_{t=0}^{\tau^*})$, $\hat{\theta}_{2,\tau^*}=\arg\max_{\theta\in(0,1]}\log L_2(\theta|\{G_t\}_{t=\tau^*}^T)$.

By Theorem \ref{th4}, we can get $\hat{\theta}_{1,\tau^*}\stackrel{a.s.}{\longrightarrow} \theta^*_1$ and $\hat{\theta}_{2,\tau^*}\stackrel{a.s.}{\longrightarrow} \theta^*_2$.

First, we consider the case $\hat{\tau}-\tau^*>\epsilon T$. Divide the time range into three intervals $[0,T]=[0,\tau^*]\cup[\tau^*+1,\hat{\tau}]\cup[\hat{\tau}+1,T]$.  
\begin{eqnarray}\label{changep1}
\nonumber\frac{\log L_2(\theta|\{G_t\}_{t=0}^{\hat{\tau}})}{\hat{\tau}}= \frac{\tau^*}{\hat{\tau}}\frac{\log L_2(\theta|\{G_t\}_{t=0}^{{\tau^*}})}{\tau^*}+\frac{\hat{\tau}-\tau^*}{\hat{\tau}}\frac{\log L_2(\theta|\{G_t\}_{t=\tau^*}^{\hat{\tau}})}{\hat{\tau}-\tau^*}.
\end{eqnarray}

Furthermore, let $\alpha=\frac{\tau^*}{\hat{\tau}}$, under Assumption \ref{a5}, we have $\alpha\in [\frac{c}{1-c},1-\frac{\epsilon}{1-c}]$.
\begin{eqnarray}
\nonumber\Theta_0=\left\{{\theta^*_{\alpha}}=\arg\max_{\theta\in(0,1]}[\alpha E_{\theta_1^*}(\log(P(Y;\theta))) +(1-\alpha)E_{\theta_2^*}(\log(P(Y;\theta)))]:\alpha\in [\frac{c}{1-c},1-\frac{\epsilon}{1-c}] \right\}.
\end{eqnarray}

For $\forall \theta^*\in (0,1]$,
\begin{eqnarray}\label{score4}
\nonumber &&\frac{\partial E_{\theta^*}(\log(P(Y;\theta)))}{\partial \theta}\\
\nonumber=&&-\sum_{k=1}^K\sum_{t=1}^T\left[\frac{[p_k+(1-\theta^*)\left(1-p_k\right)](1-p_k)p_k}{ p_k+(1-\theta)\left(1-p_k\right)}-\frac{\theta^*(1-p_k)p_k}{\theta}\right]\\
= &&\begin{cases}
0,  & \text{if }\theta=\theta^* \\
\infty, & \text{if }\theta\to 0\\
-\sum_{k=1}^K\sum_{t=1}^T {(1-\theta^*)\left(1-p_k\right)}, & \text{if }\theta=1.
\end{cases}
\end{eqnarray}

For $\forall \alpha\in [\frac{c}{1-c},1-\frac{\epsilon}{1-c}]$, Equation (\ref{score4}) implies that the derivative function of $\alpha E_{\theta_1^*}[\log(P(Y;\theta))] +(1-\alpha)E_{\theta_2^*}[\log(P(Y;\theta))]$ is strictly decreasing and has a unique root. So $\alpha E_{\theta_1^*}[\log(P(Y;\theta))] +(1-\alpha)E_{\theta_2^*}[\log(P(Y;\theta))]$ has a unique maximum point.
And $\alpha \frac{\log L_2(\theta|\{G_t\}_{t=0}^{{\tau^*}})}{\tau^*}+(1-\alpha) \frac{\log L_2(\theta|\{G_t\}_{t=\tau^*}^{\hat{\tau}})}{\hat{\tau}-\tau^*}$ also has a unique maximum point by Equation (\ref{score3}).

For $\forall \alpha\in [\frac{c}{1-c},1-\frac{\epsilon}{1-c}]$, let $\hat{\theta}_{\alpha}=\arg\max_{\theta\in(0,1]}\left[\alpha \frac{\log L_2(\theta|\{G_t\}_{t=0}^{{\tau^*}})}{\tau^*}+(1-\alpha) \frac{\log L_2(\theta|\{G_t\}_{t=\tau^*}^{\hat{\tau}})}{\hat{\tau}-\tau^*}\right]$. By Equation (\ref{mle8})--(\ref{mle10}), we can get $\hat{\theta}_{\alpha}\stackrel{a.s.}{\longrightarrow} \theta^*_{\alpha}$, where ${\theta^*_{\alpha}}=\arg\max_{\theta\in(0,1]}[\alpha E_{\theta_1^*}(\log(P(Y;\theta))) +(1-\alpha)E_{\theta_2^*}(\log(P(Y;\theta)))]$. Furthermore, we can get the larger $\alpha$ is, the closer the $\theta^*_\alpha$ is to $\theta^*_1$. So the distance from point $\theta^*_1$ to set $\Theta_0$ is $\left|\theta^*_1-\theta^*_{\{1-\frac{\epsilon}{1-c}\}}\right|$, the distance from point $\theta^*_2$ to set $\Theta_0$ is $\left|\theta^*_2-\theta^*_{\{\frac{c}{1-c}\}}\right|$.

Let $\varepsilon=\min\left(\frac{\left|\theta^*_1-\theta^*_{\{1-\frac{\epsilon}{1-c}\}}\right|}{2},\frac{\left|\theta^*_2-\theta^*_{\{\frac{c}{1-c}\}}\right|}{2}\right)$, when $T$ is large enough, we have $\max\left(\left|\hat{\theta}_{\{1-\frac{\epsilon}{1-c}\}}-\theta^*_{\{1-\frac{\epsilon}{1-c}\}}\right|,\left|\hat{\theta}_{\{\frac{c}{1-c}\}}-\theta^*_{\{\frac{c}{1-c}\}}\right|\right)<\varepsilon$ with probability $1-o_p(1)$.
\begin{eqnarray}
\nonumber\Theta_\varepsilon=\left\{\theta: \theta\in (0,1] \text{ and } \exists \theta_\alpha \in \Theta_{0} \text{ that } \left|{\theta}-\theta^*_{\alpha}\right|<\varepsilon  \right\}.
\end{eqnarray}

It's easy to figure out that $\theta_1^*,\theta_2^*\notin \Theta_\varepsilon$ but $\hat{\theta}_{\{1-\frac{\epsilon}{1-c}\}},\hat{\theta}_{\{\frac{c}{1-c}\}}\in \Theta_\varepsilon$. Back to (\ref{score4}) and $\hat{\theta}_{1,\hat{\tau}}=\arg\max_{\theta\in(0,1]}\frac{\log L_2(\theta|\{G_t\}_{t=0}^{\hat{\tau}})}{\hat{\tau}}$, when $T$ is large enough, $\hat{\theta}_{1,\hat{\tau}}$ is closer to $\theta_2^*$ than $\hat{\theta}_{\{1-\frac{\epsilon}{1-c}\}}$ and closer to $\theta_1^*$ than $\hat{\theta}_{\{\frac{c}{1-c}\}}$ with probability $1-o_p(1)$. So we get $\hat{\theta}_{1,\hat{\tau}}\in \Theta_\varepsilon$ with probability $1-o_p(1)$.

By Theorem \ref{th4}, $\hat{\theta}_{1,\tau^*}\stackrel{a.s.}{\longrightarrow} \theta^*_1$, $\hat{\theta}_{2,\tau^*}\stackrel{a.s.}{\longrightarrow} \theta^*_2$ and $\hat{\theta}_{2,\hat{\tau}}\stackrel{a.s.}{\longrightarrow} \theta_2^*$ as $T$ tends to infinity. We get the following inequalities with probability $1-o_p(1)$ when $T$ is large enough for $\forall \varepsilon_1>0$:
\begin{eqnarray}\label{changepdf2}
\nonumber\left|\frac{\log{L_2}(\hat{\theta}_{1,\tau^*}|\{G(t)\}_{t=0}^{{\tau^*}})}{\tau^*}- E_{\theta_1^*}[\log(P(Y;\theta_1^*))]\right|\leq \varepsilon_1,\,  \left|\frac{\log{L_2}(\hat{\theta}_{1,\hat{\tau}}|\{G(t)\}_{t=0}^{{\tau}^*})}{\tau^*}-E_{\theta_1^*}[\log(P(Y;\hat{\theta}_{1,\hat{\tau}}))]\right|\leq \varepsilon_1\\
\nonumber\left|\frac{\log{L_2}(\hat{\theta}_{2,\tau^*}|\{G(t)\}_{t={\tau^*}}^{\hat{\tau}})}{\hat{\tau}-\tau^*}-E_{\theta_2^*}[\log(P(Y;\theta_2^*))]\right|\leq \varepsilon_1,\,
\left|\frac{\log{L_2}(\hat{\theta}_{1,\hat{\tau}}|\{G(t)\}_{t={\tau^*}}^{\hat{\tau}})}{\hat{\tau}-\tau^*}
-E_{\theta_2^*}[\log(P(Y;\hat{\theta}_{1,\hat{\tau}}))]\right|\leq \varepsilon_1\\
 \left|\frac{\log{L_2}(\hat{\theta}_{2,\tau^*}|\{G(t)\}_{t=\hat{\tau}}^{T})}{T-\hat{\tau}}-E_{\theta_2^*}[\log(P(Y;\theta_2^*))]\right|\leq \frac{\varepsilon_1}{2},\,\left|\frac{\log{L_2}(\hat{\theta}_{2,\hat{\tau}}|\{G(t)\}_{t=\hat{\tau}}^{T})}{T-\hat{\tau}}-E_{\theta_2^*}[\log(P(Y;\theta_2^*))]\right|\leq \frac{\varepsilon_1}{2}.
\end{eqnarray}

Let $\varepsilon_1=\min\left(\frac{E_{\theta_2^*}[\log(P(Y;\theta_2^*))]-\max_{\theta\in \Theta_\varepsilon}E_{\theta_2^*}[\log(P(Y;\theta))]}{4},\frac{E_{\theta_1^*}[\log(P(Y;\theta_1^*))]-\max_{\theta\in \Theta_\varepsilon}E_{\theta_1^*}[\log(P(Y;\theta))]}{4}\right)$. Equation (\ref{changepdf2}) implies:

\begin{eqnarray}\label{changepdf2.1}
\nonumber&&\frac{\log{L_2}(\hat{\theta}_{1,\tau^*}|\{G(t)\}_{t=0}^{{\tau^*}})}{\tau^*}- \frac{\log{L_2}(\hat{\theta}_{1,\hat{\tau}}|\{G(t)\}_{t=0}^{{\tau}^*})}{\tau^*}\geq \frac{E_{\theta_1^*}[\log(P(Y;\theta_1^*))]-\max_{\theta\in \Theta_\varepsilon}E_{\theta_1^*}[\log(P(Y;\theta))]}{2}\\
\nonumber&&\frac{\log{L_2}(\hat{\theta}_{2,\tau^*}|\{G(t)\}_{t={\tau^*}}^{\hat{\tau}})}{\hat{\tau}-\tau^*}-
\frac{\log{L_2}(\hat{\theta}_{1,\hat{\tau}}|\{G(t)\}_{t={\tau^*}}^{\hat{\tau}})}{\hat{\tau}-\tau^*}\geq
\frac{E_{\theta_2^*}[\log(P(Y;\theta_2^*))]-\max_{\theta\in \Theta_\varepsilon}E_{\theta_2^*}[\log(P(Y;\theta))]}{2}\\
&&\left|\frac{\log{L_2}(\hat{\theta}_{2,\tau^*}|\{G(t)\}_{t=\hat{\tau}}^{T})}{T-\hat{\tau}}-
\frac{\log{L_2}(\hat{\theta}_{2,\hat{\tau}}|\{G(t)\}_{t=\hat{\tau}}^{T})}{T-\hat{\tau}}\right|\leq  \varepsilon_1.
\end{eqnarray}

$\hat{\tau}-\tau^*>0$ means that
\begin{eqnarray}
\nonumber{\log{L_2}(\hat{\theta}_{1,\hat{\tau}}|\{G(t)\}_{t=0}^{\hat{\tau}})}+ \log{L_2}(\hat{\theta}_{2,\hat{\tau}}|\{G(t)\}_{t=\hat{\tau}}^T)
\geq {\log{L_2}(\hat{\theta}_{1,\tau^*}|\{G(t)\}_{t=0}^{\tau^*})}+ \log{L_2}(\hat{\theta}_{2,\tau^*}|\{G(t)\}_{t=\tau^*}^T).
\end{eqnarray}

Let $A=\left\{[{\log{L_2}(\hat{\theta}_{1,\hat{\tau}}|\{G(t)\}_{t=0}^{\hat{\tau}})}+ \log{L_2}(\hat{\theta}_{2,\hat{\tau}}|\{G(t)\}_{t=\hat{\tau}}^T]- [{\log{L_2}(\hat{\theta}_{1,\tau^*}|\{G(t)\}_{t=0}^{\tau^*})}+ \log{L_2}(\hat{\theta}_{2,\tau^*}|\{G(t)\}_{t=\tau^*}^T)]\geq 0\right\}$.

Further, for $\forall \epsilon>0$,
\begin{eqnarray}\label{changepdf1}
\nonumber P\left({\hat{\tau}-\tau^*}>\epsilon T \right) =P\left(
A,{\hat{\tau}-\tau^*}>\epsilon T \right).
\end{eqnarray}

Equation (\ref{changepdf2.1}) implies that $\left[{\log{L_2}(\hat{\theta}_{1,\hat{\tau}}|\{G(t)\}_{t=1}^{\hat{\tau}})}+ \log{L_2}(\hat{\theta}_{2,\hat{\tau}}|\{G(t)\}_{t=\hat{\tau}}^T\right]
<\left[{\log{L_2}(\hat{\theta}_{1,\tau^*}|\{G(t)\}_{t=1}^{\tau^*})}+ \log{L_2}(\hat{\theta}_{2,\tau^*}|\{G(t)\}_{t=\tau^*}^T\right]$ with probability $1-o_p(1)$ when $T$ is large enough. It follows that
\begin{eqnarray}
\nonumber\lim_{T\rightarrow \infty}P\left({\hat{\tau}-\tau^*}>\epsilon T \right)=0.
\end{eqnarray}

Consider $\tau^*-\hat{\tau}>\epsilon T$. Divide the time range into three intervals $[0,T]=[0,\hat{\tau}]\cup[\hat{\tau}+1,\tau^*]\cup[\tau^*+1,T]$.

The remaining steps are similar to those for $\hat{\tau}-\tau^*>\epsilon T$, when $T$ is large enough, we have
\begin{eqnarray}
\nonumber\lim_{T\rightarrow \infty}P\left({\tau^*-\hat{\tau}}>\epsilon T \right)=0.
\end{eqnarray}

Finally, it follows that
\begin{eqnarray}
\nonumber\lim_{T\rightarrow \infty} P\left({|\hat{\tau}-\tau|}>\epsilon T \right)=0.
\end{eqnarray}

\end{proof}

\section*{Appendix}
In the supplementary material, we design simulations to illustrate the theoretical results and apply our method to real-life data.
The simulations for Section 3 confirm that our conclusions about the limit of the ratio of degrees from different groups are correct. Furthermore, the simulations for Section 4 show the effectiveness of our parameter estimation and change point detection methods.

We download three network datasets with category information from the Internet, WebKB-Wisconsin, BlogCatalog3 and CL-10K-1d8-L5. Through our parameter estimation method, their $\hat{\theta}$ are respectively $0.8100, 0.8470, 0.9999$. The results imply that the networks from different social platforms have different homophily, and it is helpful for the managers to construct a recommender system for a particular platform.

\section*{Acknowledgement}
This research is partially supported by National Natural Science Funds of China No.12001517 \& 72091212, "USTC Research Funds of the Double Frist-Class Initiative" YD2040002005 and "the Fundamental Research Funds for the Central Universities" WK2040000026 \& WK2040000027.
\bibliographystyle{siam}

\bibliography{sample}
\section*{Supplement}
\section{Simulations}
\subsection{The performance of $D_{T}^k$ and $S_T$}

This subsection verifies Theorems \ref{th1}--\ref{th2} and Corollary \ref{co1} by 500 simulation trials.

We design the simulation as follows:
\begin{itemize}
\item Let $T=10000$, $K=3$. The time range is $[0,10000]$. $B$ is the total number of trials. Let $B=500$.
\item The initial graph has 10 nodes, and $10p_k$ nodes are from group $k$, $k\in\{1,2,3\}$.
\item For each time $t$, a vertex-step occurs with probability $q$. In a vertex-step, the node from group $k$ arrives with probability $p_k$.
\item $D_T^k(b)$ records the total degrees of nodes from group $k$ and $S_T(b)$ records the number of edges whose nodes are from the same group in graph $G_T$ in the $b$th trial, $b\in\{1,\cdots,B\}$.
\end{itemize}
To test Theorem \ref{th1}--\ref{th2}, we fix $(p_1,p_2,p_3)=(0.5,0.3,0.2)$. Table \ref{ta1} shows the convergence of $\frac{D_T^k}{2T}$ and the effect of $q$ on the rate of convergence.

\begin{table}[H]
\caption{Fix $p_1=0.5,p_2=0.3,p_3=0.2$}
\center
\begin{tabular}{@{}lcrclrcl@{}}
\toprule[2pt]
&&\multicolumn{3}{c}{Bias}& \multicolumn{3}{c}{MSE} \\
$\theta$ &$q$ &
\multicolumn{1}{c@{}}{$\frac{\sum_{b=1}^B{D_T^1(b)}/{2T}}{B}-p_1$} &
\multicolumn{1}{c@{}}{$\frac{\sum_{b=1}^B{D_T^2(b)}/{2T}}{B}-p_2$} &
\multicolumn{1}{c}{$\frac{\sum_{b=1}^B{D_T^3(b)}/{2T}}{B}-p_3$}&
\multicolumn{1}{c@{}}{$\frac{\sum_{b=1}^B({D_T^1(b)}/{2T}-p_1)^2}{B}$} &
\multicolumn{1}{c@{}}{$\frac{\sum_{b=1}^B({D_T^2(b)}/{2T}-p_2)^2}{B}$} &
\multicolumn{1}{c}{$\frac{\sum_{b=1}^B({D_T^3(b)}/{2T}-p_3)^2}{B}$} \\
\hline
{0.9} & $0.9$ & $5.162e-04$ & $-1.030e-04$ & $8.68e-05$ & $9.658e-05$ & $8.3835e-05$ &$6.2627e-05$\\
          & $\frac{1}{2-\theta}$  & $-3.727e-04$   & $3.084e-04$ & $5.643e-04$  & $1.0267e-04$ & $8.6737e-05$&$6.6678e-05$\\
          & $0.1$   & $5.3532e-03$   &$ -2.7901e-03$& $-2.0631e-03$  & $0.0123$ & $9.8023e-03$ &$7.5445e-03$  \\[6pt]
{0.5} & $0.9$ & $-3.62e-05$ & $1.149e-04$ & $4.213e-04$ & $5.2097e-05$ & $4.4369e-05$ &$3.1959e-05$\\
          & $\frac{1}{2-\theta}$  & $ 5.17e-05$  & $1.618e-04$ & $2.865e-04$  & $1.3921e-04$ & $1.1696e-04$&$9.4448e-05$\\
          & $0.1$   & $-5.707e-04 $   &$ 6.5745e-03$& $-5.5038e-03$  & $0.0119$ & $0.0104$ &$7.6058e-03$  \\[6pt]
{0.3} & $0.9$ & $2.427 e-04$ & $-4.34e-04$ & $6.913e-04$ & $ 4.0563e-05$ & $3.5222e-05$ &$2.4632e-05$\\
          & $\frac{1}{2-\theta}$  & $8.325e-04$   & $-3.88e-05 $ & $-2.937e-04$  & $1.5919e-04$ & $1.4999e-04$&$1.0543e-04$\\
          & $0.1$   & $5.3489e-03$   &$ 4.5799e-03$& $-9.4288e-03$  & $0.0125$ & $0.01067$ &$7.5869e-03$  \\[6pt]
{0.1} & $0.9$ & $3.584e-04$ & $1.423e-04$ & $-7e-07$ & $3.3276e-05$ & $2.7255e-05$ &$2.2627e-05$\\
          & $\frac{1}{2-\theta}$  & $1.327 e-04$   & $-1.864 e-04 $ & $5.537e-04$  & $1.6958 e-04$ & $1.3811 e-04$&$1.22e-04$\\
          & $0.1$   & $3.3561e-03$   &$ 5.409e-04$& $-3.397e-03$  & $0.0109$ & $9.7896e-03$ &$7.2384e-03$  \\[6pt]
\toprule[2pt]
\end{tabular}
\label{ta1}
\end{table}

To test the convergence of $\frac{S_T}{T}$ by Corollary \ref{co1}, we fix $q=0.5$ in Table \ref{ta2}. Let $S^*=1+\theta(\sum^K_{k=1}p_k^2-1)$.

\begin{table}[H]
\caption{Fix $q=0.5$}
\center
\begin{tabular}{@{}lcccccccc@{}}
\toprule[2pt]
&&\multicolumn{1}{c}{Parameter}&\multicolumn{1}{c}{Estimator}&\multicolumn{1}{c}{Bias}& \multicolumn{1}{c}{MSE} \\
$\theta$ &$(p_1,p_2,p_3)$ &
\multicolumn{1}{c}{$S^*$} &
\multicolumn{1}{c}{$\frac{\sum_{b=1}^B S_T(b)/T}{B}$} &
\multicolumn{1}{c}{$\frac{\sum_{b=1}^BS_T(b)/T}{B}-S^*$}&
\multicolumn{1}{c}{$\frac{\sum_{b=1}^B({S_T(b)}/{T}-S^*)^2}{B}$}  \\
\hline
{0.9} & $(0.8,0.1,0.1)$ & $0.694$ & $ 0.6957$ & $1.6694e-03$ & $9.1195e-04$ \\
          & $(0.6,0.2,0.2)$  & $0.496$   & $0.4987$ & $2.6758e-03$  & $5.0171e-04$ \\
          & $(0.2,0.4,0.4)$   & $0.424$   &$ 0.4258$& $1.7642e-03$  & $1.0084e-04$ \\[6pt]
{0.5} & $(0.8,0.1,0.1)$ & $0.83$ & $ 0.8303$ & $2.73e-04$ & $1.6808e-04$ \\
          & $(0.6,0.2,0.2)$  & $ 0.72$  & $0.7204$ & $3.764e-04$  & $1.0823e-04$\\
          & $(0.2,0.4,0.4)$   & $0.68$   &$ 0.6806$& $5.794e-04$  & $3.4297e-05$ \\[6pt]
{0.3} & $(0.8,0.1,0.1)$ & $0.898$ & $0.8985$ & $4.988e-04$ & $5.976e-05$ \\
          & $(0.6,0.2,0.2)$  & $0.832$   & $0.8325$ & $5.488e-04$  & $3.361e-05$ \\
          & $(0.2,0.4,0.4)$   & $0.808$   &$ 0.8085$& $5.072e-04$  & $2.011e-05$ \\[6pt]
{0.1} & $(0.8,0.1,0.1)$ & $0.966$ & $ 0.9658$ & $ -2.1e-04$ & $7.5024e-06$ \\
          & $(0.6,0.2,0.2)$  & $0.944$   & $0.9439$ & $ -1.056e-04$  & $7.6358e-06$ \\
          & $(0.2,0.4,0.4)$   & $0.936$   &$ 0.9361$& $9.9e-05$  & $6.6796e-06$ \\[6pt]
\toprule[2pt]
\end{tabular}
\label{ta2}
\end{table}

\subsection{Estimators of parameters with historical information}
This subsection is to verify Theorems \ref{th4}--\ref{th5} and Corollary \ref{co2}.
The conclusion about $\hat{q},\{\hat{p}_k\}_{k=1}^K$ is apparent, so we use 500 simulation trials to test the convergence of $\hat{\theta}$.
We design the simulation as follows:
\begin{itemize}
    \item Let $T=10000,K=3$. The time range is $[0,10000]$. $B$ is the total number of trials. Let $B=500$.
\item The initial graph has 10 nodes, and $10p_k$ nodes are from group $k$,  $k\in\{1,2,3\}$.
\item For each time $t$, a vertex-step occurs with probability $q$.
In a vertex-step, the node from group $k$ arrives with probability $p_k$. And we record the $v_t$, $(g_w,g_u)$ at each time $t\in [0,T]$.
\item Gain $D_t^k$ by calculating the total degrees of nodes from group $k$ in graph $G_t$.
\item Construct the maximum likelihood equation by $\{(\{D_t^k\}_{k=1}^K,v_t,g_w,g_u)\}^T_{t=1}$. Record the estimator $\hat{\theta}(b)$ and $\hat{\Sigma}_{11}(b)$  in the $b$th trial, $b\in\{1,\cdots,B\}$.
\item According to Theorem \ref{th5} and Corollary \ref{co2}, construct the 95\% confidence interval (CI) for the parameters $\hat{\theta}$ and $\frac{\hat{\theta}-\theta}{\hat{\Sigma}^{1/2}_{11}}$. Calculate the percentage of the estimators in 500 trials that fall in the confidence interval. The results of $\hat{\theta}$ and $\frac{\hat{\theta}-\theta}{\hat{\Sigma}^{1/2}_{11}}$ are recorded in Table \ref{CI1}--\ref{CI2}.
\end{itemize}

\begin{table}[H]
\caption{}
\center
\begin{tabular}{@{}lcccccccc@{}}
\toprule[2pt]
&&\multicolumn{1}{c}{Estimator}&\multicolumn{1}{c}{Bias}&\multicolumn{1}{c}{MSE}&\multicolumn{1}{c}{Cover rate}\\
$\theta$ &$(p_1,p_2,p_3,q)$ &
\multicolumn{1}{c}{$\frac{\sum_{b=1}^B\hat{\theta}(b)}{B}$} &
\multicolumn{1}{c}{$\frac{\sum_{b=1}^B\hat{\theta}(b)}{B}-\theta$}&
\multicolumn{1}{c}{$\frac{\sum_{b=1}^B(\hat{\theta}(b)-\theta)^2}{B}$} &
\multicolumn{1}{c}{$\frac{\sum_{b=1}^B\mathbf{1}_{\{\hat{\theta}(b) \in 95\% \text{CI}\}}}{B}$} \\
\hline
{0.9} & $(0.8,0.1,0.1,0.8)$ & $0.8999$ & $ -7.1037e-05$ & $8.1465e-05$ & $0.95$\\
          & $(0.8,0.1,0.1,0.2)$  & $0.9000$   & $-3.4965e-05$ & $ 8.4753e-05$  & $ 0.94$\\
          & $(0.2,0.4,0.4,0.2)$   & $0.9005$   &$ 5.3677e-04$& $5.6645e-05$  & $0.96$ \\[6pt]
{0.5} & $(0.8,0.1,0.1,0.8)$ & $0.5002$ & $2.3398e-04$ & $9.5765e-05$ & $0.96$ \\
          & $(0.8,0.1,0.1,0.2)$  & $ 0.4997$  & $-3.2257e-04$ & $1.0074e-04$  & $0.96$\\
          & $(0.2,0.4,0.4,0.2)$   & $0.5004$   &$4.1498e-04$& $6.0659e-05$  & $0.945$ \\[6pt]
{0.3} & $(0.8,0.1,0.1,0.8)$ & $0.3000$ & $-4.8493e-05$ & $7.3418e-05$ & $0.952$ \\
          & $(0.8,0.1,0.1,0.2)$  & $0.2998$   & $-1.9043e-04$ & $7.6047e-05$ &$0.952$ \\
          & $(0.2,0.4,0.4,0.2)$   & $0.3002$   &$ 2.1887e-04$& $4.2600e-05$  & $0.944$ \\[6pt]
{0.1} & $(0.8,0.1,0.1,0.8)$ & $ 0.10007$ & $ 6.6986e-05$ & $ 2.8706e-05$ & $0.944$\\
          & $(0.8,0.1,0.1,0.2)$  & $0.0997$   & $-3.3447e-04$ & $3.1154e-05$  & $0.948$ \\
          & $(0.2,0.4,0.4,0.2)$   & $0.0998$   &$ -1.9246e-04$& $1.3169e-05$  & $0.964$ \\[6pt]
\toprule[2pt]
\end{tabular}
\label{CI1}
\end{table}

\begin{table}[H]
\caption{}
\center
\begin{tabular}{@{}lcccccccc@{}}
\toprule[2pt]
&&\multicolumn{1}{c}{Bias}&\multicolumn{1}{c}{MSE}&\multicolumn{1}{c}{Cover rate}\\
$\theta$ &$(p_1,p_2,p_3,q)$ &
\multicolumn{1}{c}{$\sum_{b=1}^B\frac{{(\hat{\theta}(b)-\theta)}/{\hat{\Sigma}_{11}^{1/2}(b)}}{B}$}&
\multicolumn{1}{c}{$\sum_{b=1}^B\frac{{\left(\hat{\theta}(b)-\theta\right)^2}/{\hat{\Sigma}_{11}(b)}}{B}$} &
\multicolumn{1}{c}{$\frac{\sum_{b=1}^B\mathbf{1}_{\{(\hat{\theta}(b)-\theta)/{\hat{\Sigma}_{11}^{1/2}(b)} \in 95\% \text{CI}\}}}{B}$} \\
\hline
{0.9} & $(0.8,0.1,0.1,0.8)$ & $ -0.0049$ & $ 0.9884$  & $ 0.948$\\
& $(0.8,0.1,0.1,0.2)$& $ -0.1044$ & $0.9781$ & $0.944$\\
          & $(0.2,0.4,0.4,0.2)$ &$  -0.0071$& $0.9223$  & $0.956$ \\[6pt]
{0.5} & $(0.8,0.1,0.1,0.8)$  & $ -0.0618$ & $1.0227$ & $0.948$ \\
          & $(0.8,0.1,0.1,0.2)$   & $0.1187$ & $1.1329$  & $0.948$\\
          & $(0.2,0.4,0.4,0.2)$   &$ -0.0261$& $0.9400$  & $0.945$ \\[6pt]
{0.3} & $(0.8,0.1,0.1,0.8)$ & $-0.0402$ & $0.9931$ & $0.955$ \\
          & $(0.8,0.1,0.1,0.2)$   & $0.0701$ & $1.0650$ &$0.945$ \\
          & $(0.2,0.4,0.4,0.2)$   &$ -0.1514$& $0.9968$  & $0.95$ \\[6pt]
{0.1} & $(0.8,0.1,0.1,0.8)$ & $ 0.0265$ & $ 0.9625$ & $0.965$\\
          & $(0.8,0.1,0.1,0.2)$ & $0.0470$ & $0.9788$  & $0.965$ \\
          & $(0.2,0.4,0.4,0.2)$ &$ 0.0725$& $1.0075$  & $ 0.955$ \\[6pt]
\toprule[2pt]
\end{tabular}
\label{CI2}
\end{table}

\subsection{Snapshot}
This subsection is to verify Theorem \ref{th6}, the convergence of snapshot estimation with the single graph $G_T$ . What is more, we compare effect of the estimator of snapshot estimation with the general MLE by 500 trials.

\begin{itemize}
    \item Let $T=10000,K=3$. The time range is $[0,10000]$. $B$ is the total number of trials. Let $B=500$.
\item The initial graph has 10 nodes, and $10p_k$ nodes are from group $k$, $k\in\{1,2,3\}$.
\item For each time $t$, a vertex-step occurs with probability $q$. In a vertex-step, the node from group $k$ arrives with probability $p_k$. We record the $v_t$ and $(g_w,g_u)$ at each time $t\in [0,T]$.
\item Gain $D_t^k$ by calculating the total degrees of nodes from group $k$ in the graph $G_t$.
\item Construct the maximum likelihood estimation (Equation (\ref{mle1})) by $\{(\{D_t^k\}_{k=1}^K,v_t,g_w,g_u)\}_{t=1}^T$. Record the estimator of $\theta$ as ${\theta}_{mle}(b)$ in $b$th trial, $b\in\{1,\cdots,B\}$.
\item Construct the snapshot estimation (Equation (\ref{smle1})) by $(\{D_T^k\}_{k=1}^K,E_T,E_T^{k,1},E_T^{k,0})$. Record the estimator of $\theta$ as ${\theta}_{snap}(b)$ in $b$th trial, $b\in\{1,\cdots,B\}$.
\item Calculate the mean absolute error (MAE) and the mean square error (MSE) for $\hat{\theta}_{mle}=\frac{\sum_{b=1}^B \theta_{mle}(b)}{B}$ and $\hat{\theta}_{snap}=\frac{\sum_{b=1}^B \theta_{snap}(b)}{B}$.
\end{itemize}

\begin{table}[H]
\caption{Fix $p_1=0.5,p_2=0.3,p_3=0.2$}
\center
\begin{tabular}{@{}lcccccccc@{}}
\toprule[2pt]
&&\multicolumn{2}{c}{Estimator}&\multicolumn{2}{c}{MAE}&\multicolumn{2}{c}{MSE}\\
$\theta$ &$q$ &
\multicolumn{1}{c}{$\hat{\theta}_{mle}$} &
\multicolumn{1}{c}{$\hat{\theta}_{snap}$}&
\multicolumn{1}{c}{$\frac{\sum_{b=1}^B |{\theta}_{mle}(b)-\theta|}{B}$}&
\multicolumn{1}{c}{$\frac{\sum_{b=1}^B |{\theta}_{snap}(b)-\theta|}{B}$} &
\multicolumn{1}{c}{$\frac{\sum_{b=1}^B |{\theta}_{mle}(b)-\theta|^2}{B}$}&
\multicolumn{1}{c}{$\frac{\sum_{b=1}^B |{\theta}_{snap}(b)-\theta|^2}{B}$} \\
\hline
{0.9} & $0.9$ & $0.9006$ & $ 0.9004$ & $0.0103$ & $0.0104$ & $1.6321e-04$ & $1.7057e-04$ \\
          & $0.5$  & $0.8995$   & $0.8985$ & $ 9.7439e-03$  & $0.0104$& $1.5096e-04$&$ 1.7223e-04$\\
          & $0.1$   & $0.9003$   &$ 0.8980$& $0.0108$  & $0.0117$&$1.9355e-04$&$2.2941e-04$ \\[6pt]
{0.5} & $0.9$ & $0.5002$ & $0.5003$ & $6.1585e-03$ & $6.2225e-03$ & $5.8738e-05$ & $6.0042e-05$\\
          & $0.5$  & $ 0.4997$  & $0.4994$ & $5.8597e-03$  & $5.9928e-03$& $5.2679e-05$ &$5.4898e-05$\\
          & $0.1$   & $0.5001$   & $0.4998$  & $6.6218e-03$ &$6.8167e-03$ &$6.6155e-05$& $7.1367e-05$ \\[6pt]
{0.3} & $0.9$ & $0.3$ & $0.3$ & $4.7002e-03$ & $4.7281e-03$ &$3.3735e-05$ &$3.4222e-05$ \\
          & $0.5$  & $0.3002$  & $0.3000$  & $4.927e-03$ & $4.9416e-03$ &$3.7533e-05$ & $3.8268e-05$ \\
          & $0.1$   & $0.2999$   &$ 0.2986$& $4.9432e-03$  & $5.197e-03$ &$3.8996e-05$&$4.3797e-05$\\[6pt]
{0.1} & $0.9$ & $0.0999$ & $ 0.0997$ & $6.1567e-03$ & $6.2798e-03$&$5.9996e-05$&$6.1705e-05$ \\
          & $0.5$  & $0.1$   & $0.0997$ & $5.5049e-03$& $5.8268e-03$  &$4.8073e-05$ &$5.4660e-05$ \\
          & $0.1$   & $0.0996$   &$ 0.0980$& $6.6773e-03$  & $7.2647e-03$&$7.4653e-05$&$8.9782e-05$ \\[6pt]
\toprule[2pt]
\end{tabular}
\label{snap}
\end{table}

\subsection{Change point}
This subsection is to verify Theorem \ref{changep}, the method of detecting the change point $\tau$ in the range $[0,T]$ in 500 trials.

\begin{itemize}
    \item Let $K=3$. $T=1000,2000$. The time range is $[0,T]$. $B$ is the total number of trials. Let $B=500$.
\item The initial graph has 10 nodes, and $10 p_k$ nodes are from group $k$, $k\in\{1,2,3\}$.
\item For each time $t$, a vertex-step occurs with probability $q$. In a vertex-step, the node from group $k$ arrives with probability $p_k$.
\item  Set $c_0=\frac{t_0}{T}=0.2\, ,\frac{\tau}{T}\in \{0.25,0.5,0.75\}$. Let the homophily parameter be $\theta_1$ before the change point and $\theta_2$ after the change point.
\item We record the $v_t$ and $(g_w,g_u)$ for each time $t\in [0,T]$. Gain $D_t^k$ by calculating the total degrees of nodes from group $k$ in the graph $G_t$.
\item By the Equation (\ref{changepoint}) and get the estimator $\tau(b)$ and $\theta_1(b),\theta_2(b)$ in the $b$th trial, $b\in\{1,\cdots,B\}$.
\item Tables \ref{change}--\ref{change2} show the result of $\hat{\tau}=\frac{\sum_{b=1}^B\tau(b)}{B}$, $\hat{\theta}_1=\frac{\sum_{b=1}^B \theta_1(b)}{B}$, $\hat{\theta}_2=\frac{\sum_{b=1}^B \theta_2(b)}{B}$.
\end{itemize}

\begin{table}[H]
\caption{Fix $T=1000,p_1=0.5,p_2=0.3,p_3=0.2$}
\center
\begin{tabular}{@{}lccccccccc@{}}
\toprule[2pt]
&&\multicolumn{4}{c}{Estimator}&\multicolumn{3}{c}{MSE}\\
$\tau$ &$\theta_1,\theta_2$ &
\multicolumn{1}{c}{$\hat{\tau}$} &
\multicolumn{1}{c}{$\frac{\hat{\tau}-\tau}{T}$} &
\multicolumn{1}{c}{$\hat{\theta}_1$}&
\multicolumn{1}{c}{$\hat{\theta}_2$}&
\multicolumn{1}{c}{$\frac{\sum_{b=1}^B[({\tau}(b)-\tau)/{T}]^2}{B}$}&
\multicolumn{1}{c}{$\frac{\sum_{b=1}^B({\theta_1}(b)-\theta_1)^2}{B}$}&
\multicolumn{1}{c}{$\frac{\sum_{b=1}^B({\theta_2}(b)-\theta_2)^2}{B}$}\\
\hline
{250} & $0.1,0.9$ & $250.55$ & $5.5e-04$ & $0.1007$ & $0.9033$ & $1.048e-05$ & $6.7403e-04$ &5.88e-04 \\
          & $0.4,0.6$  & $278.07$   & $0.0281$ & $ 0.3861$  & $0.6025$& $0.0196$&$ 4.6399e-03$&3.9222e-03\\
{500} & $0.1,0.9$ & $500.775$ & $7.75e-04$ & $0.0991$ & $0.8989$ & $1.1755e-05$ & $2.972e-04$&$1.2449e-03$\\
          & $0.4,0.6$  & $ 486.43$  & $-0.0136$ & $0.3893$  & $0.6002$& $0.0115$ &$2.8825e-03$&$3.3877e-03$\\
{750} & $0.1,0.9$ &$750.04$& $4e-05$ & $0.1018$ & $0.8946$ & $1.218e-05$ &$1.9772e-04$ &$2.5049e-03$ \\
          & $0.4,0.6$  & $736.53$  & $-0.01347$  & $0.3923$ & $0.6215$ &$0.0122$ & $1.2203e-03$& $3.7263e-03$ \\
\toprule[2pt]
\end{tabular}
%
\label{change}
\end{table}

\begin{table}[H]
\caption{Fix $T=2000,p_1=0.5,p_2=0.3,p_3=0.2$}
\center
\begin{tabular}{@{}lccccccccc@{}}
\toprule[2pt]
&&\multicolumn{4}{c}{Estimator}&\multicolumn{3}{c}{MSE}\\
$\tau$ &$\theta_1,\theta_2$ &
\multicolumn{1}{c}{$\hat{\tau}$} &
\multicolumn{1}{c}{$\frac{\hat{\tau}-\tau}{T}$} &
\multicolumn{1}{c}{$\hat{\theta}_1$}&
\multicolumn{1}{c}{$\hat{\theta}_2$}&
\multicolumn{1}{c}{$\frac{\sum_{b=1}^B[({\tau}(b)-\tau)/{T}]^2}{B}$}&
\multicolumn{1}{c}{$\frac{\sum_{b=1}^B({\theta_1}(b)-\theta_1)^2}{B}$}&
\multicolumn{1}{c}{$\frac{\sum_{b=1}^B({\theta_2}(b)-\theta_2)^2}{B}$}\\
\hline
{500} & $0.1,0.9$ & $500.59$ & $2.95e-04$ & $0.1007$ & $0.9016$ & $2.3525e-06$ & $2.7622e-04$ &4.0328e-04 \\
          & $0.4,0.6$  & $497.55$   & $-1.225e-03$ & $ 0.3896$  & $0.6005$& $2.5461e-03$&$ 2.1478e-03$&4.3125e-04\\
{1000} & $0.1,0.9$ & $1000.23$ & $1.15e-04$ & $0.1004$ & $0.9016$ & $1.6275e-06$ & $1.2538e-04$&$6.2472e-04$\\
          & $0.4,0.6$  & $ 1004.2$  & $2.1e-03$ & $0.3989$  & $0.6044$& $9.1024e-04$ &$5.3224e-04$&$5.7490e-04$\\
{1500} & $0.1,0.9$ &$1501.03$& $5.15e-04$ & $0.1016$ & $0.9018$ & $2.1325e-06$ &$1.0825e-04$ &$1.1565e-03$ \\
          & $0.4,0.6$  & $1497.04$  & $-1.48e-03$  & $0.3996$ & $0.6036$ &$3.4314e-03$ & $4.4038e-04$& $1.5538e-03$ \\
\toprule[2pt]
\end{tabular}
%
\label{change2}
\end{table}

\section{Data application}

We selected three real network datasets to test our KPA model and estimation methods. These network datasets all have group labels or category labels of nodes. However, none of the datasets recorded the timestamp data for edge connections. We
summarize the datasets below:

\begin{itemize}

  \item WebKB-Wisconsin

  WebKB is a webpage dataset collected from the computer science departments of various
universities (Cornell, Texas, Washington, and Wisconsin) by Carnegie Mellon University. WebKB-Wisconsin is the sub dataset from the University of Wisconsin. We obtained the network form of WebKB-Wisconsin from the web NetworkRepository (\cite{nr}).
\href{https://networkrepository.com/webkb-wisc.php}{https://networkrepository.com/webkb-wisc.php}.
For the network data, nodes represent web pages, and edges are hyperlinks between them. The web pages are manually classified
into five categories: student, project, course, staff, and faculty. Since some categories have too few nodes, we consider merging them. Let the nodes labeled ``student'' or ``faculty'' be in the same group. The remaining nodes are assigned to other groups.

  \item BlogCatalog3

  BlogCatalog3 is a dataset crawled from BlogCatalog (\href{http://datasets.syr.edu/pages/datasets.html}{http://datasets.syr.edu/pages/datasets.html}). BlogCatalog is a social blog directory website that manages bloggers and blogs. Both the contact network and category information are included. A blogger's interests can be gauged by the categories they publish their blogs, and each blogger could list their blog under more than one category. This network dataset has 10312 nodes, 333983 edges, and 39 categories. Since a blog may have several categories simultaneously, we only keep the first one in the given order. The nodes are divided into 39 groups based on category labels. We select four groups with at least 700 members to construct a sub-network of BlogCatalog3.

  \item CL-10K-1d8-L5

  CL-10K-1d8-L5 is a network dataset with group information from the web Network Repository (\cite{nr}).

\href{https://networkrepository.com/CL-10K-1d8-L5.php}{https://networkrepository.com/CL-10K-1d8-L5.php}.

\end{itemize}

The basic information about these network datasets is shown in Table \ref{data1}.
\begin{table}[H]
\centering
\caption{Information about network datasets}\label{data1}
\begin{tabular}{|l|c|c|c|}
\toprule[2pt]
Name & WebKB-Wisconsin&BlogCatalog3 &CL-10K-1d8-L5 \\
\hline
Number of nodes& 265& 3760 &10000 \\
Number of groups & 2 &4 &5\\
Number of nodes from group 1& 198 &759 &2000\\
Number of nodes from group 2& 67 & 843 &2000\\
Number of nodes from group 3& $\backslash$ & 735 &2000\\
Number of nodes from group 4& $\backslash$ & 1423 &2000\\
Number of nodes from group 5& $\backslash$ & $\backslash$&2000 \\
Number of edges &530 & 59893 &  44896\\
Total degrees from group 1 & 769 &  22316 &18965\\
Total degrees from group 2 & 291& 28383 &  17381\\
Total degrees from group 3 & $\backslash$ & 26879 & 17309\\
Total degrees from group 4 & $\backslash$ & 42208 & 17949\\
Total degrees from group 5 & $\backslash$ & $\backslash$ & 18188\\
The estimator $\hat{\theta}$&  0.8100&  0.8470 & 0.9999\\
\bottomrule
\end{tabular}
\end{table}

By the snapshot estimation method, we can get the estimators of parameters. We exhibit the estimation of $\theta$ in Table \ref{data1}. It shows that the $\hat{\theta}$ of WebKB-Wisconsin and BlogCatalog3 is less than 1. Both networks have homophily: their nodes bond with others in their group at a higher rate. $\hat{\theta}$ of CL-10K-1d8-L5 is very close to 1, which implies that there is no difference between the connection of nodes from the same group and different groups.

Further, rich-get-richer is another essential mechanism of the KPA model. Figures \ref{F4}--\ref{F8} show the power-law degree distribution of the three datasets.

\begin{figure}[H]
\centering
\subfigure[ Histogram of degree distribution]{
\begin{minipage}[t]{0.5\linewidth}
\centering
\includegraphics[width=2.5in]{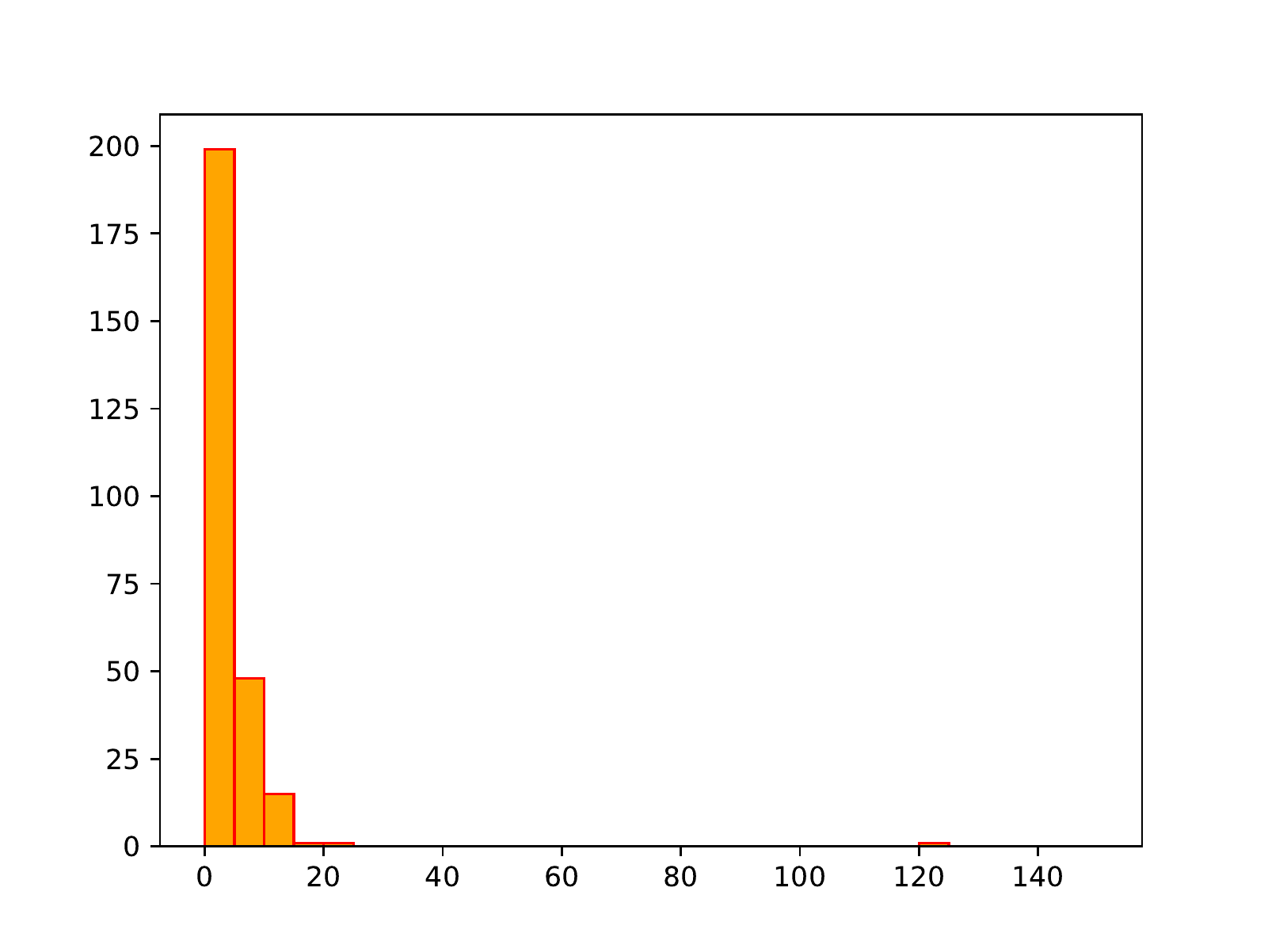}
\end{minipage}%
}%
\subfigure[Orange denotes group 1 and blue denotes group 2]{
\begin{minipage}[t]{0.5\linewidth}
\centering
\includegraphics[width=2.5in]{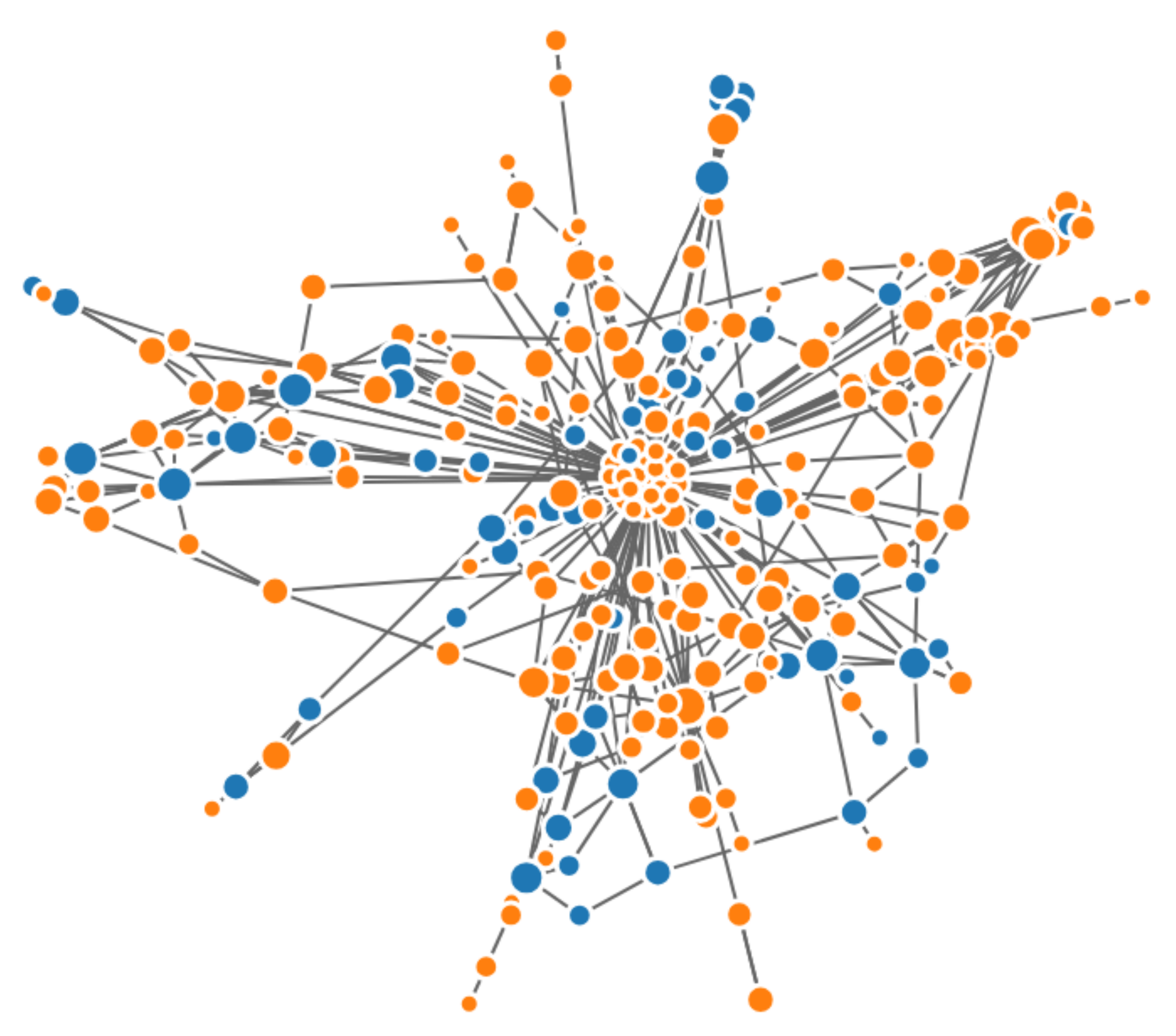}
\end{minipage}%
}%
\centering
\caption{Network graph of WebKB-Wisconsin}
\label{F4}
\end{figure}

\begin{figure}[H]
\centering
\subfigure[Histogram of degree distribution]{
\begin{minipage}[t]{0.5\linewidth}
\centering
\includegraphics[width=2.5in]{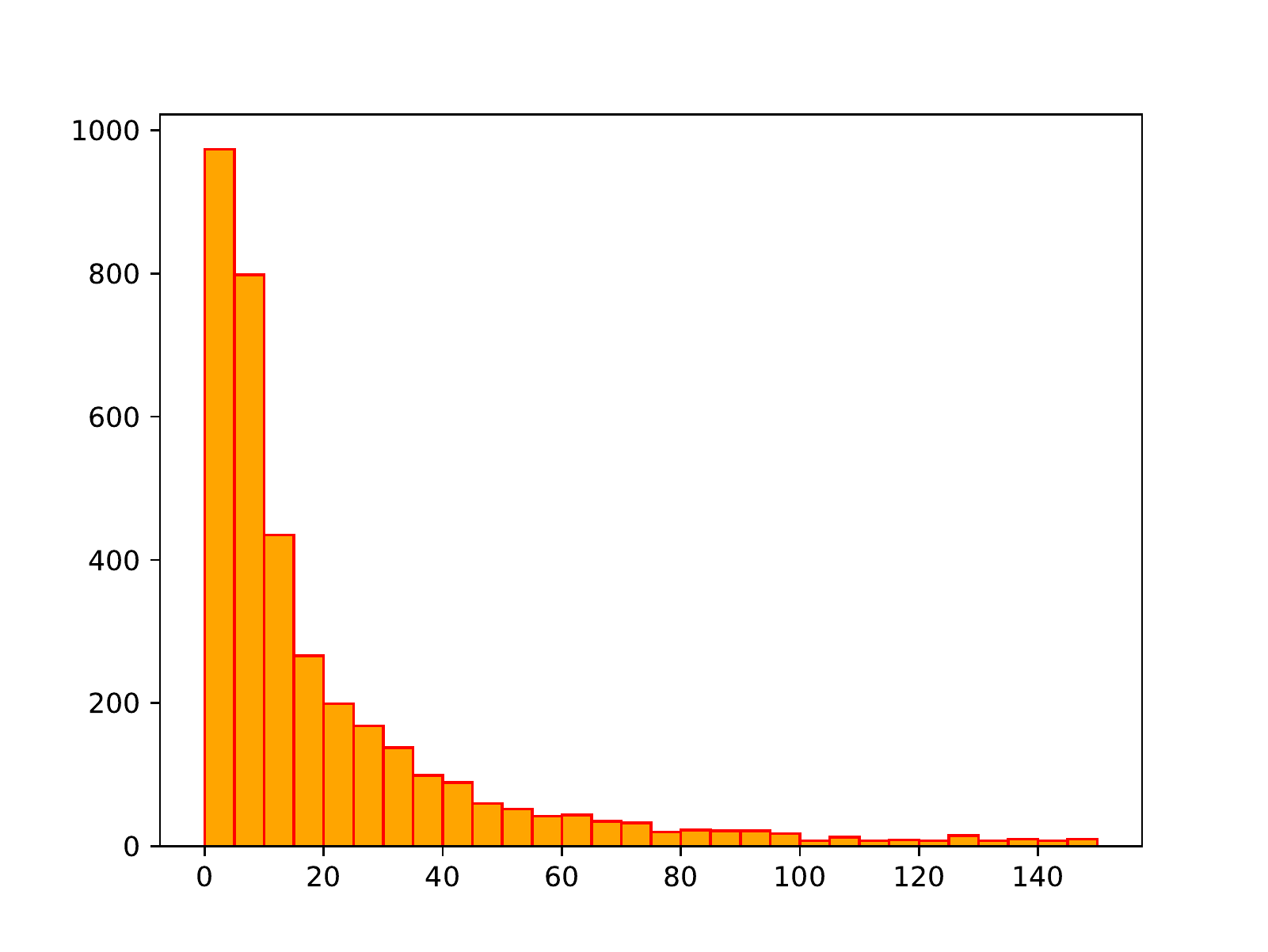}
\end{minipage}%
}%
\subfigure[log-log scale of degree distribution]{
\begin{minipage}[t]{0.5\linewidth}
\centering
\includegraphics[width=2.5in]{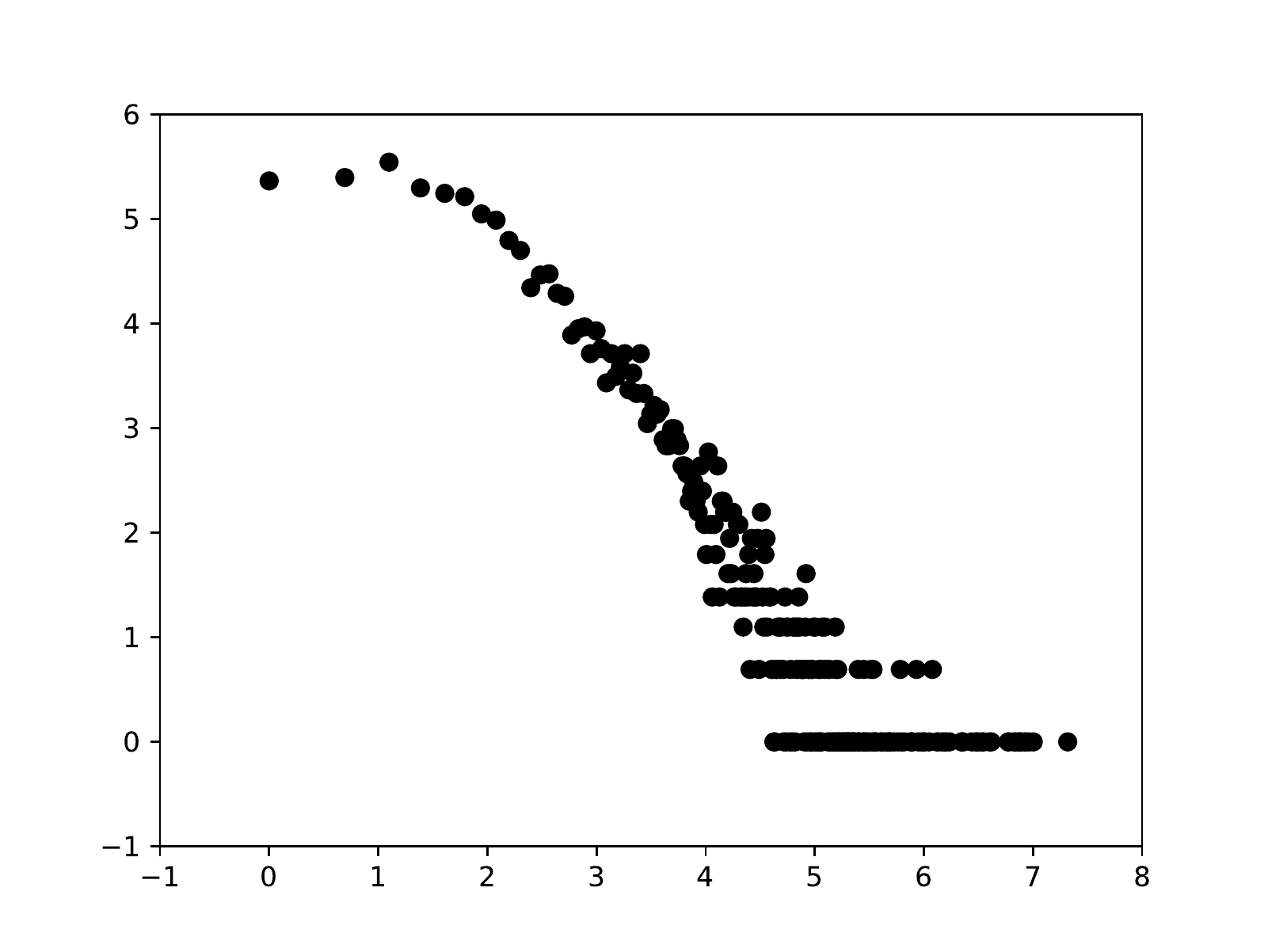}
\end{minipage}%
}%
\centering
\caption{Degree distribution of BlogCatalog3}
\label{F5}
\end{figure}
\begin{figure}[H]
\centering
\subfigure[Histogram of degree distribution]{
\begin{minipage}[t]{0.5\linewidth}
\centering
\includegraphics[width=2.5in]{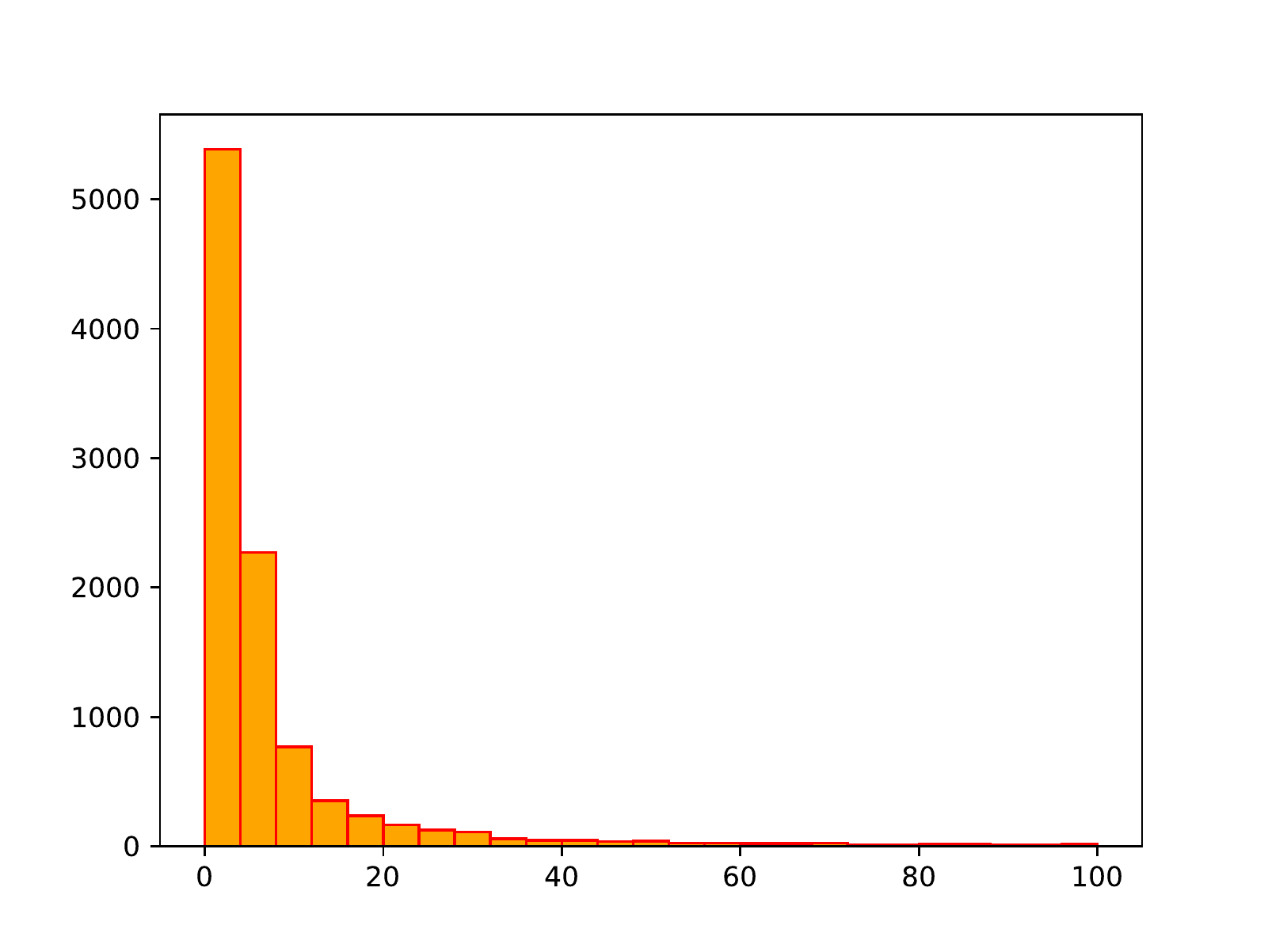}
\end{minipage}%
}%
\subfigure[log-log scale of degree distribution]{
\begin{minipage}[t]{0.5\linewidth}
\centering
\includegraphics[width=2.5in]{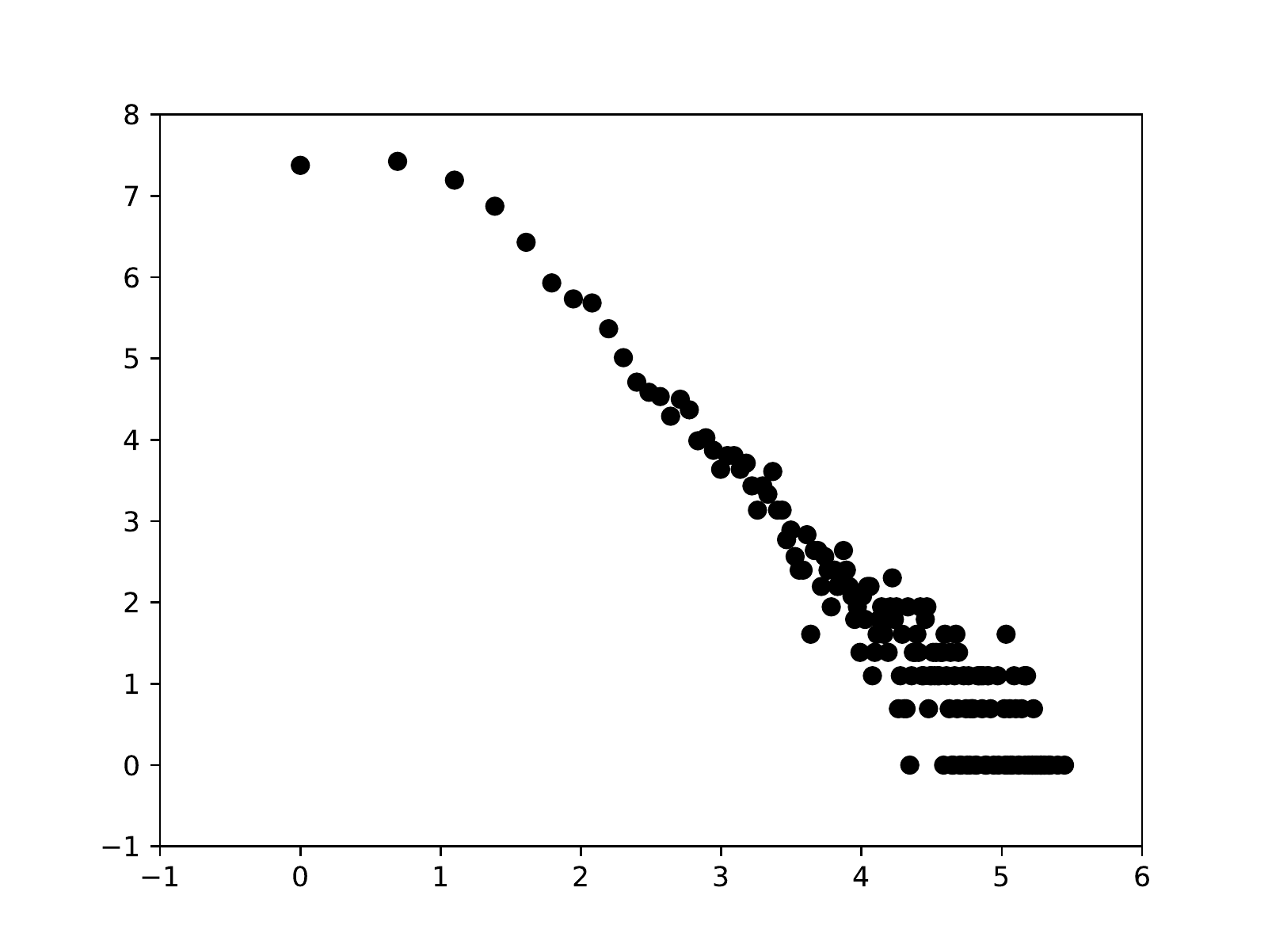}
\end{minipage}%
}%
\centering
\caption{Degree distribution of CL-10K-1d8-L5}
\label{F7}
\end{figure}

\begin{figure}[H]
\centering
\subfigure[Group 1]{
\begin{minipage}[t]{0.5\linewidth}
\centering
\includegraphics[width=2.5in]{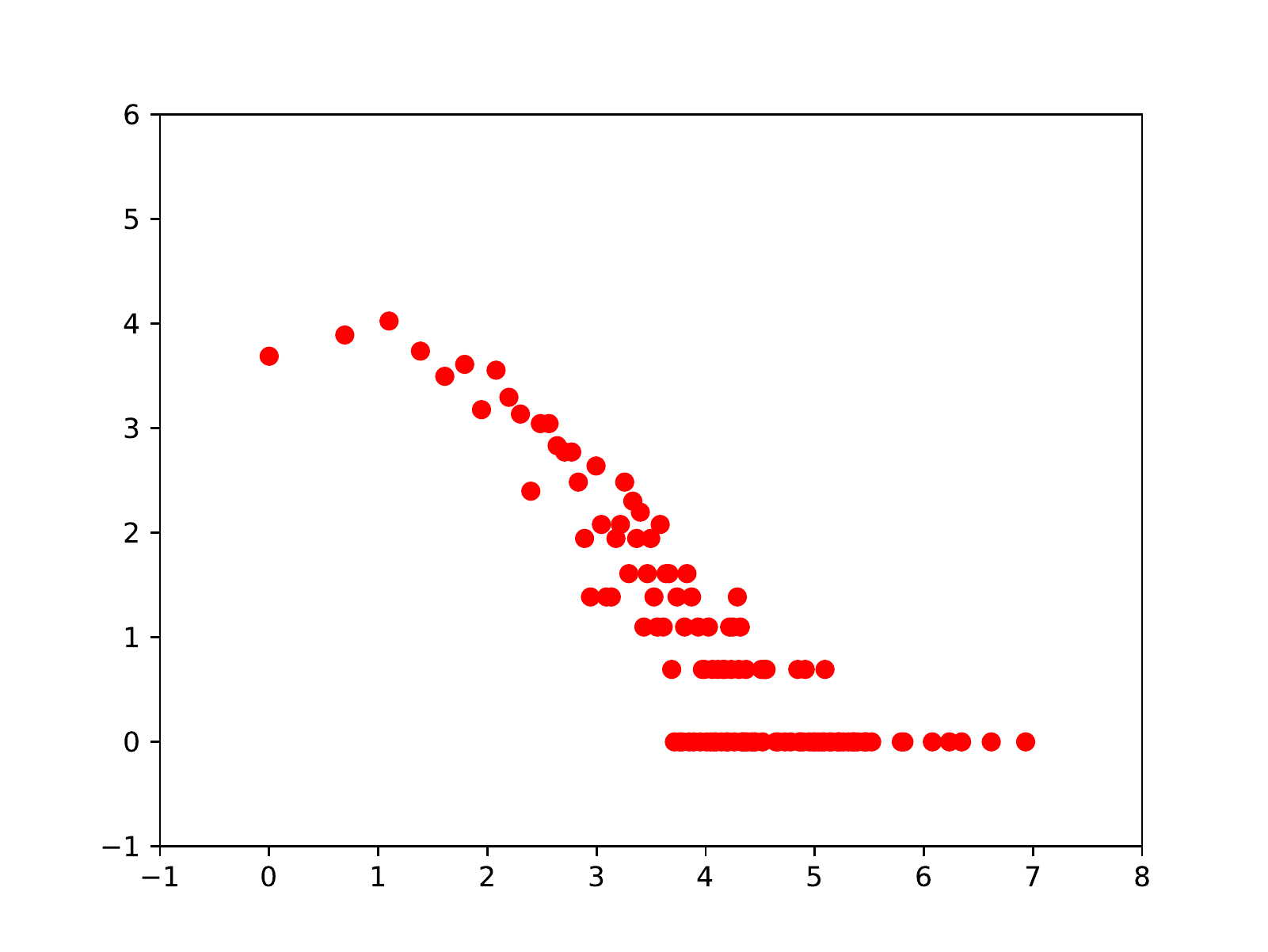}
\end{minipage}%
}%
\subfigure[Group 2]{
\begin{minipage}[t]{0.5\linewidth}
\centering
\includegraphics[width=2.5in]{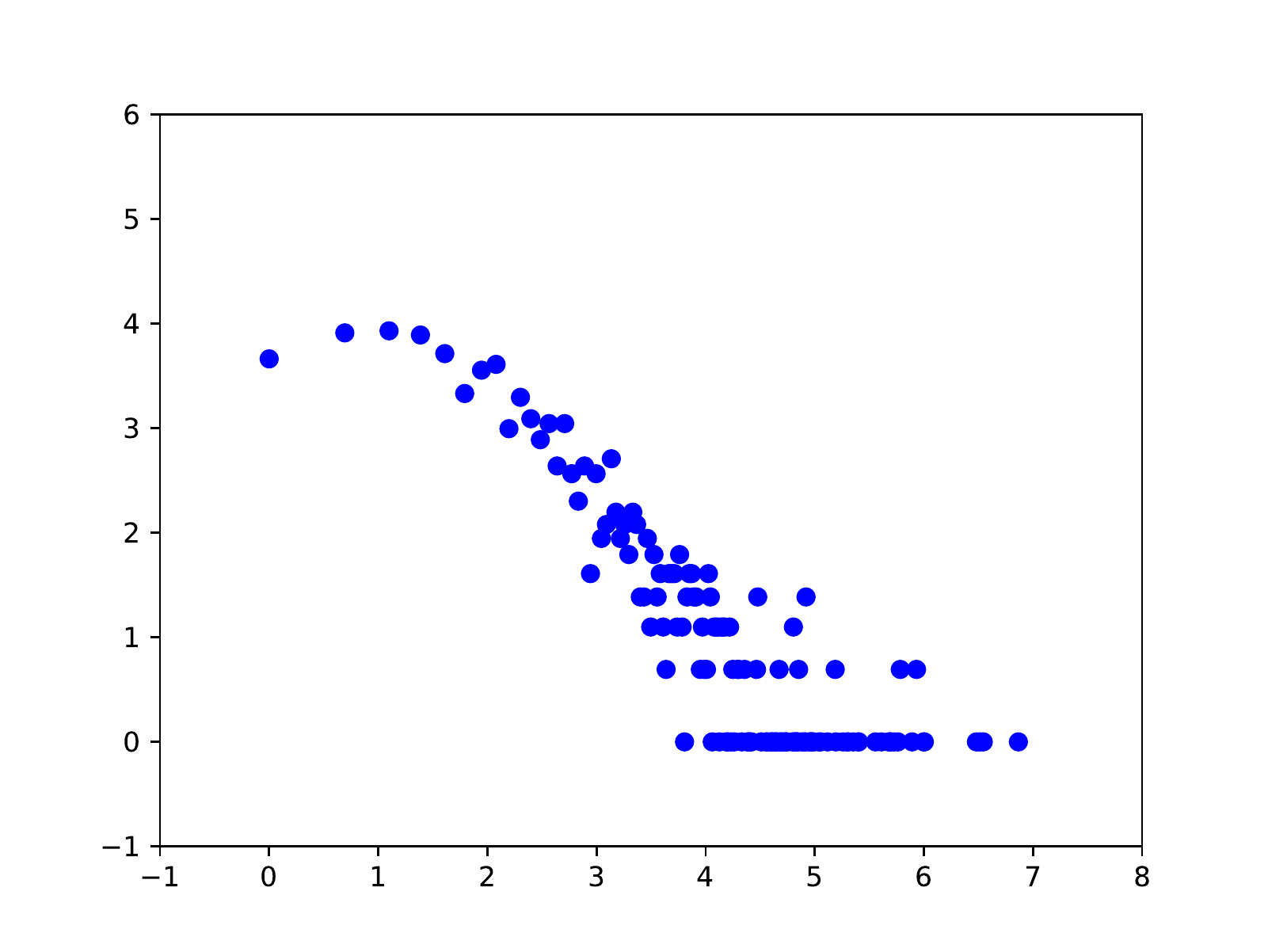}
\end{minipage}%
}%

\centering
\subfigure[Group 3]{
\begin{minipage}[t]{0.5\linewidth}
\centering
\includegraphics[width=2.5in]{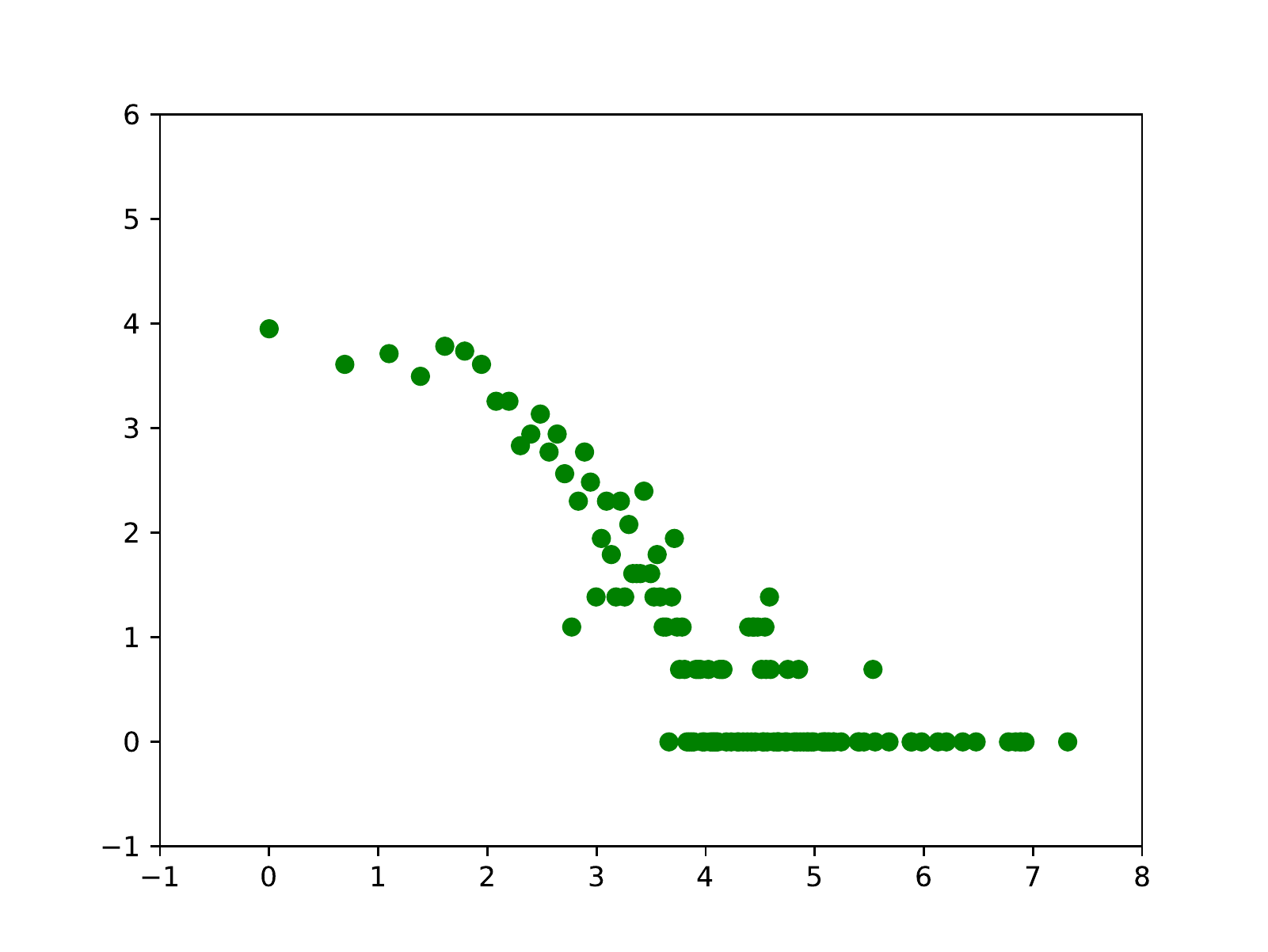}
\end{minipage}%
}%
\subfigure[Group 4]{
\begin{minipage}[t]{0.5\linewidth}
\centering
\includegraphics[width=2.5in]{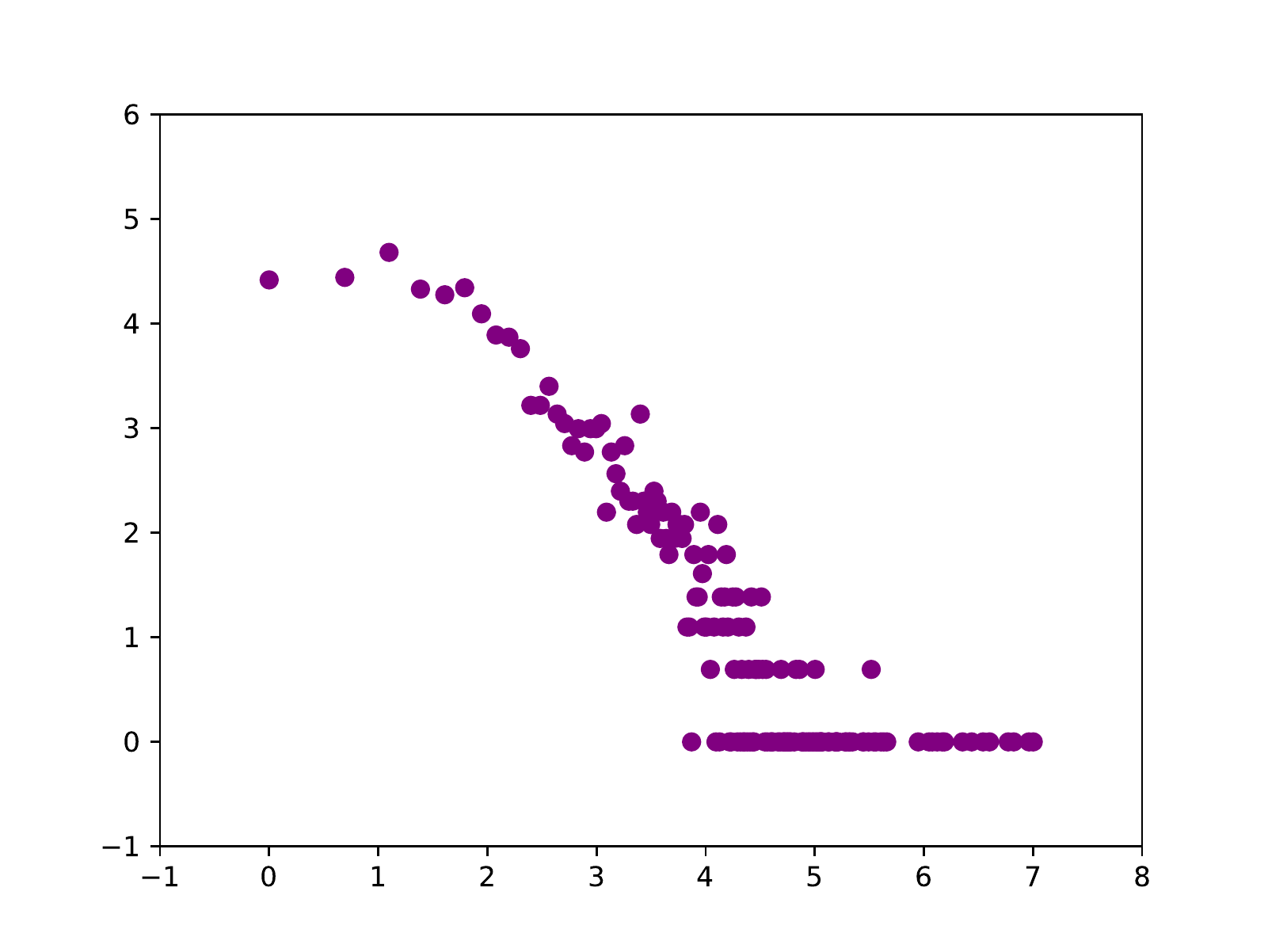}
\end{minipage}%
}%
\centering
\caption{Log-log degree distribution from different groups of BlogCatalog3}
\label{F6}
\end{figure}

\begin{figure}[H]
\centering
\subfigure[Group 1]{
\begin{minipage}[t]{0.35\linewidth}
\centering
\includegraphics[width=2.5in]{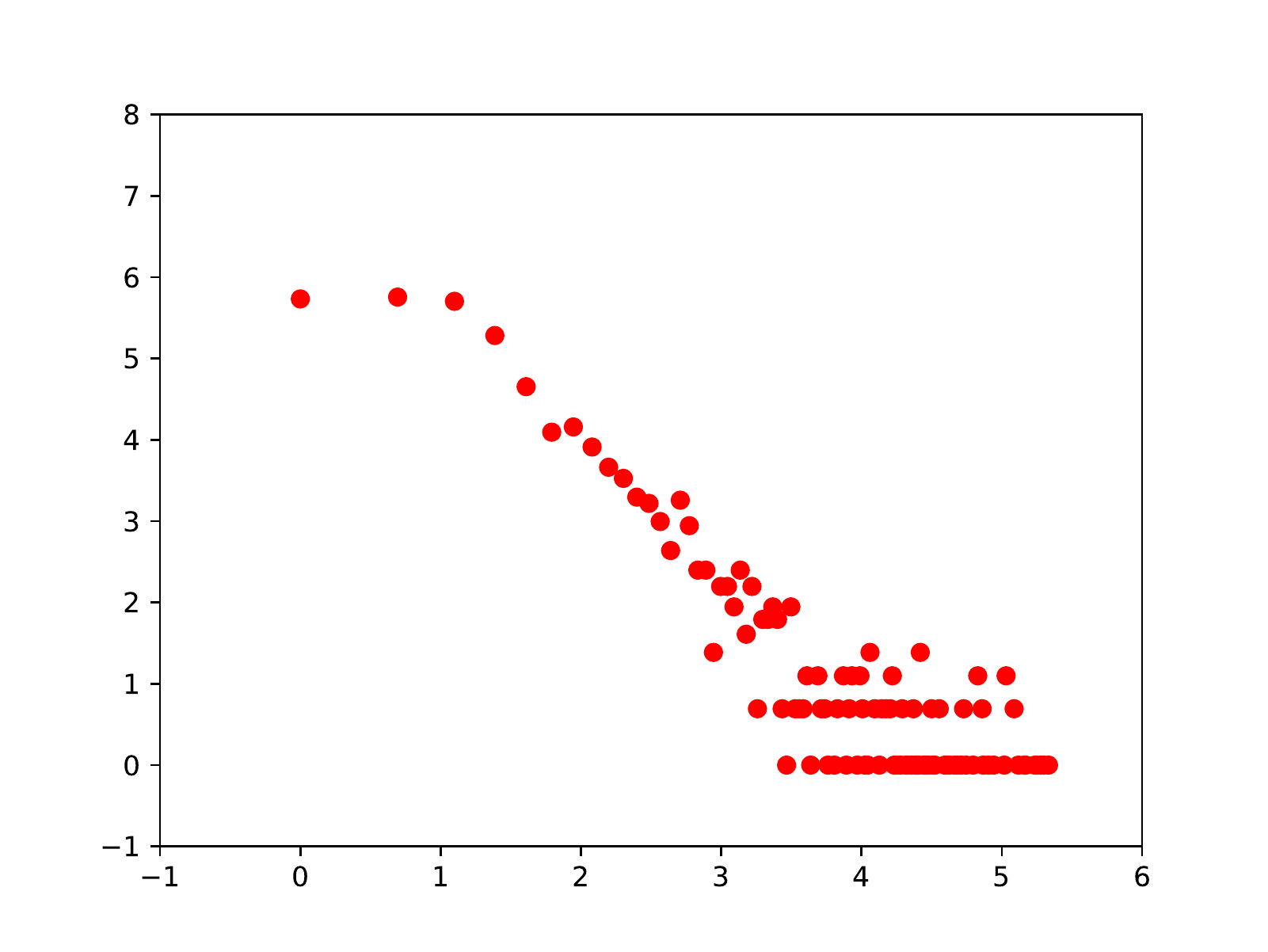}
\end{minipage}%
}%
\subfigure[Group 2]{
\begin{minipage}[t]{0.35\linewidth}
\centering
\includegraphics[width=2.5in]{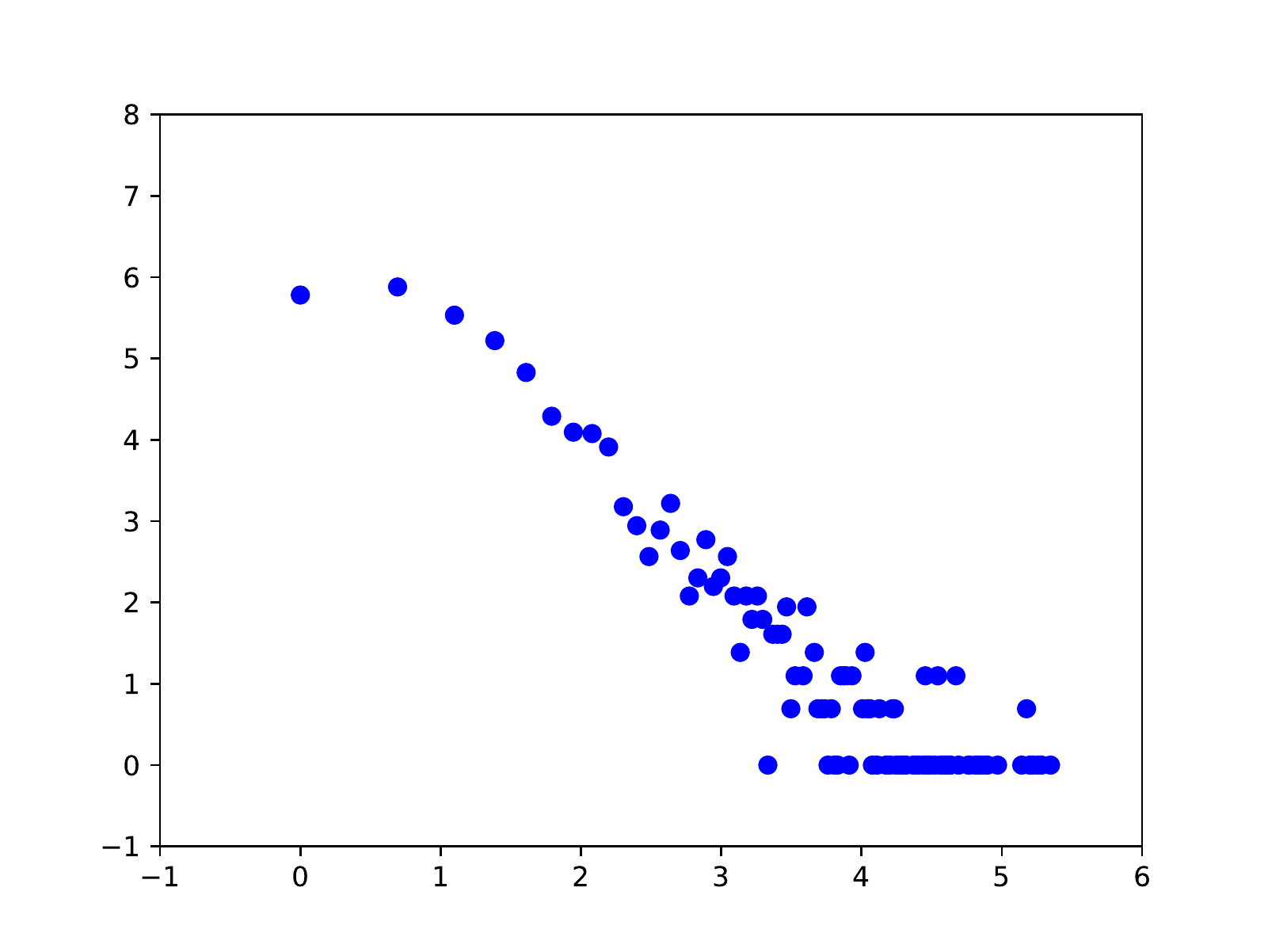}
\end{minipage}%
}%
\subfigure[Group 3]{
\begin{minipage}[t]{0.3\linewidth}
\centering
\includegraphics[width=2.5in]{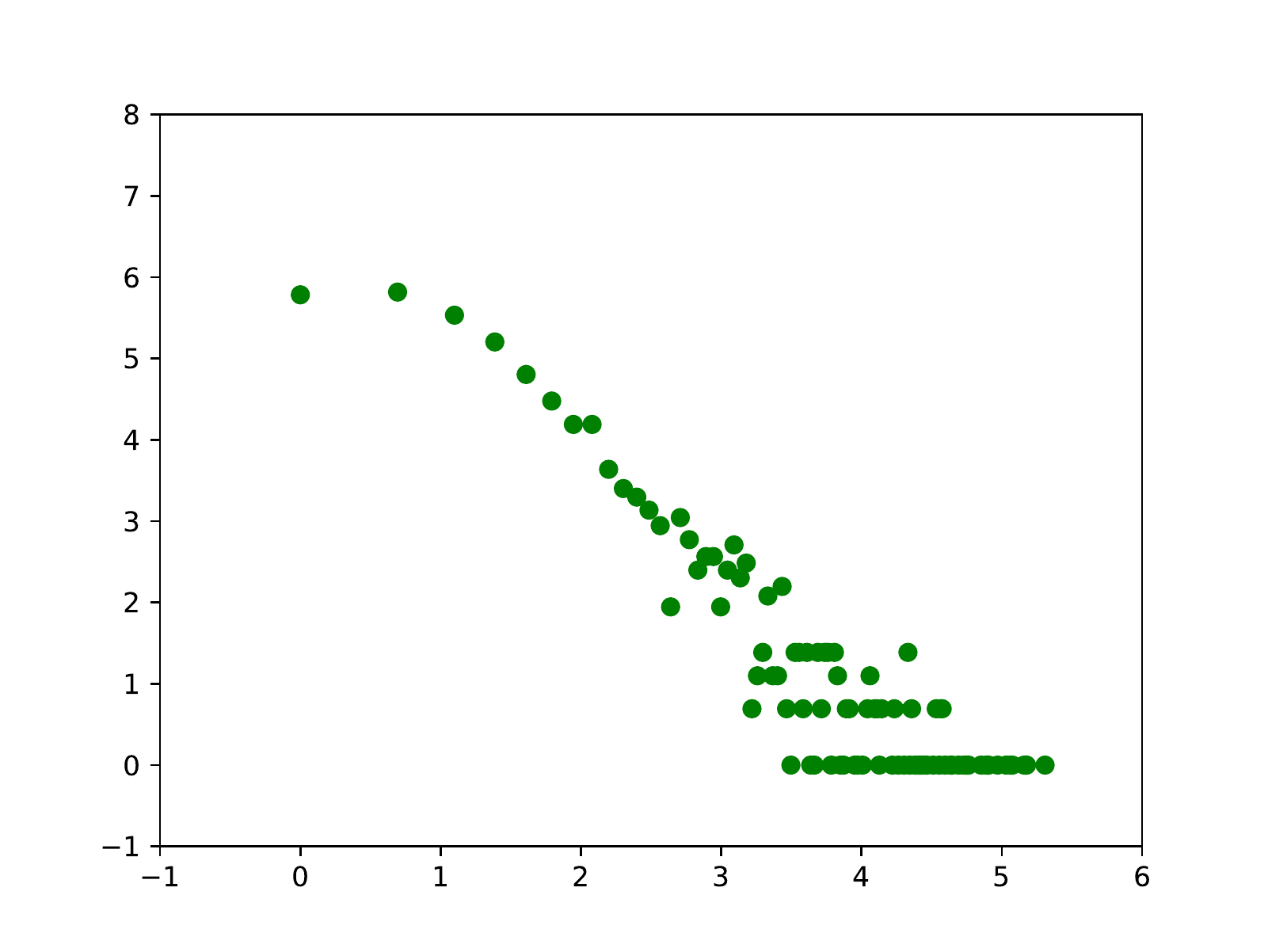}
\end{minipage}%
}%

\centering
\subfigure[Group 4]{
\begin{minipage}[t]{0.5\linewidth}
\centering
\includegraphics[width=2.5in]{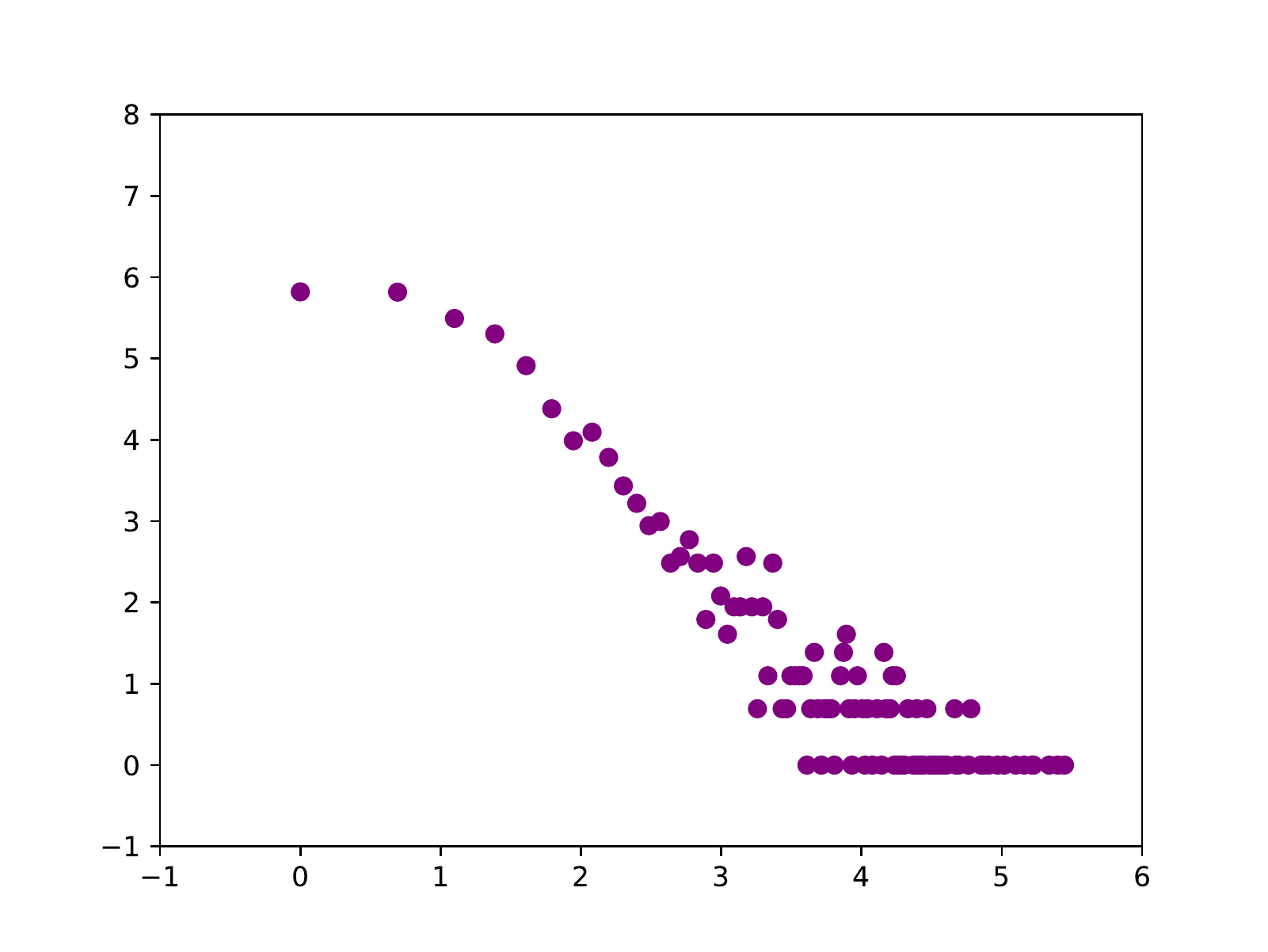}
\end{minipage}%
}%
\subfigure[Group 5]{
\begin{minipage}[t]{0.5\linewidth}
\centering
\includegraphics[width=2.5in]{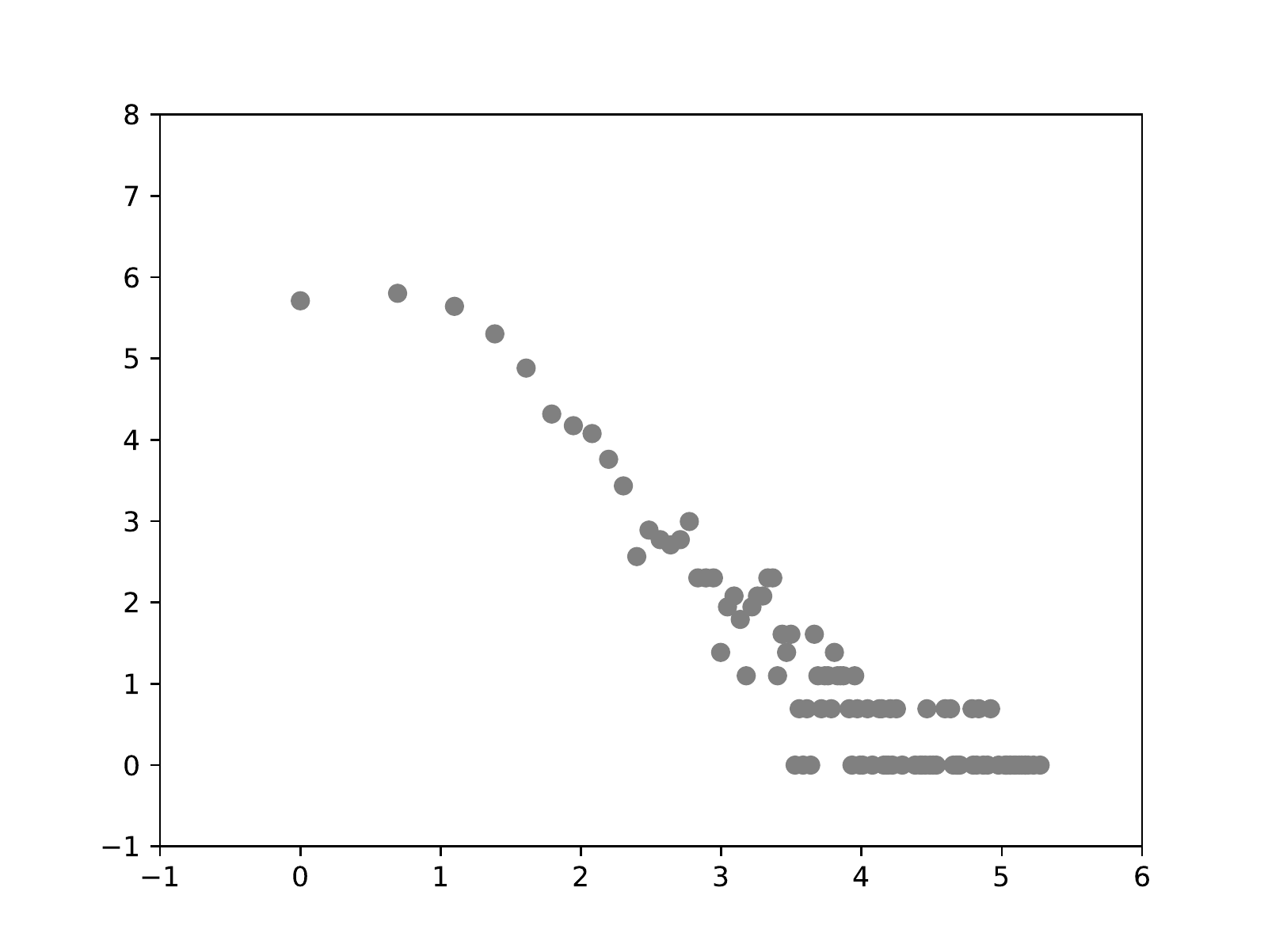}
\end{minipage}%
}%
\centering
\caption{Log-log degree distribution from different groups of CL-10K-1d8-L5}
\label{F8}
\end{figure}


\end{document}